\documentclass[reqno,12pt,letterpaper]{amsart}
\usepackage{xcolor,colortbl}
\usepackage{graphicx}
\usepackage{geometry}
\usepackage{enumitem}
\usepackage{amsmath,amssymb,amsthm,amsfonts}
\usepackage{algorithm}
\usepackage{algorithmic}
\usepackage{dcolumn}
\usepackage{epstopdf}
\usepackage{bm}
\usepackage{appendix}
\usepackage{multirow}
\usepackage{braket}
\usepackage[english]{babel}
\usepackage{hyperref}
\usepackage[capitalize]{cleveref}
\usepackage[square,numbers]{natbib}
\usepackage[T1]{fontenc}
\usepackage{xpatch}
\usepackage{tikz}
\usepackage{adjustbox}
\usepackage{xspace}
\usepackage[roman]{complexity}

\newcommand{\bvec}[1]{\mathbf{#1}}

\newcommand{\va}{\bvec{a}}

\newcommand{\vg}{\bvec{g}}

\newcommand{\vk}{\bvec{k}}

\newcommand{\vq}{\bvec{q}}
\newcommand{\vr}{\bvec{r}}

\newcommand{\vu}{\bvec{u}}
\newcommand{\vv}{\bvec{v}}
\newcommand{\vw}{\bvec{w}}

\newcommand{\vG}{\bvec{G}}

\newcommand{\vzero}{\boldsymbol{0}}

\newcommand{\diag}{\operatorname{diag}}

\renewcommand{\Re}{\operatorname{Re}}
\renewcommand{\Im}{\operatorname{Im}}

\newcommand{\mc}[1]{\mathcal{#1}}

\newcommand{\abs}[1]{\left\lvert#1\right\rvert}

\newcommand{\ud}{\,\mathrm{d}}

\newcommand{\vac}{\mathrm{vac}}

\usepackage{xcolor}
\definecolor{purp}{RGB}{160, 32, 240}

\usetikzlibrary{fit}
\tikzset{%
  highlight/.style={rectangle,rounded corners,fill=blue!15,draw,fill opacity=0.3,thick,inner sep=0pt}
}

\global\long\def\R{\mathbb{R}}

\usepackage{ulem}
\usepackage{amsmath,amssymb,amsthm,graphicx,mathrsfs,url}
\usepackage{amsxtra}
\usepackage{hyperref}
\usepackage{enumitem}
\usepackage{arydshln}
\usepackage{placeins}
\usepackage{wasysym} 
\usepackage{graphicx}
\usepackage{subcaption}

\usepackage{array}
\newcolumntype{P}[1]{>{\centering\arraybackslash}m{#1}}

\Crefname{theo}{Theorem}{Theorems}
\Crefname{lemm}{Lemma}{Lemmas}
\Crefname{assumption}{Assumption}{Assumptions}
\Crefname{corr}{Corollary}{Corollaries}
\Crefname{prop}{Proposition}{Propositions}

\def\?[#1]{\textbf{[#1]}\marginpar{\Large{\textbf{??}}}}

\let\epsilon=\varepsilon

\newcommand{\RR}{{\mathbb{R}}}

\newcommand{\CC}{{\mathbb{C}}}

\newcommand{\ZZ}{{\mathbb{Z}}}
\newcommand{\Z}{{\mathbb{Z}}}

\newtheorem{theo}{Theorem}
\newtheorem{prop}{Proposition}[section] 
\newtheorem{defi}[prop]{Definition}

\newtheorem{assumption}{Assumption}
\newtheorem{lemm}[prop]{Lemma}
\newtheorem{corr}[theo]{Corollary}
\newtheorem{rem}{Remark}
\numberwithin{equation}{section}
\newtheorem{result}[theo]{Result}

\DeclareMathOperator{\Spec}{Spec}

\let\Im=\Imag

\let\Re=\Real

\DeclareMathOperator{\tr}{tr}

\usepackage{scalerel}

\newcommand\reallywidehat[1]{\arraycolsep=0pt\relax%
\begin{array}{c}
\stretchto{
  \scaleto{
    \scalerel*[\widthof{\ensuremath{#1}}]{\kern-.5pt\bigwedge\kern-.5pt}
    {\rule[-\textheight/2]{1ex}{\textheight}} 
  }{\textheight} %
}{0.5ex}\\           
#1\\                 
\rule{-1ex}{0ex}
\end{array}
}

\allowdisplaybreaks[1]

\begin{document}

\title[Interacting electrons in magic angle graphene]{Exact ground state of interacting electrons \\ in magic angle graphene}

\author{Simon Becker}
\email{simon.becker@math.ethz.ch}
\address{ETH Zurich, 
Institute for Mathematical Research, 
Rämistrasse 101, 8092 Zurich, 
Switzerland}

\author{Lin Lin}
\email{linlin@math.berkeley.edu}
\address{Department of Mathematics, University of California, Berkeley, CA 94720, USA; Applied Mathematics and Computational Research Division, Lawrence Berkeley National Laboratory, Berkeley, CA 94720, USA}
\author{Kevin D. Stubbs}
\email{kstubbs@berkeley.edu}
\address{Department of Mathematics, University of California, Berkeley, CA 94720, USA}

\begin{abstract}
One of the most remarkable theoretical findings in magic angle twisted bilayer
graphene (TBG) is the emergence of ferromagnetic Slater determinants as exact
ground states for the interacting Hamiltonian at the chiral limit. This
discovery provides an explanation for the correlated insulating phase which has been experimentally observed at half
filling. This work is the first mathematical study of interacting models in
magic angle graphene systems. These include not only TBG but also TBG-like
systems featuring four flat bands per valley, and twisted trilayer graphene
(TTG) systems with equal twist angles. We identify symmetries of the
Bistritzer-MacDonald Hamiltonian that are responsible for characterizing the
Hartree-Fock ground states as zero energy many-body ground states. Furthermore,
 for a general class of Hamiltonian, we establish criteria  that the ferromagnetic Slater determinants
are the unique ground states within the class of uniformly half-filled, translation invariant Slater determinants. We
then demonstrate that these criteria can be explicitly verified for TBG and
TBG-like systems at the chiral limit, using properties of Jacobi-$\theta$ and
Weierstrass-$\wp$ functions.
\end{abstract}

\maketitle 
\section{Introduction}

Over the past few years, moir{\'e} structures, notably ``magic angle'' twisted bilayer graphene (TBG), have attracted significant attention in the condensed matter physics community.
This surge of interest began after numerous experiments \cite{2018Nature,Serlin,CaoFatemiDemir2018,LuStepanovYang2019,YankowitzChenPolshyn2019} which demonstrate that TBG has nearly flat energy bands and can host intricate phases, such as the correlated insulator (CI) phases at integer fillings and superconducting (SC) phases at non-integer fillings.
Before these exciting experimental results however, the Bistritzer-MacDonald (BM) model \cite{BistritzerMacDonald2011} successfully predicted these nearly flat energy bands near the magic angle of roughly $1.1^{\circ}$.
Yet, the BM model does not directly include electron-electron repulsion (i.e., it is a non-interacting model) and in fact, due to symmetry restrictions, the BM model predicts that twisted bilayer graphene is always in a simple metallic phase.
As a result, the CI and SC phases must arise from electron-electron interactions beyond the non-interacting BM model.

At the chiral limit, we can define a flat-band interacting (FBI) Hamiltonian for TBG that only consists of terms describing electron-electron interactions~\cite{BultinckKhalafLiuEtAl2020,ChatterjeeBultinckZaletel2020,SoejimaParkerBultinckEtAl2020,XieMacDonald2020,WuSarma2020,DasLuHerzog-Arbeitman2021,BernevigSongRegnaultEtAl2021,LiuKhalafLee2021,SaitoGeRademaker2021,JiangLaiWatanabe2019,PotaszXieMacDonald2021,LiuKhalafLeeEtAl2021,FaulstichStubbsZhuEtAl2023}.
One of the most remarkable theoretical results on this model is that, at half filling, the FBI Hamiltonian is \emph{frustration-free}, i.e., the Hamiltonian can be written as a sum of terms (in this case, these terms do not commute), and there exists a ground state that minimizes the energy of each term.
For the FBI Hamiltonian, one possible ground state is a single Slater determinant, which is the simplest type of fermionic many-body wavefunctions.
Moreover, it can be shown that for this ground state, in the thermodynamic limit, a strictly positive amount of energy is needed to either add or remove an electron.
This offers \emph{one} possible explanation for the CI phase in TBG at half-filling.

To our knowledge, mathematical studies so far have focussed on understanding the BM model of twisted bilayer and multilayer graphene at the chiral limit \cite{BeckerEmbreeWittstenEtAl2022, BeckerEmbreeWittstenEtAl2021, becker2023chiral,BeckerHumbertZworski2023,BeckerHumbertZworski2022a, BeckerHumbertZworski2022, Yang2023, WatsonLuskin2021}, and systematic derivations of the BM Hamiltonian from more complex models such as the atomistic tight binding models \cite{WatsonKongMacDonaldEtAl2022,CancesGarrigueGontier2023,cancs2023semiclassical}.
In this paper, we provide the \emph{first} mathematical study of interacting systems, and specifically the FBI Hamiltonian and its ground states. We note that our analysis applies to the FBI Hamiltonian defined on a single valley of the moir{\'e} Brillouin zone.

Our analysis is based on the symmetries of the BM Hamiltonian at the chiral limit. These symmetries are preserved as the number of layers and flat bands varies, which allows us to study TBG as well as TBG-like systems, including twisted multilayer graphene structures.
For example, the standard chiral model for TBG has two layers and two flat bands; a system we refer to as TBG-2.
By changing the interlayer potential, it is possible for two layers of graphene to host four flat bands while still preserving the relevant symmetries \cite{BeckerHumbertZworski2023}; we refer to this system as TBG-4.
We can also consider a twisted trilayer graphene (TTG) system, where the relative twist between the top and middle and the relative twist between the middle and bottom are the same, and can host 
four flat bands \cite{PopovTarnopolsky2023}.  We refer to this system as the equal angle twisted trilayer graphene (eTTG-4).

Although the FBI Hamiltonian is an idealized and simplified representation of
interacting electrons, it still contains many intricate details.
Our paper does not discuss the derivation of the FBI Hamiltonian from the BM model, or the derivation of the BM model from the finer level atomistic models. 
We assume  readers are familiar with second quantization and standard treatments of non-interacting periodic systems.

\subsection{Notation}\label{sec:mainresults}
In this paper, operators and matrices acting on the Fock space are represented using the hat notation, such as $\hat{f}^{\dag}, \hat{f}, \hat{H}_{FBI}$. Operators and matrices that operate in the single particle space (such as $L^2(\RR^2;\CC^2\times \CC^2)$ for TBG) are indicated without the hat notation, as seen in $H, D$. For a single particle operaotr $H$, the Bloch-Floquet transformed Hamiltonian is denoted by $H_\vk$, where $\vk$ is the Brillouin zone index. With some slight abuse of the hat notation, vectors defined in real space are denoted without the hat, for example, $u(\vr)$. Their corresponding Fourier transforms are indicated with the hat notation, as in $\hat{u}(\vq)$.
For any given matrix $A$, the operations of entrywise complex conjugation, transpose, and Hermitian conjugation are represented by $\overline{A}, A^{\top},$ and $A^{\dag}$, respectively.

Due to the symmetries present in twisted graphene systems, the number of flat bands is always an even number which we will denote by $2M$. We denote the set of indices of the flat bands by $\mc{N}=\{ -M, \cdots, -1, 1, \cdots, M\}$. 
After a proper discretization of a single valley of the moir{\'e} Brillouin zone into a discrete set $\mc{K}$ with $N_{\vk}$ points (see the definition of $\mc{K}$ in \eqref{eq:mcK}), we can consider a finite dimensional Hilbert space $\mc{F}$ consisting of $2M N_{\vk}$ fermionic modes.
This space is spanned by $2^{2M N_{\vk}}$ basis vectors of the form
\begin{equation}
\ket{s_{1},\ldots,s_{2MN_{\vk}}}=\prod_{n\in \mc{N}}\prod_{\vk \in \mathcal K}(\hat{f}^{\dag}_{n\vk})^{s}\ket{\vac}, \quad s_j\in\{0,1\},
\end{equation}
where $\hat{f}_{n\vk}^{\dag}, \hat{f}_{n\vk}$ are fermionic creation and annihilation operators satisfying the canonical anticommutation relation (CAR), and the vacuum state $\ket{\vac}=\ket{0,\ldots,0}$ satisfies $\hat{f}_{n\vk}\ket{\vac}=0$ for each $n \in \mc{N}$ and $\vk \in \mc{K}$.

The FBI Hamiltonian takes the form 
\begin{equation}\label{eqn:FBI_abstract1}
  \hat{H} = \sum_{\vq'} \hat{A}_{\vq'}^{\dag} \hat{A}_{\vq'},
\end{equation}
where each $\hat{A}_{\vq'}$ takes the form
\begin{equation}\label{eqn:FBI_abstract2}
\hat{A}_{\vq'}  = \sum_{\vk \in \mc{K}}\sum_{m, n \in\mc{N}} [\alpha_{\vk,\vk+\vq'}]_{m,n} \hat{f}^{\dag}_{m\vk} \hat{f}_{n(\vk + \vq')} + [\beta_{\vk,\vk+\vq'}]_{m,n}.
\end{equation}
Here, the sum over $\vq' \in \R^{2}$ is taken over a discrete set which will be specified later; $\vq'$ represents a momentum difference and so the terms in $\hat{A}_{\vq'}$ connect states at momentum $\vk$ and $\vk + \vq'$.
The coefficient matrices $\alpha_{\vk,\vk+\vq'}$ and $\beta_{\vk,\vk+\vq'}$ are matrices of size $2 M \times 2M$ (defined in \cref{eq:h-fbi}) depend on the eigenstates of the non-interacting BM Hamiltonian at momenta $\vk$ and $\vk + \vq'$ and inherit a number of symmetries from the BM Hamiltonian.

The total number operator $\hat{N}=\sum_{n\in\mc{N}} \sum_{\vk} \hat{f}^{\dag}_{n\vk} \hat{f}_{n\vk}$ counts the number of electrons in a state, i.e.,
if $\hat{N}\ket{\psi}=\nu N_{\vk}\ket{\psi}$, then the number of electrons in $\ket{\psi}$ is $\nu N_{\vk}$. The number operator
$\hat{N}$ commutes with $\hat{H}$ so we can restrict to functions that are simultaneously eigenfunctions of $\hat{H}$ and $\hat{N}$.
The integer filling regime refers to the case when $\nu$ is an integer ($0\le \nu\le 2M$).
One particular integer filling is $\nu=M (= \frac{1}{2} (2M))$, and the corresponding $\ket{\psi}$ is also called a \emph{half-filled} state.
If we further have $\sum_{n\in\mc{N}} \hat{f}^{\dag}_{n\vk} \hat{f}_{n\vk}\ket{\psi}=M\ket{\psi}$ for all $\vk\in \mc{K}$, then $\ket{\psi}$ is a \emph{uniformly half-filled} state.
The one-body reduced density matrix (1-RDM) is defined as $[P(\vk,\vk')]_{nm}=\braket{\psi|\hat{f}^{\dag}_{m\vk'} \hat{f}_{n\vk}|\psi}$. If $[P(\vk,\vk')]_{nm}=[P(\vk)]_{nm} \delta_{\vk,\vk'}$, then $\ket{\psi}$ is a \emph{translation invariant} state.

This paper considers the uniformly half-filled, translation invariant states, which forms a subspace of $\mc{F}$ denoted by $\mc{F}_u$. The ground-state energy in $\mc{F}_u$ is the solution to the following optimization problem
\begin{equation}
E=\min_{\ket{\psi}\in \mc{F}_u} \frac{\braket{\psi|\hat{H}|\psi}}{\braket{\psi|\psi}},
\end{equation}
and its minimizer (which may not be unique) is called a ground state. For any $\ket{\psi}\in\mc{F}$, by construction $\braket{\psi|\hat{H}|\psi}=\sum_{\vq'} \braket{\hat{A}_{\vq'} \psi|\hat{A}_{\vq'} \psi}\ge 0$. Therefore $\hat{H}$ is a positive semidefinite (PSD) Hamiltonian, and $\ket{\psi}\in\mc{F}_u$ satisfying $\hat{H}\ket{\psi}=0$ is a ground state.

We further consider a subset of $\mc{F}_u$  
\begin{equation}\label{eqn:slater_S}
\mc{S}=\left\{\prod_{i=1}^{M} \prod_{\vk\in\mc{K}} \hat{b}_{i\vk}^{\dag}\ket{\vac}\Big\vert~\hat{b}^{\dag}_{i\vk}=\sum_{n\in \mathcal N}\hat{f}^{\dag}_{n\vk}[\Xi(\vk)]_{ni}, \quad \sum_{n \in \mathcal N} [\overline{\Xi(\vk)}]_{ni} [\Xi(\vk)]_{nj}=\delta_{ij}\right\}.
\end{equation}
Each element of $\mc{S}$ is called a uniformly half-filled, translation invariant Slater determinant. The \emph{Hartree-Fock} theory solves a much simpler optimization problem
\begin{equation}\label{eqn:HF_min}
E_{\text{HF}}=\min_{\substack{\ket{\psi}\in \mc{S}}} \frac{\braket{\psi|\hat{H}|\psi}}{\braket{\psi|\psi}}.
\end{equation}
By definition we have $E\le E_{\text{HF}}$ and in general $E< E_{\text{HF}}$.

\subsection{Main results}\label{sec:mainresults}

In this section, we summarize the main findings of this paper,
emphasizing the algebraic structures, while reserving some of the more
technical specifics for later in the paper.

Our first result is that for the FBI Hamiltonian, there exist two states in $\mc{S}$ which are exact ground states. This also implies that the Hartree-Fock theory is exact in this case.

\begin{result}[Informal version of \cref{prop:hf-gs}]
  \label{result:exact-gs}
The single Slater determinants
  \begin{equation}\label{eqn:ferromagnetic}
    \ket{\Psi_{+}}=\prod_{n>0,n\in\mc{N}} \prod_{\vk\in\mc{K}}\hat{f}_{n\vk}^{\dagger} \ket{\vac}, \quad
    \ket{\Psi_{-}}=\prod_{n<0,n\in\mc{N}} \prod_{\vk\in\mc{K}}\hat{f}_{n\vk}^{\dagger},
  \end{equation}
  satisfy
  \begin{equation}
  \hat{A}_{\vq'} \ket{\Psi_{\pm}}=0
  \end{equation}
  for all $\vq'$. Therefore $E=E_{\operatorname{HF}}=0$, and $\ket{\Psi_{\pm}}$ are exact ground states of $\hat{H}$.
\end{result}

The proof of \cref{result:exact-gs} is a generalization of the results in \cite{BernevigSongRegnaultEtAl2021,BultinckKhalafLiuEtAl2020} for TBG. The mechanism of constructing these exact ground states is to verify that $\ket{\Psi_{\pm}}$ are ground states of \emph{each} $\hat{H}_{\vq'}=\hat{A}^{\dag}_{\vq'} \hat{A}_{\vq'}$, even though the terms $\hat{H}_{\vq'}$ do not commute with each other. Therefore the FBI Hamiltonian is an  example of a frustration-free Hamiltonian with nonlocal interactions.

However, there may be other ground states. For instance, any non-vanishing linear combination $a_{+}\ket{\Psi_{+}}+a_{-}\ket{\Psi_{-}}$ is automatically a ground state.
It is in general a difficult task to identify the entire ground state manifold.
Our main result identifies additional assumptions under which $\ket{\Psi_{\pm}}$ are the \emph{unique} states\footnote{With slight abuse of language, we refer to the statement that the ferromagnetic Slater determinant states are the only elements of $\mc{S}$ which are exact ground states of $\hat{H}$ as  ``unique states'', even though technically there are two of them.
} in $\mc{S}$ that are ground states of $\hat{H}$.
\begin{result}[Main result. Informal version of~\cref{thm:hf-gs-unique}]\label{result:main-informal}
  Assume that  the set of matrices $\{ \alpha_{\vk,\vk+\vq'}, \beta_{\vk,\vk+\vq'}\}$ satisfy additional non-degeneracy assumptions (\cref{thm:hf-gs-unique}), then $\ket{\Psi_{\pm}}$ in \cref{eqn:ferromagnetic} are the unique states in $\mc{S}$ that are exact ground states $\hat{H}$.
\end{result}

Both \cref{result:exact-gs} and \cref{result:main-informal} are apply to general FBI Hamiltonians which can be used to describe the interacting electrons in twisted $N$-layer graphene systems. 
The non-degeneracy assumptions can be explicitly verified for specific systems. We verify these conditions for TBG-2, TBG-4, and eTTG-4 systems using the analytic expression for $\{\alpha_{\vk,\vk+\vq'},\beta_{\vk+\vq'}\}$.
This leads to the following results.

\begin{result}[Informal version of \cref{thm:application-tbg-2,,thm:application-tbg-4,,thm:application-ettg-4}]\label{result:unique}
For TBG-2, TBG-4, and eTTG-4, $\ket{\Psi_{\pm}}$ in \cref{eqn:ferromagnetic} are the unique states in $\mc{S}$ that are exact ground states of $\hat{H}$.
\end{result}

While \cref{result:exact-gs} and \cref{result:main-informal} provide a characterization of the ground state, these results alone do not explain why the uniformly half-filled TBG-2 system (and similar systems) are insulators. In~\cref{sec:mb-charge-gap}, we compute the ground state energy of systems with $MN_{\vk}+1$ and  $MN_{\vk}-1$ electrons. Both energies are finite and non-zero in the thermodynamic limit ($N_{\vk}\to \infty$). In other words, it costs a finite amount of energy to add or to remove an electron from the system. Even if the precise nature of the state can be debated, this shows that the system cannot be in a simple metallic state.

\subsection{Discussion and open questions}
The states $\ket{\Psi_{\pm}}$ are sometimes referred to as ``ferromagnetic Slater determinants''. This term is inherited from the physics context where the modes created by $\hat{f}_{n\vk}$ with $n=-M,\ldots, -1$ and with $n=1,\ldots,M$ are vectors of the sublattice symmetry operator $\mc{Z}$ with eigenvalues $+1$ and $-1$, respectively.  These modes can be interpreted as ``pseudo'' up spins and down spins. With this analogy, we can identify $\ket{\Psi_{+}}$ with
$\ket{\uparrow\uparrow\uparrow\cdots}$ and $\ket{\Psi_{-}}$ with $\ket{\downarrow\downarrow\downarrow\cdots}$, which are ground states of a ferromagnet.
The states $\ket{\Psi_{\pm}}$ also have nonvanishing Chern numbers, and are also referred to as (integer) quantum hall states~\cite{BultinckChatterjeeZaletel2020}.

The mechanism for the uniqueness of the ferromagnetic Slater determinant ground states in twisted graphene is rather different from previous uniqueness results in many-body systems (for example, on the Hubbard model on a planar graph \cite{Mielke1991a,Mielke1991,Mielke1992}). 
Unlike these results, the FBI Hamiltonian involves interactions between all momenta and the ground state is not determined by the connectivity of the underlying graph structure.
Instead, the selection of ferromagnetic Slater determinant as the ground states seems to be related to the joint invariant subspaces of the collection of matrices $\{\alpha_{\vk,\vk+\vq'},\beta_{\vk+\vq'}\}$; any state whose one-body reduced density matrix does not lie in an invariant subspace acquires an energy penalty.

These conditions in \cref{result:main-informal} are physically relevant. This paper focuses on the spinless, valleyless (or in physical terms, spin and valley polarized) FBI Hamiltonian. When valley degrees of freedom are taken into account, these assumptions become invalid. This leads to additional ground states in $\mc{S}$, such as the inter-valley coherent (IVC) states~\cite{BultinckKhalafLiuEtAl2020,BernevigSongRegnaultEtAl2021}. The generalization of our result to  FBI Hamiltonians taking into account both spins and valleys will be our future work.

On the other hand, the conditions in \cref{result:main-informal} are difficult to verify. They can be relaxed and explicitly verified for $M=1$ or $2$, as shown in \cref{result:unique}. However, for $M>2$, a generally computationally verifiable method for reformulating these conditions remains unknown.

Even when the conditions are satisfied, \cref{result:main-informal} only establishes uniqueness among uniformly half-filled, translation-invariant Slater determinants.  We conjecture that at half filling, the uniqueness result should extend to all single Slater determinants.
This requires techniques beyond the scope of this initial work and will be our future research.

\subsection{Organization}
Our article is structured as follows:
\begin{itemize}
    \item In Section \ref{sec:FBI} we define the flat-band interacting model for twisted graphene systems.
    \item In Section \ref{sec:revi-twist-bilay} we review the chiral model of twisted graphene sheets.
    \item In Section \ref{sec:requ-symm-gauge}, we outline the symmetries that are relevant for our analysis and fix a gauge of Bloch functions.
\item In Section \ref{sec:many-body-properties}, we review basics on Hartree-Fock theory for FBI Hamiltonians  and prove some of their important many-body properties.
    
    \item In Section \ref{sec:hartree-fock-ground}, we characterize the ground states of the flat band interacting model and state our main result, Theorem \ref{thm:hf-gs-unique}.
    \item In Section \ref{sec:proof-thm-hf-gs}, we give the proof of our main result, Theorem \ref{thm:hf-gs-unique}.
    \item In Section \ref{sec:appl-main-result}, we verify the assumptions of Theorem \ref{thm:hf-gs-unique} for TBG-2, TBG-4, and eTTG-4.
    \item Our article contains three technical appendices, \cref{sec:reform-fock-energy,,sec:proof-full-rank,,sec:real-space-proof} where derive~\cref{eq:fock-energy-proof}, state the proof of \cref{lem:full-rank}, and derive a real space condition for the assumptions of \cref{thm:hf-gs-unique}, respectively.

\end{itemize}

\subsection*{Acknowledgments}
This work was supported by the Simons Targeted Grants in Mathematics and Physical Sciences on Moir\'e Materials Magic (K.D.S.) as well as the SNF Grant PZ00P2 216019 (S.B.). 
This material is based upon work supported by the U.S. Department of Energy, Office of Science, Office of Advanced Scientific Computing Research and Office of Basic Energy Sciences, Scientific Discovery through Advanced Computing (SciDAC) program under Award Number DE-SC0022198 (L.L.). L.L. is a Simons Investigator in Mathematics.  We thank Dumitru Calugaru, Eslam Khalaf, Patrick Ledwith, and Oskar Vafek for their insightful discussions.

\section{The Flat-Band Interacting Hamiltonian for Twisted Graphene}
\label{sec:FBI}
\subsection{Notational Setup for Twisted Graphene}
Before defining the flat-band interacting model for twisted graphene, we introduce a general notation for twisted $N$-layer graphene based on the Bistritzer-MacDonald model.
Twisted bilayer and trilayer graphene correspond to choosing $N = 2$ and $N = 3$ respectively.

After a proper choice of units, the lattice vectors for the moir{\'e} unit cell in the real space are 
\begin{equation}
\label{eq:lattice_vector}
  \vv_1 := -\begin{bmatrix}\frac{2\pi}{\sqrt{3}},\frac{2\pi}{3}\end{bmatrix}^\top \qquad \vv_{2} := \begin{bmatrix}\frac{2\pi}{\sqrt{3}}, - \frac{2\pi}{3}\end{bmatrix}^\top.
\end{equation}
These vectors generate the real space moir{\'e} lattice $\Gamma$:
\begin{equation}
  \Gamma := \vv_{1} \Z + \vv_{2} \Z = \{ a_{1} \vv_{1} + a_{2} \vv_{2} : a_{1}, a_{2} \in \Z  \}.
\end{equation}
The moir{\'e} unit cell in the real space is denoted by $\Omega := \R^{2} / \Gamma$ and can be identified with
\begin{equation}
  \Omega := \left\{ t_1 \vv_1 + t_2 \vv_2 ; t_1,t_2 \in \left[-\frac{1}{2},\frac{1}{2} \right) \right\}.
\end{equation}
The dual lattice (or the reciprocal lattice) of $\Gamma$ is denoted by $\Gamma^{*}$, and is generated by
\begin{equation}
\label{eq:dual_lattice_vector}
  \vg_1 := -\begin{bmatrix}\frac{\sqrt{3}}{2},\frac{3}{2}\end{bmatrix}^\top \qquad  \vg_{2} := \begin{bmatrix}\frac{\sqrt{3}}{2},-  \frac{3}{2}\end{bmatrix}^\top.
\end{equation}
Note that $\vv_{i} \cdot \vg_{j} = 2 \pi \delta_{ij}$. 
The unit cell in the reciprocal space (or the Brillouin zone) is denoted by $\Omega^{*} := \R^{2} / \Gamma^{*}$ and can be identified with
\begin{equation}
  \Omega^* =\left\{ t_1 \vg_1 + t_2 \vg_2 : t_1,t_2 \in \left[-\frac{1}{2},\frac{1}{2}\right)\right\}.
\end{equation}
Throughout the paper, a vector in the reciprocal lattice $\Gamma^*$ is often denoted by $\vG$, while a vector in the Brillouin zone $\Omega^*$ is often denoted by $\vk$ or $\vq$. Note that a generic vector $\vq'\in\RR^2$ can always be uniquely decomposed as $\vq'=\vq+\vG$ for some $\vq\in \Omega^*$ and $\vG\in \Gamma^*$.

In addition to the generating vectors $\vg_{1}$ and $\vg_{2}$, we also identify three special momentum vectors $\vq_{1}, \vq_{2}, \vq_{3}$ which are related to each other by a $\frac{2 \pi}{3}$-counterclockwise rotation $R_{3}$:
\[
  R_{3} := \frac{1}{2}
  \begin{bmatrix}
    -1 & -\sqrt{3} \\
    \sqrt{3} & -1
  \end{bmatrix}.
\]
In particular,
\begin{equation}
  \label{eq:q1q2q3-def}
\vq_{1} := [0, 1]^{\top} \quad \vq_{2} := R_{3} \vq_{1} = [-\sqrt{3}/2, -1/2]^{\top}  \quad \vq_{3} := R_{3} \vq_{2} = [\sqrt{3}/2, -1/2]^{\top}
\end{equation}
Notice that $\vg_{1} = \vq_{2} - \vq_{1}$ and $\vg_{2} = \vq_{3} - \vq_{1}$.

Each moir{\'e} unit cell consists of many atoms located on one of the two sublattices, denoted by $\{ A = 1, B = -1\}$. In the BM model for twisted bilayer and multilayer graphene, all atomistic details in each layer are smeared out except the sublattice information.
The BM Hamiltonian $H$ (called a single particle Hamiltonian) acts on a dense subset of the function space $L^2(\RR^2;\CC^2\times \CC^N)$, where $N$ is the number of layers.
A function in $L^2(\RR^2;\CC^2\times \CC^N)$ can be written as $\psi(\vr;\sigma,\ell)$, where $\vr=[x_1,x_2]^{\top}\in\RR^2$ is the real space coordinate, $\sigma\in\{\pm1\}$ is the sublattice index, and $\ell\in\{1,\ldots,N\}$ is the layer index.
The BM Hamiltonian satisfies certain translation symmetry with respect to the moir{\'e} lattice $\Gamma$. Therefore an eigenfunction of $H$ can be labeled as $\psi_{n\vk}(\vr;\sigma,\ell)$.
Here $n$ is called the band index, $\vk= (k_{1}, k_{2})\in \Omega^*$ is the Bloch vector.
For example, a wavefunction for the Hamiltonian of TBG can be written as
\begin{equation}
   \begin{bmatrix} \psi_{n\vk}(\vr; A, 1) & \psi_{n\vk}(\vr; A, 2) & \psi_{n\vk}(\vr; B, 1) & \psi_{n\vk}(\vr; B, 2) \end{bmatrix}^{\top} \\
\end{equation}
and for TTG:
\begin{equation}
    \begin{bmatrix} \psi_{n\vk}(\vr; A, 1) & \psi_{n\vk}(\vr; A, 2) & \psi_{n\vk}(\vr; A, 3) & \psi_{n\vk}(\vr; B, 1) & \psi_{n\vk}(\vr; B, 2) & \psi_{n\vk}(\vr; B, 3)  \end{bmatrix}^{\top}.
\end{equation}

\subsection{The Flat-Band Interacting (FBI) Hamiltonian}
While we adopt first quantization in expressing the BM Hamiltonian (see \cref{sec:revi-twist-bilay}), the FBI Hamiltonian is most conveniently expressed in second quantization.
The FBI Hamiltonian is defined in terms of the flat-band eigenfunctions of the single particle Hamiltonian $H$; therefore, our first step will be to construct such eigenfunctions.
We make the following mild assumptions on $H$:
\begin{assumption}
  \label{assume:h}
  We assume that $H$ is periodic with respect to the lattice $\Gamma$ and has Bloch eigenpairs $\{ (\epsilon_{n\vk}, \psi_{n\vk}) : n \in \Z\setminus \{0\}, \vk \in \Omega^{*} \}$ so that the eigenvalues $\epsilon_{n\vk}$ are ordered in non-decreasing order and
\begin{equation}
  \begin{split}
  H \psi_{n\vk}(\vr) & = \epsilon_{n\vk} \psi_{n\vk}(\vr), \\
    \psi_{n\vk}(\vr + \va) & = e^{i \va \cdot \vk} \psi_{n\vk}(\vr) \qquad \forall \va \in \Gamma, \\
   \cdots \leq \epsilon_{-1\vk} \leq & ~0 \leq \epsilon_{1\vk} \leq \cdots \qquad \forall \vk \in \Omega^{*}.
  \end{split}
\end{equation}

We furthermore define the Bloch-Floquet transformed Hamiltonian $H_{\vk} := e^{i \vk \cdot \vr} H e^{-i \vk \cdot \vr}$ for all $\vk \in \Omega^*$ and assume that
\begin{enumerate}
  \item The resolvent operator $(H_{\vk} - i)^{-1}$ is compact.
  \item If we identify $\vk \in \Omega^*$ with the complex number $k_{x} + i k_{y} \in \CC$, the operator valued complex map $\vk \mapsto (H_{\vk} - i)^{-1}$ depends analytically on $\vk$. 
  \item There exists a non-empty set of flat bands $\mc{N}$
    \begin{equation}
    \mc{N} := \{ n \in \Z : \forall \vk \in \Omega^{*}, \, \epsilon_{n\vk} = 0 \}.
    \end{equation} 
\end{enumerate}
\end{assumption}

\begin{figure}
  \centering
  \includegraphics[width=7cm]{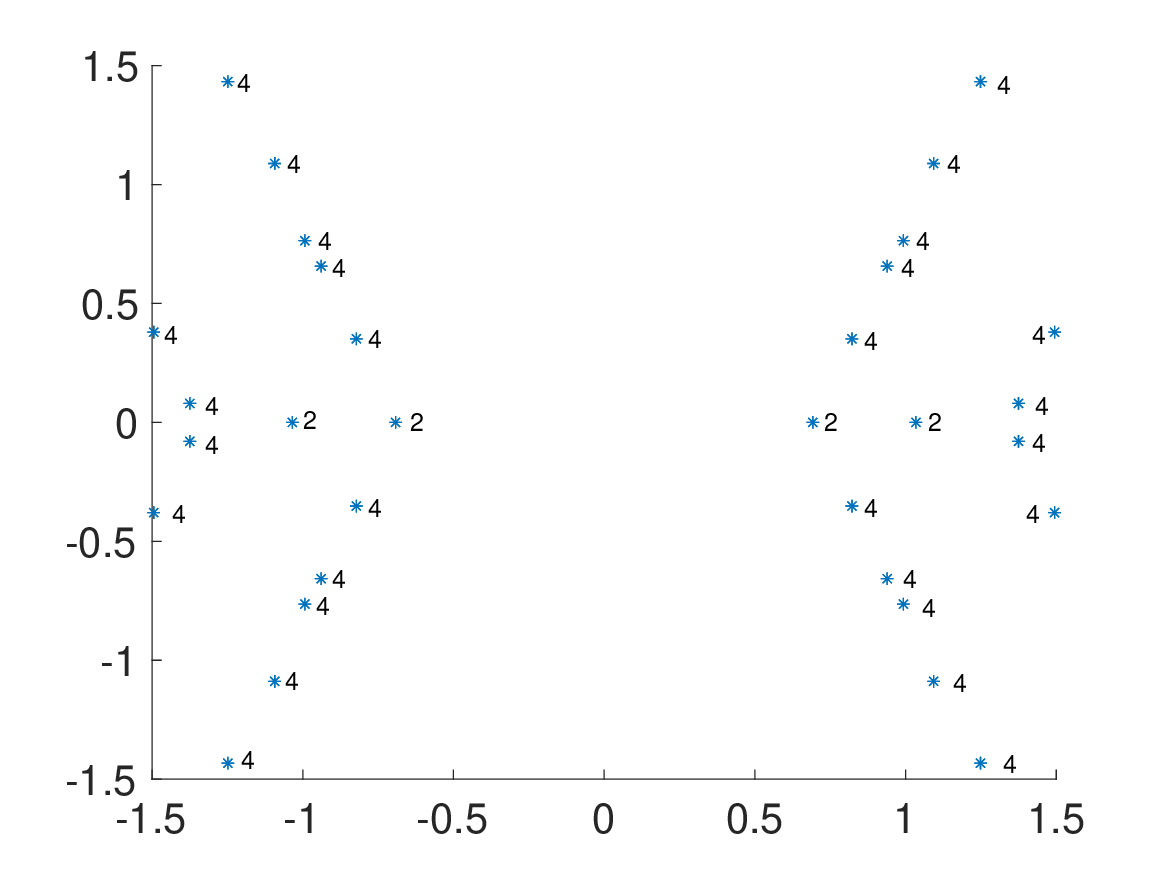}
  \includegraphics[width=7cm]{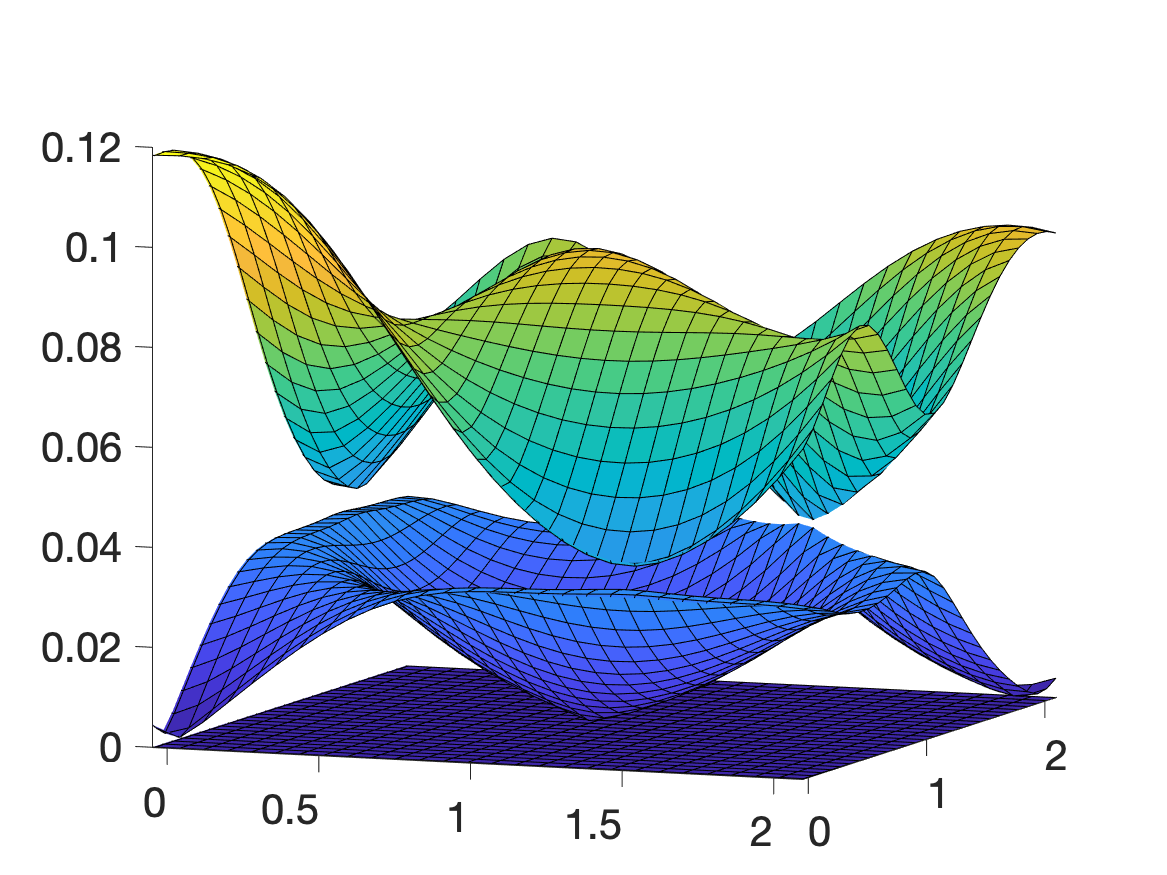}
  \caption{Magic angles with multiplicities in the chiral limit of 7 equally twisted layers of graphene with multiplicities and $U_+=U_0$, see \eqref{eq:ref_pots}. Band configuration for largest magic angle $\alpha = 0.6922$ showing the six lowest bands with positive energy including the four flat ones showing three flat band crossings at high symmetry points.}
  \label{fig:trilayer-magic-angle-bands}
\end{figure}

In some systems, such as twisted trilayer graphene (TTG), there can be energy bands which intersect with the flat bands but are not flat themselves.
This can happen in all $N$-layer systems with $N\ge 3$ (see Figure \ref{fig:trilayer-magic-angle-bands} for an example).
Let us define such the set of momenta where such crossings occur as follows:
\begin{equation}
\label{eq:k-crossing}
\mc{K}_{\rm crossing} = \{ \vk \in \Omega^{*} : \epsilon_{n\vk} = 0, \text{ for some }n \not\in \mc{N} \}.
\end{equation}
Due to the analyticity assumption on $H_{\vk}$, standard results on analytic families of operators \cite[VII.5, Theo 1.10]{Kato} implies that the set of such crossings must have an empty interior.
Furthermore, Rellich's theorem implies that the spectral projectors of $H_{\vk}$ may be chosen analytically throughout the entire Brillioun zone \cite[VII.5, Theo 3.9]{Kato}.
Therefore, we will assume without loss of generality that the flat-band eigenfunctions $\{ \psi_{n\vk} : n \in \mc{N}, \vk \in \mc{K} \}$ have been chosen so that the spectral projection onto these eigenfunctions is analytic in $\vk$.

In order to write down a many-body wavefunction with a finite number of particles, we discretize the Brillouin zone using a uniform grid $\mc{K}$ of size $(n_{k_x}, n_{k_y})$
\begin{equation}
\label{eq:mcK}
\mc{K} := \left\{ \frac{i}{n_{k_{x}}} \vg_1 + \frac{j}{n_{k_{y}}} \vg_2 : i \in \{0, 1, \cdots, n_{k_{x}} - 1\}, j \in \{0, 1, \cdots, n_{k_{y}} - 1 \} \right\} \subseteq \Omega^*.
\end{equation}
Here $\vg_1$, $\vg_2$ are a pair of generating vectors for the moir{\'e} dual lattice $\Gamma^{*}$, see \eqref{eq:dual_lattice_vector}, and we define $N_{\vk} := \# | \mc{K} | = n_{k_{x}} n_{k_{y}}$.

The set $\mc{K}$ is known as a Monkhorst-Pack grid, and the \emph{thermodynamic limit} (TDL) can be reached by taking $n_{k_x},n_{k_y}\to \infty$.
We also make the following technical assumption on the grid $\mc{K}$ in relation to the Hamiltonian $H$:
\begin{assumption}
  \label{assume:grid}
  For each $\vk \in \mc{K}$, let $\Pi(\vk)$ be the orthogonal projector onto the flat-band eigenfunctions at $\vk$.
  We assume that the grid $\mc{K}$ has been chosen so that for all pairs of momenta $\vk, \vk' \in \mc{K}$ there exists a sequence of momenta $\{ \vk_{i} \}_{i=1}^L$ so that $\vk_{1} = \vk$, $\vk_{L} = \vk'$ and $\| \Pi(\vk_{i+1}) - \Pi(\vk_{i}) \| < 1$ for all $i \in \{ 1, \ldots, L \}$.
\end{assumption}
\begin{rem}
  Since $H_{\vk}$ is assumed to be analytic (\cref{assume:h}), the resolvent mapping $\vk \mapsto (H_{\vk}- i)^{-1}$ is Lipschitz in $\vk$.
  Therefore, for \cref{assume:grid} to be true, it is sufficient for $n_{k_{x}} \gtrsim L$ and $n_{k_{y}} \gtrsim L$ where $L$ is the Lipschitz constant for $(H_{\vk} - i)^{-1}$.
\end{rem}

Having fixed the grid $\mc{K}$ and the flat-band eigenfunctions, we may now define the FBI Hamiltonian.
For each flat-band $\psi_{n\vk}(\vr),\vk\in\mc{K}$, we first define the band creation and annihilation operators, $\hat{f}_{n\vk}^\dagger$ and $\hat{f}_{n\vk}$, which create or annihilate a particle in state $\psi_{n\vk}(\vr)$ respectively, which satisfy the canonical anti-commutation relation (CAR):
\begin{equation}
\{\hat{f}^{\dag}_{n\vk},\hat{f}_{n'\vk'}\}
= \delta_{nn'} \delta_{\vk\vk'}, \quad \{\hat{f}^{\dag}_{n\vk},\hat{f}^{\dag}_{n'\vk'}\}
=\{\hat{f}_{n\vk},\hat{f}_{n'\vk'}\}=0.
\end{equation}
The definition of the creation and annihilation operators can be periodically extended outside the Brillouin zone as 
\begin{equation}
  \label{eq:period_f}
\hat f^{\dag}_{n(\vk+\vG)}=\hat f^{\dag}_{n\vk}, \quad \hat f_{n(\vk+\vG)}=\hat f_{n\vk}, \quad \vG\in\Gamma^*.
\end{equation}
Next, we define the flat-band periodic Bloch functions
\begin{equation}
  \label{eq:periodic-bloch}
  u_{n\vk}(\vr) := e^{-i \vk \cdot \vr} \psi_{n\vk}(\vr)
\end{equation}
which are normalized in the unit cell, i.e.,  $\int_{\Omega} |u_{n\vk}(\vr)|^{2} \ud\vr = 1$.
Let $\hat{u}_{n\vk}(\vG)$ denote the Fourier coefficients of $u_{n\vk}(\vr; \sigma, j)$
\begin{equation}
\label{eq:Fourier}
 \hat{u}_{n\vk}(\vG; \sigma, j) := \int_{\Omega} e^{-i \vG \cdot \vr} u_{n\vk}(\vr; \sigma, j)  d\vr.
\end{equation}
and define the \textit{form factor}
\begin{equation}
  \label{eq:form-factor-def}
[\Lambda_{\vk}(\vq + \vG)]_{mn} := \frac{1}{| \Omega |} \sum_{\vG' \in \Gamma^*} \sum_{\sigma,j}\overline{\hat{u}_{m\vk}(\vG'; \sigma, j)} \hat{u}_{n(\vk + \vq +\vG)}(\vG'; \sigma, j).
\end{equation}
Given $\vk,\vq\in \Omega^*,\vG\in\Gamma^*$, $\Lambda_{\vk}(\vq + \vG)$ is a $\# |\mc{N}| \times \# |\mc{N}|$ momentum dependent matrix. Note that $\{\vq+\vG~|~\vq\in \Omega^*,\vG\in\Gamma^*\}=\RR^2$; this notation facilitates later discussions.

The following identities which will be useful at various points in our proof.
\begin{lemm}
  \label{lem:form-factor-identities}
  The form factor matrix $\Lambda_{\vk}(\vq + \vG)$ satisfies the following identities for all $\vk,\vq \in \mc{K}$ and all $\vG, \vG' \in \Gamma^*$
  \begin{align}
   \Lambda_{\vk}(\vq + \vG)^{\dagger} &  = \Lambda_{\vk+\vq}(-\vq -\vG), \label{eq:form-factor-dagger} \\
   \Lambda_{\vk + \vG'}(\vq + \vG) & = \Lambda_{\vk}(\vq + \vG). \label{eq:form-factor-shift}
  \end{align}
\end{lemm}
\begin{proof}
  This follows from the following simple calculation
  \begin{equation}
    \begin{split}
      \Lambda_{\vk}(\vq + \vG)^{\dagger}
      & = \frac{1}{| \Omega |} \sum_{\vG' \in \Gamma^*} \sum_{\sigma,j}\hat{u}_{n\vk}(\vG'; \sigma, j) \overline{\hat{u}_{m(\vk + \vq + \vG)}(\vG'; \sigma, j)} \\
      & = \frac{1}{| \Omega |} \sum_{\vG' \in \Gamma^*} \sum_{\sigma,j}\overline{\hat{u}_{m(\vk + \vq + \vG)}(\vG'; \sigma, j)} \hat{u}_{n\vk}(\vG'; \sigma, j) \\
      & = \Lambda_{\vk+\vq}(-\vq - \vG).
    \end{split}
  \end{equation}
  As for the second identity, one easily checks that $\hat{u}_{n(\vk + \vG')}(\vG) =\hat{u}_{n\vk}(\vG + \vG')$ and so
  \begin{equation}
    \begin{split}
      \Lambda_{\vk + \vG'}(\vq + \vG)
      & = \frac{1}{| \Omega |} \sum_{\vG'' \in \Gamma^*} \sum_{\sigma,j} \overline{\hat{u}_{m(\vk + \vG')}(\vG'' + \vG'; \sigma, j)} \hat{u}_{n(\vk + \vq + \vG + \vG')}(\vG''; \sigma, j) \\
      & = \frac{1}{| \Omega |} \sum_{\vG'' \in \Gamma^*} \sum_{\sigma,j} \overline{\hat{u}_{m\vk}(\vG'' + \vG'; \sigma, j)} \hat{u}_{n(\vk + \vq + \vG)}(\vG'' + \vG'; \sigma, j) \\
      & = \Lambda_{\vk}(\vq + \vG).
    \end{split}
  \end{equation}
\end{proof}

We can now define the flat-band interacting (FBI) Hamiltonian for the twisted $N$-layer graphene  \cite{BultinckKhalafLiuEtAl2020,BernevigSongRegnaultEtAl2021}.
When we compare the definition of $\hat{H}_{FBI}$ with that in \cref{eqn:FBI_abstract1,eqn:FBI_abstract2}, we need to show that $\hat{\rho}^{\dag}(\vq')=\hat{\rho}(-\vq')$, and that $\hat{H}_{FBI}$ is a positive semi-definite Hamiltonian. These relations will be verified in \cref{sec:mb-posit-semid}.

\begin{defi}
  Let  $\hat{f}_{m\vk}^{\dagger}, \hat{f}_{m\vk}$ denote the flat-band creation and annihilation operators, and let $\Lambda_{\vk}(\vq + \vG)$ be defined as in \cref{eq:form-factor-def}.
  The flat-band interacting (FBI) Hamiltonian is:
  \begin{equation}
    \label{eq:h-fbi}
    \begin{split}
      \hat{H}_{FBI} & := \frac{1}{N_{\vk} |\Omega|}  \sum_{\vq' \in \mc{K} + \Gamma^*} \hat{V}(\vq') \widehat{\rho}(\vq') \widehat{\rho}(-\vq') \\
      \widehat{\rho}(\vq') & := \sum_{\vk \in \mc{K}} \sum_{m,n \in \mc{N}} [\Lambda_{\vk}(\vq')]_{mn} \left(\hat{f}_{m\vk}^\dagger \hat{f}_{n(\vk+\vq')} - \frac{1}{2} \delta_{mn} \sum_{\vG}\delta_{\vq',\vG}  \right).
    \end{split}
  \end{equation}
  Here, $\hat{V}(\vq') = \hat{V}(|\vq'|)$ is the Fourier transform of a radially symmetric electron-electron potential which satisfies $\hat{V}(\vq') > 0$ for all $\vq \in \R^2$. 
\end{defi}

For instance, for TBG, the double gate-screened Coulomb interaction in Fourier space reads (Note that $\lim_{\vq' \to \vzero} \hat{V}(\vq') = \frac{\pi d}{\varepsilon}$ is well defined)
\begin{equation}
\hat{V}(\vq')=\frac{2\pi}{\epsilon}\frac{\tanh(\abs{\vq'}d/2)}{\abs{\vq'}}.
\end{equation}
Here $\epsilon,d>0$ parametrize the strength and length of the screened Coulomb interaction, respectively (see e.g., \cite[Appendix C]{BernevigSongRegnaultEtAl2021}).

\section{A Review of Twisted Bilayer Graphene and equal Twist Angle Trilayer Graphene}
\label{sec:revi-twist-bilay}

To make the assumptions of \cref{result:main-informal} more concrete, we review the properties of the TBG and eTTG \cite{TarnopolskyKruchkovVishwanath2019,BeckerEmbreeWittstenEtAl2022,PopovTarnopolsky2023,becker2023chiral,WatsonKongMacDonaldEtAl2022,CancesGarrigueGontier2023} in the Bistritzer-MacDonald Hamiltonian at the chiral limit.
The single particle Hamiltonian for both chiral TBG and chiral eTTG depend on a parameter $\alpha$, that is inversely proportional to the twisting angle, and take the form of a matrix-valued differential operator
\begin{equation}
  \label{eq:chiral-ham}
H(\alpha)  = \begin{bmatrix} 0 & D(\alpha)^\dag \\ D(\alpha) & 0 \end{bmatrix},
\end{equation}
For TBG, $D(\alpha)$ acts on $H^{1}(\R^{2}; \CC^{2})$
\begin{equation}
\label{eq:d-tbg}
     D_{\rm TBG}(\alpha)  = \begin{bmatrix} D_{x_1}+i D_{x_2} & \alpha U_+(\vr) \\ \alpha U_-( \vr) & D_{x_1}+i D_{x_2}  \end{bmatrix},
\end{equation}
and for eTTG, $D(\alpha)$ acts on $H^{1}(\R^{2}; \CC^{3})$
\begin{equation}
\label{eq:d-ettg}
D_{\rm eTTG}(\alpha) = \begin{bmatrix} D_{x_1}+i D_{x_2} & \alpha U_+(\vr) & 0 \\ \alpha U_-( \vr) & D_{x_1}+i D_{x_2}  &  \alpha U_+(\vr) \\ 0 &\alpha U_-( \vr) & D_{x_1}+i D_{x_2} \end{bmatrix}.
\end{equation}
Similarly, single particle Hamiltonians for $N$-layer systems can be defined in terms of operators $D(\alpha)$ acting on $H^{1}(\R^{2}; \CC^{N})$ (see, for example, \cite{Yang2023}).

We remark that the specific form of the single particle Hamiltonian will not be used in our proofs of \cref{result:exact-gs,,result:unique,,result:main-informal}, however we will make use of two important properties which hold for Bistritzer-MacDonald-type Hamiltonians at the chiral limit: (1) the existence of a magic angle with exactly flat bands, and (2) the single particle Hamiltonian commutes with a number of symmetry operations (see~\cref{sec:requ-symm-gauge}).

For small twist angles, due to the lattice structure of graphene, the tunneling potentials $U_{\pm}(\vr)$ satisfy the following symmetry properties independent of the number of layers \cite{BistritzerMacDonald2011,TarnopolskyKruchkovVishwanath2019}:
\begin{align}
  U_{\pm}(\vr + \va) & = \overline{\omega}^{\pm (a_1+a_2)} U_{\pm}(\vr), \label{eq:u-translation} \\
  U_{\pm}(R_{3} \vr) & = \omega U_{\pm}(\vr), \label{eq:u-rotation} \\
  U_{\pm}(x_{1}, -x_{2}) & = \overline{U_{\pm}(x_{1}, x_{2})}, \label{eq:u-mirror}
\end{align}
where $\vr = [x_{1}, x_{2}]^{\top}$, $\va = [a_{1}, a_{2}]^{\top} \in \Gamma$, $\omega := e^{2 \pi i / 3}$, and $R_{3}$ is the $\frac{2 \pi}{3}$-counterclockwise rotation.
Potentials $U_{\pm}$ satisfying symmetries \cref{eq:u-translation,eq:u-rotation,eq:u-mirror} are of the general form
\begin{equation}
\label{eq:pot}
U_{\pm}(\vr)=\sum_{n,m\in\Z} c_{nm} e^{\pm i (m \vg_{1} - n \vq_{2} + \vq_{1}) \cdot \vr}
\end{equation}
where we recall the definitions of the dual lattice vectors $\vg_1,\vg_2,$ (\cref{eq:dual_lattice_vector}) and $\vq_{1}$ (\cref{eq:q1q2q3-def}).
The coefficients $c_{nm}$  satisfy, see \cite[Prop.2.1]{becker2023chiral},  
\[\begin{split}
c_{nm}&=\omega c_{(m-n-1)(-n)}=\omega^{2} c_{(-m)(n-m+1)}\text{ and }\\
\overline{c_{nm}}&=c_{(-m)(-n)}=\omega c_{(-n+m-1)m}=\omega^{2} c_{n(n-m+1)}
\text{ for all }n, m \in \Z^2.\end{split}
\]

It is important to note that due to~\cref{eq:u-translation}, both TBG and eTTG as defined above are not periodic with respect to $\Gamma$, however they can both be made into periodic Hamiltonian by performing a unitary transformation (see~\cref{sec:periodicity-tbg-ettg}).
While this non-periodic formulation is more convenient for the analysis of the magic angles; we will use the periodic formulation in our definition of the FBI Hamiltonians (\cref{eq:h-fbi}).

\subsection{Magic Angles in TBG}
\label{sec:magic-angles-tbg}
The TBG Hamiltonian (\cref{eq:chiral-ham} with \cref{eq:d-tbg}) commutes with the following symmetries
\begin{equation}
\label{eq:symmetries2}
  \mc{T}_{\va} \vu(\vr) =
                          \begin{bmatrix}
                            \omega^{a_{1} + a_{2}} &&& \\
                                                   & 1 && \\
                                                   && \omega^{a_{1} + a_{2}} & \\
                                                   &&& 1
                          \end{bmatrix} \vu(\vr+\va), \qquad
  \mc{R} \vu (\vr) = \begin{bmatrix}
                            1 &&& \\
                                                   & 1 && \\
                                                   && \overline{\omega} & \\
                                                   &&& \overline{\omega}
                          \end{bmatrix} \vu(R_{3} \vr).
\end{equation}
Therefore, we can define
\begin{equation}
  L^2_{\ell,p}:= \{ \vu \in L^2_{\text{loc}}(\RR^2;\CC^4); \mc{T}_{\va} \vu  = \omega^{\ell(a_1+a_2)} \vu, \text{ for } \va \in \Gamma \text{ and } \mc{R} \vu  = \overline{\omega}^p \vu\}
\end{equation}
where $\ell,p \in \ZZ_3.$
We may also consider symmetries \eqref{eq:symmetries2} just acting on $\CC^2$-valued spinors by considering merely the first two components.

The $\mc{T}_{\va}$-periodicity of the TBG Hamiltonian in \cref{eq:chiral-ham} allows us to define the Bloch-Floquet transformed TBG Hamiltonian 
\begin{equation}
  \label{eq:Hk}
  H_{\vk}(\alpha)  = \begin{bmatrix} 0 & D_{\rm TBG}(\alpha)^\dag +( k_1-ik_2)  \operatorname{id}_{\mathbb C^2} \\ D_{\rm TBG}(\alpha) + (k_1+ik_2 )\operatorname{id}_{\mathbb C^2}& 0 \end{bmatrix}
\end{equation}
acting on $L^2_{\ell} :=\bigoplus_{p \in \ZZ_3} L^2_{\ell,p}$.

The TBG Hamiltonian additionally has a layer symmetry $\mc{L}$ which acts as follows 
\begin{equation}
\label{eq:Lsymm}
  \mc{L} \vv(\vr):=
  \begin{bmatrix}
    &&& -1 \\
    && 1 & \\
    & -1 && \\
    1 &&& \\
  \end{bmatrix}
  \vv(-\vr).
\end{equation}
One furthermore can check that $\mc{L} : L^{2}_{\ell,p} \to L^{2}_{-\ell+1,p}$.
In particular, for $\ell=2$ the map $\mc{L}$ leaves spaces $L^2_{2,p}$ invariant.
Thus, we shall without loss of generality, consider $H_{\mathbf k}(\alpha)$ on $L^2_{\ell=2}$ in the sequel, since the family of $H_{\mathbf k}(\alpha)$ on $L^2_{\ell}$ are all equivalent. 

One then defines a magic angle of the chiral Hamiltonian as a parameter $\alpha \in \CC$ such that
\begin{equation}
0\in \bigcap_{\vk \in \R^2} \Spec(H_{\vk}(\alpha)),
\end{equation}
where $H_{\vk}(\alpha)$ is defined as in \eqref{eq:Hk}.

For the Hamiltonian \eqref{eq:Hk} the first $2$ components correspond to lattice sites $A$ and the remaining ones to lattice sites $B$. Since $D(\alpha)+(k_1+ik_2)\operatorname{id}_{\mathbb C^2}$ is a Fredholm operator of index $0$, at a magic angle $\alpha$, the null space of the Hamiltonian decomposes into
\begin{equation}
  \ker(H_{\vk}(\alpha)) = \ker(D(\alpha)+(k_1+ik_2)\operatorname{id}_{\mathbb C^2}) \oplus \ker(D(\alpha)^*+(k_1-ik_2)\operatorname{id}_{\mathbb C^2})
\end{equation}
where $\ker(D(\alpha)+(k_1+ik_2)\operatorname{id}_{\mathbb C^2})$ and $\ker(D(\alpha)^*+(k_1-ik_2)\operatorname{id}_{\mathbb C^2})$ are of equal dimension. Elements of $\ker(D(\alpha)+(k_1+ik_2)\operatorname{id}_{\mathbb C^2})$ therefore correspond to states at zero energy that are $A$-lattice polarized. This leads us to make the following definition
\begin{defi}[Multiplicity]
    A magic angle $\alpha$ is $n$\text{-fold degenerate} if the number of zero energy flat bands of the Hamiltonian is $2n$\text{-fold degenerate.} A $1$-fold degenerate magic angle is also called \emph{simple}. 
\end{defi}

In the case of TBG and eTTG we indicate the number $2M$ of flat bands at a magic angle by appending  a number $2M$, e.g. TBG$-2$ for simple magic angles and TBG$-4$ for $2$-fold degenerates ones.

The simplest one-parameter family of potentials satisfying \eqref{eq:pot} is, for $\varphi \in \mathbb R/\mathbb Z$,  given by
\[ U_{\varphi}(\vr) = \cos(2\pi \varphi)  \sum_{i=0}^2 \omega^i e^{-i \vq_i\cdot \vr}  + \sin(2\pi \varphi) \sum_{i=0}^2\omega^i e^{2i \vq_i \cdot \vr}.\]
In twisted bilayer graphene, higher Fourier modes in the tunnelling potential are necessary to describe lattice relaxation effects, see e.g. \cite[Remark $2.4$]{WatsonKongMacDonaldEtAl2022} and references therein.
We then consider special cases
\begin{equation}
\label{eq:ref_pots}
U_0(\vr) = \sum_{i=0}^2 \omega^i e^{-i \vq_i\cdot \vr} \text{ and }U_{7/8}(\vr) =\frac{1}{\sqrt{2}}\Bigg( U_0(\vr) - \sum_{i=0}^2\omega^i e^{2i \vq_i \cdot \vr}\Bigg).
\end{equation}

While both $U_0$ and $U_{7/8}$ satisfy the above symmetries   \eqref{eq:u-translation}, \eqref{eq:u-rotation}, and \eqref{eq:u-mirror}, they give rise to a different number of flat bands for \emph{magic} $\alpha \in \mathbb R$. Numerical experiments suggest, see Figure \ref{fig:many}, that in case of $U=U_0$ and $\alpha \in \RR$ magic, precisely two bands of the chiral Hamiltonian become flat, in the case of $U=U_{7/8}$ and $\alpha \in \mathbb R$ the Hamiltonian exhibits four flat bands at \emph{magic angles}. Thus, the chiral model with the two potentials $U=U_0$ and $U=U_{7/8}$ serve as models for TBG$-2$ and TBG$-4$. While simple and two-fold degenerate magic angles are the only types of magic angles that appear for generic choices of tunnelling potentials in chiral limit TBG, see \cite[Theo. $3$]{BeckerHumbertZworski2023}, they exhibit an almost equidistant spacing for the potentials $U_0$ and $U_{7/8}$, as shown in Table \ref{table:mag_ang}.

\begin{figure}
\includegraphics[width=7.5cm]{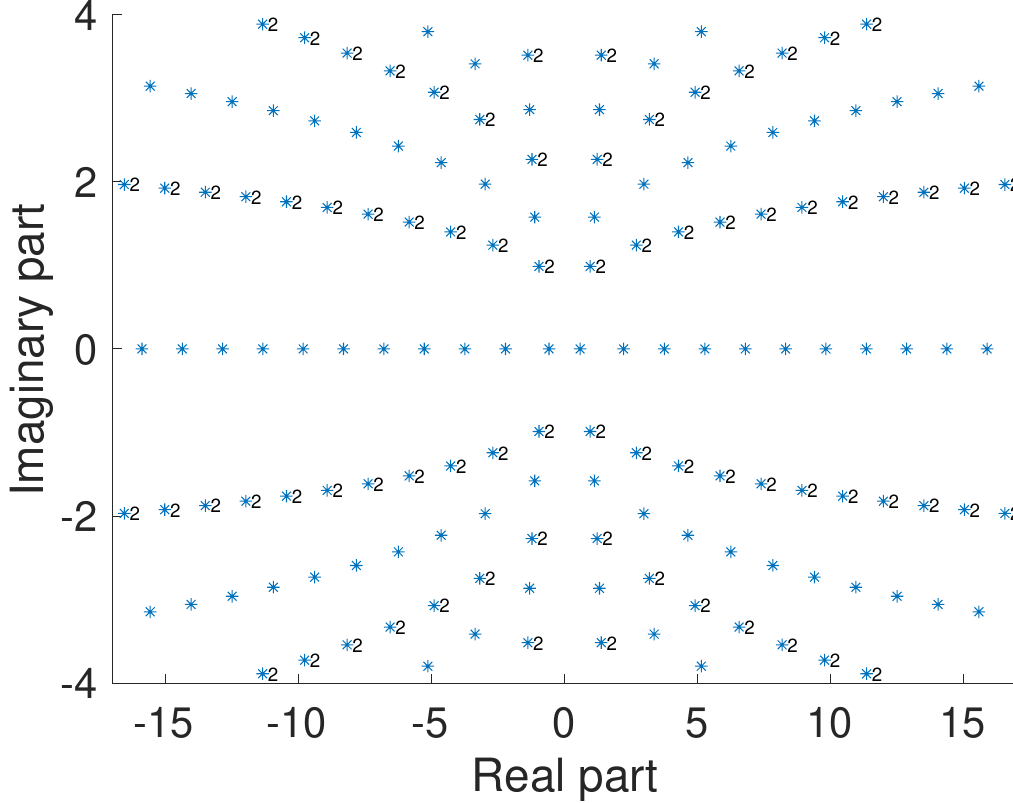}
\includegraphics[width=7.5cm]{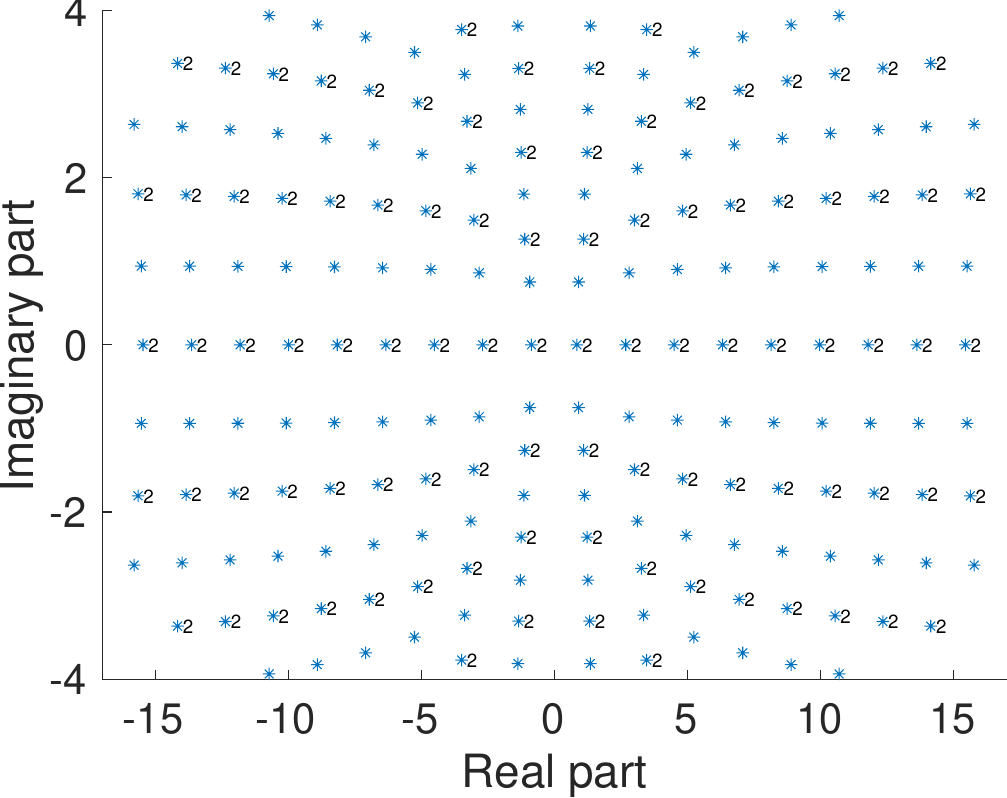}
\caption{\label{fig:many} Magic angles $\alpha$ derived from potentials $U=U_0$ (left) and $U = U_{7/8} $ (right). The dimension of $\operatorname{ker}(D(\alpha)+(k_1+ik_2)\operatorname{id}_{\mathbb C^2})$ is indicated by the numbers in the figure. No number indicates a one-dimensional subspace. The number of flat bands of the Hamiltonian is twice the dimension of $\operatorname{ker}(D(\alpha)+(k_1+ik_2)\operatorname{id}_{\mathbb C^2}).$ }
\end{figure}

\begin{table}[!htb]
\begin{minipage}{.45\linewidth}

\begin{center}
\begin{tabular}{rclcc}
\multicolumn{1}{c}{$k$} & &
\multicolumn{1}{c}{$\alpha_k$} &
 & $\alpha_{k}-\alpha_{k-1}$ \\[2pt] \hline
1  &\ &   \phantom{0}0.58566 &\ &               \\
2  &&     \phantom{0}2.22118  && 1.6355        \\
3  &&     \phantom{0}3.75140 && 1.5302        \\
4  &&     \phantom{0}5.27649   && 1.5251        \\
5  &&     \phantom{0}6.79478   && 1.5183        \\
6  &&     \phantom{0}8.31299    && 1.5182        \\
7  &&     \phantom{0}9.82906    && 1.5161        \\
8  &&    11.34534     && 1.5163        \\
9  &&    12.86061        && 1.5153        \\
10 &&    14.37607         && 1.5155        \\
11 &&    15.89096          && 1.5149        \\
\end{tabular}
\end{center}

\end{minipage}
\begin{minipage}{.45\linewidth}

\begin{center}
\begin{tabular}{rclcc}
\multicolumn{1}{c}{$k$} & &
\multicolumn{1}{c}{$\alpha_k$} &
 & $\alpha_{k}-\alpha_{k-1}$ \\[2pt] \hline
1  &\ &   \phantom{0}0.853799 &\ &               \\
2  &&     \phantom{0}2.691433  && 1.8376       \\
3  &&     \phantom{0}4.507960  && 1.8165        \\
4  &&     \phantom{0}6.332311  && 1.8244        \\
5  &&     \phantom{0}8.157130    && 1.8248        \\
6  &&     \phantom{0}9.983510     && 1.8264       \\
7  &&     11.809376      && 1.8259        \\
8  &&     13.635446      && 1.8261        \\
9  &&     15.460894        && 1.8255        \\
10 &&    17.286231        && 1.8253        \\
11 &&    19.111041 && 1.8248
\end{tabular}
\end{center}

\end{minipage}
\caption{First $11$ real magic angles, rounded to 6 digits, for $U=U_0$ (left) and $U=U_{7/8}$ (right).\label{table:mag_ang}}
\end{table}

\subsection{Magic Angles in eTTG}
\label{sec:magic-angles-ettg}
When studying twisted trilayer graphene (TTG), the case of equal twisting angles (eTTG) is of particular interest.

In a recent paper by Popov and Tarnopolsky \cite{PopovTarnopolsky2023}, the authors demonstrated that the collection of magic parameters $\mathcal A_{\mathrm{eTTG}}$ exhibits a particularly interesting relation, which connects the set of magic angles of chiral limit TBG to the magic angles of equal twisting angle TTG (eTTG)
\begin{equation}
  \label{eq:eTTG}
  \sqrt 2\mathcal A_{\mathrm{TBG}-M}= \mathcal A_{\mathrm{eTTG}-2M},
\end{equation}
see also the top right figure in Figure \ref{fig:TTG}.
 Here $\mathcal A_{\mathrm{eTTG}-2M}$ is the set of magic parameters for eTTG-2M. Moreover, the multiplicity on the right is at least twice the one on the left. The argument provided in  \cite{PopovTarnopolsky2023} constructs protected flat bands for eTTG-2M from flat bands of TBG-M. This construction works for eTTG-2M, however it does not extend to more than three layers nor does not seem to extend beyond the equal twisting angle case, as well.

At the center of this connection, is the map with obvious indices corresponding to $D(\alpha) = D_{\text{TBG}}(\alpha)$ as in \eqref{eq:d-tbg} and $D(\alpha) = D_{\text{eTTG}}(\alpha)$ as in \eqref{eq:d-ettg}
\begin{equation}\label{eqn:ttgmapping}
  \begin{split}
    \times :  \ker_{L^2_{k,j}}( D_{\text{TBG}}(\alpha_0))&\oplus  \ker_{L^2_{k',j'}}(D_{\text{TBG}}(\alpha_0))\to  \ker_{ L^2_{k+k',j+j'}}(D_{\text{eTTG}}(\sqrt 2\alpha_0)) \\
    (v,w) &\mapsto (v_1w_1,2^{-1/2}(v_1w_2 + v_2 w_1), v_2w_2).
  \end{split}
\end{equation}
This construction allows us to construct from any two elements of the nullspaces of $D_{\text{TBG}}$, for some magic $\alpha_0$, with respect to different representations, an element of the nullspace of $D_{\text{eTTG}},$ but for a rescaled magic parameter $\sqrt{2}\alpha_0.$ We observe that the number of zeros doubles by applying the map above. This means that the number of flat bands for eTTG, with $\sqrt{2}\alpha_0$ is (at least) twice the number of flat bands for TBG with $\alpha_0.$

\begin{figure}
 \includegraphics[width=7.5cm]{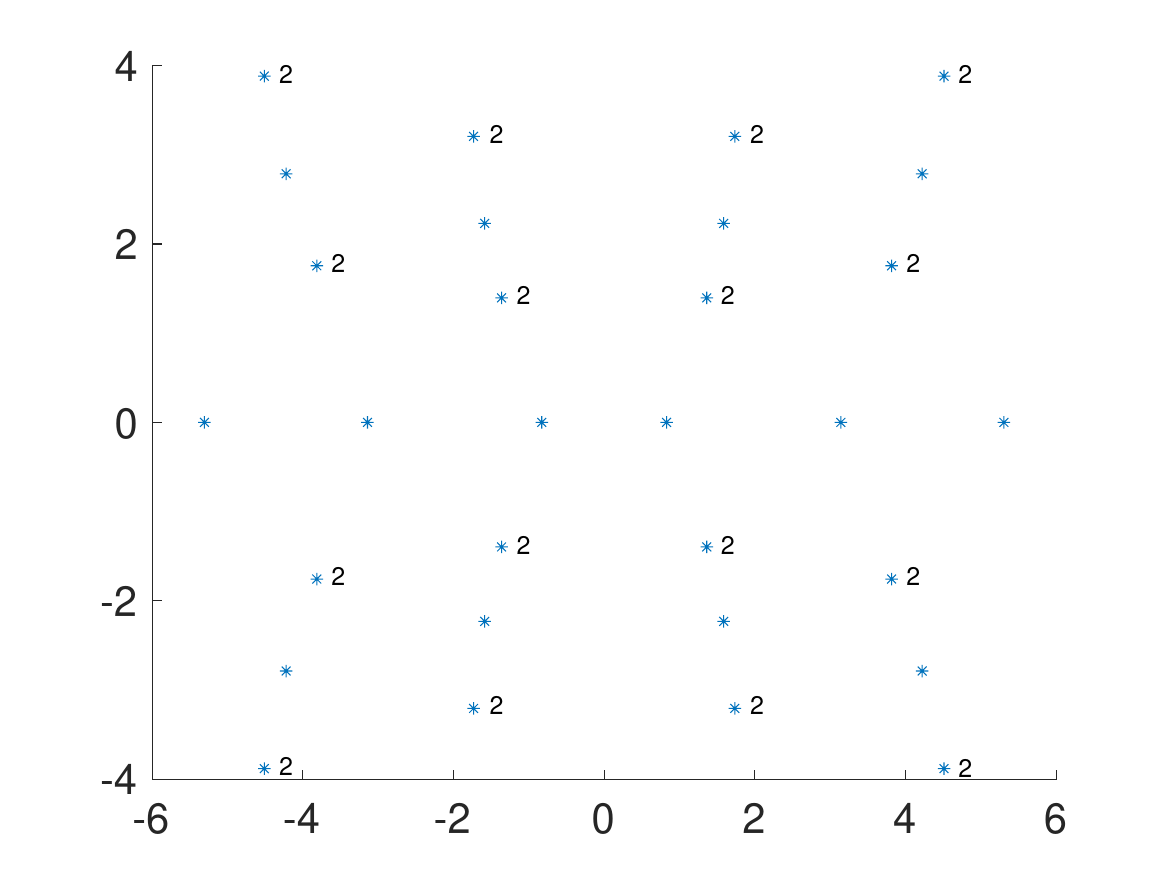}
  \includegraphics[width=7.5cm]{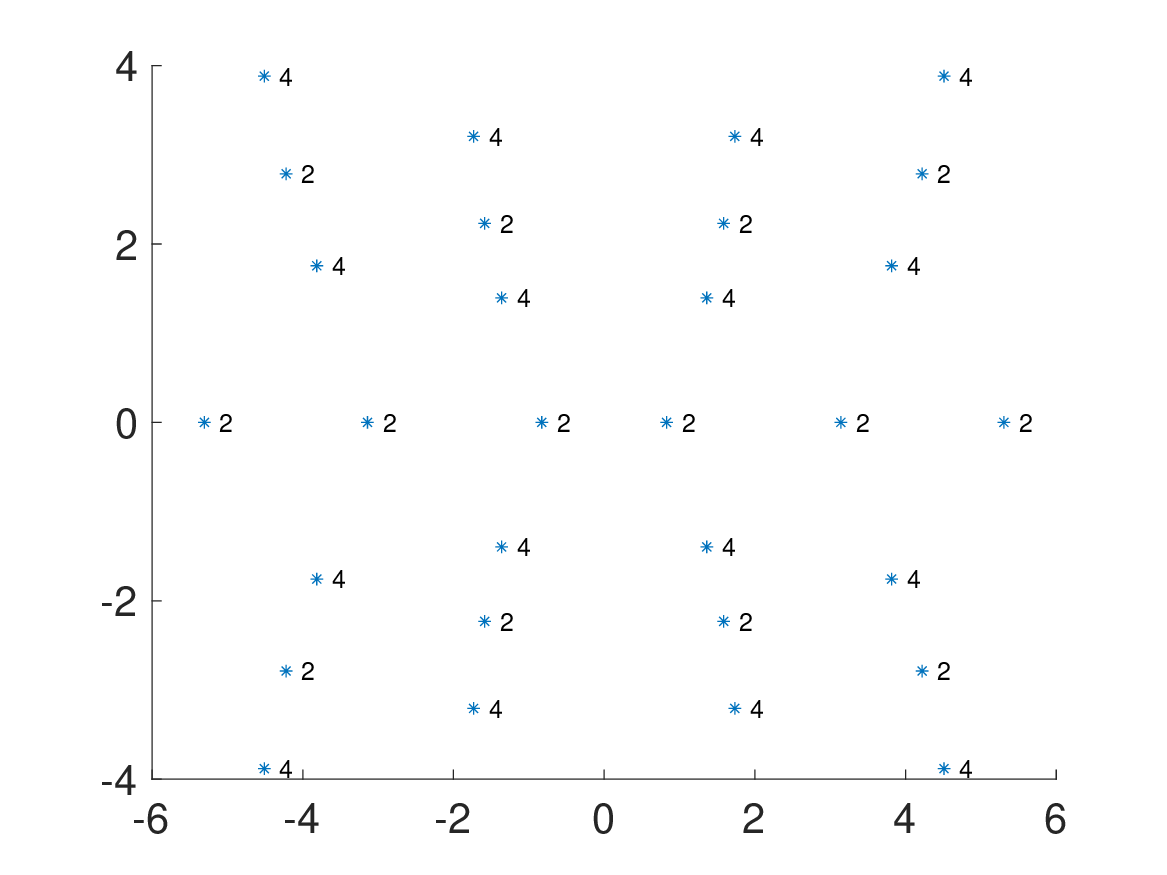}
  \caption{Magic angles of TBG multiplied by $\sqrt{2}$, $\sqrt{2}\mathcal A_{\text{TBG}-2}$ (left) and eTTG $\mathcal A_{\text{eTTG}-4}$ (right) with $U_+=U_0$, see \eqref{eq:ref_pots}. The multiplicity of the flat bands of the Hamiltonian is twice the multiplicity indicated in the figures. }
  \label{fig:TTG}
\end{figure}
\FloatBarrier

\subsection{Periodicity of TBG and eTTG}
\label{sec:periodicity-tbg-ettg}
As discussed previously, the Hamiltonians for TBG and eTTG as defined in the beginning of this section are not periodic with respect to $\Gamma$.
However, by using the general form of the potentials $U_{\pm}(\vr)$ (\cref{eq:pot}), we see that we can turn TBG into a $\Gamma$-periodic Hamiltonian by conjugating it by $V(\vr):=\operatorname{diag}(1,e^{i \vq_1 \cdot \vr},1,e^{i \vq_1 \cdot \vr})$ to obtain the equivalent periodic Hamiltonian
\[
  \mathscr H(\alpha) = \begin{bmatrix} 0 & \mathscr  D(\alpha)^{\dagger} \\ \mathscr  D(\alpha) & 0\end{bmatrix} \text{ with }\mathscr D(\alpha) = \begin{bmatrix}D_{x_1}+i D_{x_2}   & \alpha e^{iq_0 \cdot \vr} U_+(\vr) \\ \alpha e^{-iq_0 \cdot \vr} U_-(\vr) &D_{x_1}+i D_{x_2} + i  \end{bmatrix}.
\]
Similarly, eTTG can be made periodic by conjugating by
\[
V(\vr) = \operatorname{diag}(1,e^{i \vq_1 \cdot \vr}, e^{-i \vq_{1} \cdot \vr},1,e^{i \vq_1 \cdot \vr}, e^{-i \vq_{1} \cdot \vr})
\]
and using the fact that $3 \vq_{1} = (0, 3)^{\top} \in \Gamma^{*}$.

\section{Symmetries and Gauge Fixing}
\label{sec:requ-symm-gauge}

We now discuss the symmetry assumptions required for our main theorem which mirror the symmetries present in TBG and eTTG.
Due to the form of the chiral Hamiltonian~\cref{eq:chiral-ham} and the symmetries of the $U_{\pm}(\vr)$ (\cref{eq:u-translation,eq:u-rotation,eq:u-mirror}), the Hamiltonians for TBG and eTTG satisfy a number of symmetries in addition to the translation $\mc{T}_{\va}$ and rotation $\mc{R}$.
In particular, they satisfy a sublattice symmetry $\mc{Z}$, a layer symmetry $\mc{L}$, and a composite symmetry $\mc{Q}$, whose actions are defined as follows:
\begin{align}
  \mc{Z} \psi(\vr; \sigma, j) & = \sigma \psi(\vr; \sigma,j) \label{eq:z-symm-def} \\
  \mc{L} \psi(\vr; \sigma, j) & = (-1)^{j} \psi(-\vr; -\sigma, N-j). \label{eq:l-symm-def} \\
    \mc{Q} \psi(\vr; \sigma, j) & = \overline{\psi(-\vr; -\sigma,j)} \label{eq:q-symm-def} 
\end{align}
In words, the sublattice symmetry multiplies the $A$ sublattice by $+1$ and states supported on the $B$ sublattice by $-1$.
The layer symmetry $\mc{L}$, reverses the order of the layers, swaps the $A$ and $B$ sublattices, multiplies by an alternating minus sign, and maps $\vr \mapsto -\vr$.
The symmetry $\mc{Q}$ is a composition of a $C_{2z}$ rotation and time reversal $\mc{T}$ and acts by swapping the $A$ and $B$ sublattices, taking a complex conjugate, and mapping $\vr \mapsto -\vr$. The composite symmetry $\mc{Q}$ is often referred to as the $C_{2z}\mc{T}$ symmetry. 
It is easily checked that $\{ \mc{Z}, \mc{Q} \} = 0$, $[ \mc{Z}, \mc{L} ] = 0$, and $[ \mc{Q}, \mc{L} ] = 0$.

We will use these symmetries to fix a specific choice of Bloch eigenbasis for the flat bands; a procedure we refer to as ``gauge fixing'' following the physics terminology.

\begin{assumption}[Symmetry Assumptions]
  \label{assume:symm}
  We assume that the single particle Hamiltonian $H$ commutes with $\mc{Q}$ and anticommutes with $\mc{Z}$ and $\mc{L}$.
\end{assumption}
\begin{rem}
  Since we study the nullspace of $H$, we can consider both symmetries that commute or anticommute with $H$ as both types of symmetries fix the nullspace.
\end{rem}
If $\mc{Z}$, $\mc{Q}$, and $\mc{L}$ are symmetries of a periodic Hamiltonian, it is easily checked that $\mc{Z}$ and $\mc{Q}$ map $\vk$ to $\vk$ and $\mc{L}$ maps $\vk$ to $-\vk$.

Since the zero energy eigenstates are degenerate within the flat band, once we
fix a basis of orthogonal flat-band eigenfunctions we may perform any $U(2M)$
transformation to this basis and get an alternative choice of eigenbasis which
spans the same space. Any such transformation takes the form of the mapping 
\[
\psi_{n\vk}(\vr) \mapsto \sum_{m \in \mc{N}} \psi_{m\vk}(\vr) [U(\vk)]_{mn}
\]
where $U(\vk)$ is a unitary matrix. Under this basis transformation, the creation and annihilation operators likewise transform as
\[
\hat{f}^{\dagger}_{n\vk} \mapsto \sum_{m \in \mc{N}} \hat{f}^{\dagger}_{m\vk} [U(\vk)]_{mn} \qquad \qquad \hat{f}_{n\vk} \mapsto \sum_{m \in \mc{N}} \hat{f}_{m\vk}  [U(\vk)^{*}]_{mn}
\]
and similarly the form factor transforms as
\[
\Lambda_{\vk}(\vq + \vG) \mapsto U(\vk)^{\dagger} \Lambda_{\vk}(\vq + \vG) U(\vk + \vq).
\]
By fixing a specific choice of eigenbasis, which we refer to as ``fixing a gauge'' following physics terminology, we can force the form factor to satisfy certain properties which will be a key part of our characterization of the Hartree-Fock ground state.
Note that the FBI Hamiltonian only depends on the form factor $\Lambda_{\vk}(\vq + \vG)$ which in turn only depends on the periodic Bloch functions, $u_{n\vk}(\vr; \sigma, j)$.
As such we will define our gauge fixing in terms of the periodic Bloch functions, $u_{n\vk}(\vr; \sigma, j)$, instead of the Bloch functions.
\begin{rem}
  When the Hamiltonian $H$ exhibits band crossings at zero energy, the set $\{ u_{n\vk}(\vr) : n \in \mc{N}, \vk \in \Omega^{*} \}$ is still closed under the action of symmetry operations for $\vk \in \mc{K}_{\rm crossing}$ since the basis at the crossing points is defined through continuity.
\end{rem}

Since the zero energy eigenstates $\psi_{n\vk}(\vr)$ are indexed by $n \in \mc{N}$ and $\mc{Z}$ is a symmetry of the Hamiltonian, we may change the gauge and relabel the eigenfunctions so that:
\begin{equation}
  \label{eq:sublattice-symm}
  \begin{split}
    \mc{Z} u_{n\vk}(\vr; \sigma, j) & = (+1) u_{n\vk}(\vr; \sigma, j) \qquad n > 0, \\
    \mc{Z} u_{n\vk}(\vr; \sigma, j) & = (-1) u_{n\vk}(\vr; \sigma, j) \qquad n < 0.
  \end{split}
\end{equation}
In particular, this implies that $u_{n\vk}(\vr)$ is completely supported on the $A$ sublattice if $n > 0$ and completely supported on the $B$ sublattice if $n < 0$.
We can now properly define the uniformly half-filled, translation invariant, ferromagnetic Slater determinant states:
\begin{defi}[Ferromagnetic Slater Determinants]
  \label{def:ferro-sd}
  Suppose that the Bloch eigenbasis has been chosen so that it satisfies \cref{eq:sublattice-symm}. We define two uniformly half-filled, translation invariant ferromagnetic Slater determinants, or ferromagnetic Slater determinants for short, to be many-body states of the form:
  \[
    \ket{\Psi_{+}} = \prod_{\vk \in \mc{K}} \prod_{n > 0, n\in\mc{N}} \hat{f}_{n\vk}^{\dagger} \ket{\vac} \qquad \ket{\Psi_{-}} = \prod_{\vk \in \mc{K}} \prod_{n < 0, n\in\mc{N}} \hat{f}_{n\vk}^{\dagger} \ket{\vac}.
  \]
  That is, $\ket{\Psi_{\pm}}$ fully fills one of the two eigenspaces of $\mc{Z}$.
\end{defi}

Since $\mc{Q}$ is  a symmetry of the Hamiltonian and relates the two different sublattices, after fixing the Bloch eigenfunctions to satisfy~\cref{eq:sublattice-symm}, we may fix the span $\{ u_{n\vk}(\vr) : n > 0\}$, and determine $n < 0$ by the sublattice symmetry $\mc{Q}$. That is, for all $n > 0$ we define
\begin{equation}
  \label{eq:composite-symm}
  u_{(-n)\vk}(\vr; \sigma, j) := \mc{Q} u_{n\vk}(\vr; \sigma, j) = \overline{u_{n\vk}(-\vr; -\sigma, j)}.
\end{equation}

Additionally since $\mc{L}$ is also a symmetry of the Hamiltonian and relates $\vk$ and $-\vk$, after fixing the periodic Bloch functions to satisfy~\cref{eq:sublattice-symm,eq:composite-symm} we may fix the space $\{ u_{n\vk}(\vr) : k_{1} \geq 0 \}$ and determine $k_{1} < 0$ by the layer symmetry $\mc{L}$. That is,
\begin{equation}
  \label{eq:layer-symm}
  u_{n(-k_{1},k_{2})}(\vr; \sigma, j) := \mc{L} u_{n(k_{1},-k_{2})}(\vr; \sigma, j) = (-1)^{j} u_{n(k_{1}, -k_{2})}(-\vr; \sigma, N-j).
\end{equation}

\begin{lemm}
  \label{lem:form-factor-properties}
  Suppose that the gauge has been chosen so that periodic Bloch functions $\{ u_{n\vk}(\vr) : n \in \mc{N}, \vk \in \Omega^{*} \}$ satisfy \cref{eq:sublattice-symm,eq:composite-symm,eq:layer-symm}. Then the form factor $\Lambda_{\vk}(\vq + \vG)$ can be written as
  \begin{equation}
    \label{eq:form-factor-a}
    \Lambda_{\vk}(\vq + \vG) =
    \begin{bmatrix}
      A_{\vk}(\vq + \vG) & \\
                         & \overline{A_{\vk}(\vq + \vG)}
    \end{bmatrix}
  \end{equation}
  where
  \begin{equation}
    \label{eq:form-factor-block}
    [A_{\vk}(\vq + \vG)]_{mn} = \frac{1}{| \Omega |} \sum_{\vG'} \sum_{\sigma,j} \overline{\hat{u}_{m,\vk}(\vG; \sigma, j)} \hat{u}_{n(\vk+\vq)}(\vG + \vG'; \sigma, j) \qquad m, n > 0.
  \end{equation}
  Let $\mc{K}$ be as in~\cref{eq:mcK}, then $A_{\vk}(\vq + \vG)$ additionally satisfies the following \emph{sum rule}: 
  \begin{equation}
    \label{eq:form-factor-sum-rule}
    \sum_{\vk \in \mc{K}} \Im{\tr{(A_{\vk}(\vG) )}} = 0, \quad \vG \in \Gamma^{*}.
  \end{equation}
\end{lemm}
\begin{proof}
  Recall the definition of the form factor
  \begin{equation}
    [\Lambda_{\vk}(\vq + \vG)]_{mn} := \frac{1}{| \Omega |} \sum_{\vG' \in \Gamma^*} \sum_{\sigma,j}\overline{\hat{u}_{m\vk}(\vG'; \sigma, j)} \hat{u}_{n(\vk + \vq + \vG)}(\vG'; \sigma, j).
  \end{equation}
  Since the definition of the form factor involves a sum over the sublattice $\sigma$ and $u_{n\vk}(\vr)$ have disjoint sublattice support for $n > 0$ and $n < 0$ we immediately see that $[\Lambda_{\vk}(\vq + \vG)]_{mn} = 0$ if $m n < 0$ and hence $\Lambda_{\vk}(\vq + \vG)$ can be written as a block diagonal matrix.

  Due to $\mc{Q}$ symmetry (\cref{eq:composite-symm}) we have that
  \begin{equation}
    \begin{split}
      \hat{u}_{(-n)\vk}(\vG; \sigma, j)
      & = \int_{\Omega} e^{-i \vG \cdot \vr} u_{(-n)\vk}(\vr; \sigma, j) \ud\vr \\
      & = \int_{\Omega} e^{-i \vG \cdot \vr} \overline{u_{n\vk}(-\vr; \sigma, j)} \ud\vr \\
      & = \overline{\int_{\Omega} e^{-i \vG \cdot \vr} u_{n\vk}(\vr; \sigma, j) \ud\vr} \\
      & = \overline{\hat{u}_{n\vk}(\vG; \sigma, j)}.
    \end{split}
  \end{equation}
  From this relation we easily see that
  \begin{equation}
    [\Lambda_{\vk}(\vq + \vG)]_{(-m)(-n)} = [\overline{\Lambda_{\vk}(\vq + \vG)}]_{mn}
  \end{equation}
  which together with the block diagonal structure implies \cref{eq:form-factor-a}. Now we turn to prove the sum rule \cref{eq:form-factor-sum-rule}.

  Due to the layer symmetry (\cref{eq:layer-symm}) we have that
  \begin{equation}
    \label{eq:form-factor-layer}
    \begin{split}
      \hat{u}_{n(-\vk)}(\vG; \sigma, j)
      & = \int_{\Omega} e^{-i \vG \cdot \vr} (-1)^{j} u_{n\vk}(-\vr; -\sigma, N - j) \ud\vr \\
      & = (-1)^{j} \hat{u}_{n\vk}(-\vG; -\sigma, N - j).
    \end{split}
  \end{equation}
  Recall that $\Gamma^{*}$ is closed under the map $\vG \mapsto -\vG$ and therefore
  \begin{equation}
    \label{eq:form-factor-sum-rule-1}
    \begin{split}
      \Lambda_{\vk}(\vq + \vG)
      & = \frac{1}{| \Omega |} \sum_{\vG' \in \Gamma^*} \sum_{\sigma,j}\overline{\hat{u}_{m\vk}(\vG'; \sigma, j)} \hat{u}_{n(\vk + \vq + \vG)}(\vG'; \sigma, j) \\
      & = \frac{1}{| \Omega |} \sum_{\vG' \in \Gamma^*} \sum_{\sigma,j}\overline{\hat{u}_{m(-\vk)}(-\vG'; -\sigma, N - j)} \hat{u}_{n(-\vk - \vq - \vG)}(-\vG'; -\sigma, N - j) \\
      & = \Lambda_{-\vk}(-\vq - \vG)
    \end{split}
  \end{equation}
  which for $\vq = \vzero$ reduces to $\Lambda_{\vk}(\vG) = \Lambda_{-\vk}(-\vG) = \Lambda_{-\vk}(\vG)^{\dagger}$ where the last equality is due to~\cref{eq:form-factor-dagger}.
  Note that this immediately implies that $\Lambda_{\vzero}(\vG)$ is Hermitian.

  We recall that elements of $\mc{K}$ can be written as $\vk =\frac{i}{n_{k_{x}}} \vg_{1} + \frac{j}{n_{k_{y}}} \vg_{2}$ for some $i \in [n_{k_{x}}]$ and $j \in [n_{k_{y}}]$.
  For any $\vk$ of this form, we can find a $\vk' \in \mc{K}$ and $\vG' \in \Gamma^{*}$ so that $-\vk = \vk' + \vG'$.
  Therefore, we may partition the momentum grid $\mc{K}$ into three disjoint sets $\{ \vzero \}$, $\mc{K}_{1}$  and $\mc{K}_{2}$.
  The point $\vzero$ is the unique point so that $\vk = -\vk$ and $\mc{K}_{1}$ and $\mc{K}_{2}$ are defined so that for each $\vk \in \mc{K}_{1}$ there exists a $\vk' \in \mc{K}_{2}$ so that $\vk - \vk' \in \Gamma^{*}$.
  Since $A_{\vk}(\vG) = A_{\vk + \vG'}(\vG)$ for all $\vG' \in \Gamma^{*}$ (\cref{eq:form-factor-shift}), we have
  \begin{equation}
    \begin{split}
      \sum_{\vk \in \mc{K}} A_{\vk}(\vG)
      & = A_{\vzero}(\vG) + \sum_{\vk \in \mc{K}_{1}} A_{\vk}(\vG) + \sum_{\vk \in \mc{K}_{2}} A_{\vk}(\vG) \\
      & = A_{\vzero}(\vG) + \sum_{\vk \in \mc{K}_{1}} A_{\vk}(\vG) + \sum_{\vk \in \mc{K}_{1}} A_{-\vk}(\vG) \\
      & = A_{\vzero}(\vG) + \sum_{\vk \in \mc{K}_{1}} (A_{\vk}(\vG) + A_{\vk}(\vG))^{\dagger}) \\
    \end{split}
  \end{equation}
  and hence $\sum_{\vk \in \mc{K}} \Im{ \tr{( A_{\vk}(\vG))}} = 0$. Here we have used $A_{\vzero}(\vG)$ is Hermitian.
\end{proof}

\section{Properties of the Flat-Band Interacting Hamiltonian}
\label{sec:many-body-properties}
In this section, we start by reviewing the basics on Hartree-Fock theory for FBI Hamiltonians (\cref{sec:review-hf-theory}) and then prove some of their important many-body properties.
In particular, we prove that the FBI Hamiltonian (\cref{eq:h-fbi}) is positive semidefinite (\cref{sec:mb-posit-semid}) and therefore any state which has zero energy must be a ground state.
Using this fact, along with the sum rule (\cref{eq:form-factor-sum-rule}), we prove that the ferromagnetic Slater determiant state are ground states (\cref{sec:mb-ferr-slat-determ}).
Finally, we show that the ferromagnetic Slater determinants are insulating in the sense of a charge gap, i.e., both adding and removing an electron costs a finite amount of energy (\cref{sec:mb-charge-gap}).

\subsection{A Review of Hartree-Fock Theory}
\label{sec:review-hf-theory}
Recall that Slater determinants defined by the set $\mc{S}$ in \cref{eqn:slater_S} take the form:
\begin{equation}
\label{eq:hf-sd}
\ket{\Psi_S} = \prod_{i=1}^{M} \prod_{\vk \in \mc{K}} \hat b_{i\vk}^{\dag}\ket{\mathrm{vac}},
\end{equation}
where $\ket{\mathrm{vac}}$ is the vacuum state, $M$ is the half the number of bands, and
\begin{equation}
\label{eq:hf-orbitals}
\hat b_{i\vk}^{\dag}=\sum_{n \in \mc{N}}\hat f_{n\vk}^{\dag} \Xi_{ni}(\vk) \qquad \sum_{n \in \mc{N}} \Xi_{ni}^*(\vk) \Xi_{nj}(\vk) = \delta_{ij}
\end{equation}
defines the creation operator for the Hartree-Fock orbitals for each $\vk\in \mc{K}$.

The Hartree-Fock equations can be expressed in terms of the one-body reduced density matrix (1-RDM). The 1-RDM associated with a given Slater determinant $\ket{\Psi_S}$ can be written as
\begin{equation}
[P(\vk)]_{nm}=\braket{\Psi_S|\hat{f}_{m\vk}^{\dagger}\hat{f}_{n\vk}|\Psi_S}=
\sum_{i=1}^{N_{\rm occ}} \Xi_{ni}(\vk)\Xi_{mi}^*(\vk).
\end{equation}
Due to the orthogonality relation on $\Xi(\vk)$ given in~\cref{eq:hf-orbitals}, we may verify that for each $\vk$ the 1-RDM $P(\vk)$ is an orthogonal projection onto an $M$-dimensional vector space. 
\begin{rem}
\label{rem:fsd-rdm}
    The specific matrix representation of a 1-RDM depends on the choice of basis used to define the creation and annihilation operators $\hat{f}_{m\vk}^\dagger$ and $\hat{f}_{m\vk}$.
    Using the gauge fixing scheme from \cref{sec:requ-symm-gauge}, the two ferromagnetic Slater determinant states have 1-RDMs as follows (where $\eta \in \{ \pm 1 \}$):
    \begin{equation}
        \braket{\Psi_{\eta} | \hat{f}_{m\vk}^{\dagger}\hat{f}_{n\vk} | \Psi_{\eta}} =
        \begin{cases}
            \delta_{mn} & \eta m, \eta n > 0 \\
            0 & \eta m < 0 \text{ or } \eta n < 0
        \end{cases}
    \end{equation}
    That is, the 1-RDM for $\ket{\Psi_{+}}$ is identity for all $\vk$ and all for $m, n > 0$ and zero otherwise (and similarly for $\ket{\Psi_{-}}$).
\end{rem}

Let us define the manifold of all admissible uniformly-filled 1-RDMs:
\[
\mc{M} := \{ P \in \CC^{(2 M) \times (2M) } : P^2 = P,~ P^\dagger = P,~ \tr{(P)} = M \}.
\]
Following the standard derivation of Hartree-Fock theory (see e.g.,~\cite{SzaboOstlund1989,MartinReiningCeperley2016}), the Hartree-Fock energy is given by finding the Slater determinant $\ket{\Psi_{S}}$ which minimizes the energy of the interacting Hamiltonian:
\[
\mathcal{E}^{\rm (HF)} = \min_{\ket{\Psi_S}} \braket{\Psi_S | \hat{H}_{FBI} | \Psi_S}.
\]
Since $\hat{H}_{FBI}$ only involves the electron-electron interaction term, the Hartree-Fock energy can be written as the sum of two functionals acting on the 1-RDM $P \in \mc{M}^{\mc{K}}$. In particular,
\begin{equation}
  \label{eq:hf-energy}
  \mathcal{E}^{\rm (HF)} = \min_{P \in \mc{M}^{\mc{K}}} \Big(J[P] + K[P] \Big)
\end{equation}
where $J[\cdot]$ and $K[\cdot]$ are non-linear functionals in $P$, referred to as the Hartree and Fock energy functionals, respectively.

Due to the specific form of $\hat{H}_{FBI}$, using Wick's theorem, the Hartree and Fock energies (up to a physically irrelevant constant) can be concisely written in terms the matrix $Q(\vk) := 2 P(\vk) - I$ as follows:
\begin{align}
  J[P] & = \frac{1}{| \Omega | N_{\vk}} \sum_{\vk, \vq \in \mc{K}} \sum_{\vG \in \Gamma^*} V(\vG)   \tr{\Big(\Lambda_{\vk}(\vG) Q(\vk) \Big)} \tr{\Big(\Lambda_{\vk + \vq}(\vG)^\dagger Q(\vk + \vq)\Big)},
  \label{eq:hartree-energy} \\
  K[P] & = -\frac{1}{| \Omega | N_{\vk}} \sum_{\vk,\vq \in \mc{K}} \sum_{\vG \in \Gamma^*} V(\vq + \vG)  \tr{\Big( \Lambda_{\vk}(\vq + \vG) Q(\vk + \vq) \Lambda_{\vk}(\vq + \vG)^\dagger Q(\vk)\Big) }.
\label{eq:fock-energy}
\end{align}

\subsection{Positive Semidefiniteness}
\label{sec:mb-posit-semid}
In this section, we prove that the many-body Hamiltonian $\hat{H}_{FBI}$ is positive semidefinite.
We start by proving the following
\begin{lemm}
  For all $\vq' \in \R^{2}$, the operator $\hat{\rho}(\vq')$ satisfies $\widehat{\rho}(\vq')^{\dagger} = \widehat{\rho}(-\vq')$.
\end{lemm}
\begin{proof}
For any fixed $\vq \in \mc{K}$ and $\vG \in \Gamma^{*}$, we calculate
\begin{equation}
  \begin{split}
    \widehat{\rho}(\vq + \vG)^{\dagger}
    & = \sum_{\vk \in \mc{K}} \sum_{m,n \in \mc{N}} [\Lambda_{\vk}(\vq + \vG)^{*}]_{mn} \left(\hat{f}_{n(\vk+\vq)}^{\dagger} \hat{f}_{m\vk} - \frac{1}{2} \delta_{\vq,0} \delta_{mn} \right) \\
    & = \sum_{\vk \in \mc{K}} \sum_{m,n \in \mc{N}} [\Lambda_{\vk}(\vq + \vG)^{*}]_{nm} \left(\hat{f}_{m(\vk+\vq)}^{\dagger} \hat{f}_{n\vk} - \frac{1}{2} \delta_{\vq,0} \delta_{mn} \right).
  \end{split}
\end{equation}
Now we would like to make the change of variables $\vk \mapsto \vk - \vq$.
Unfortunately, in general  $\vk, \vq \in \mc{K}$ does not imply that $\vk - \vq \in \mc{K}$.
However, due to the definition of $\mc{K}$, we may always find a $\vG_{\vk,\vq} \in \Gamma^{*}$ so that
\begin{equation}
 \vk - \vq = \widetilde{\vk - \vq} + \vG_{\vk,\vq} \quad \text{where} \quad \widetilde{\vk - \vq} \in \mc{K}.
\end{equation}
Under the change of variables $\vk \mapsto \widetilde{\vk - \vq}$ we have that
\begin{equation}
  \begin{split}
    \widehat{\rho}(\vq + \vG)^{\dagger}
    & = \sum_{\vk \in \mc{K}} \sum_{m,n \in \mc{N}} [\Lambda_{\vk - \vq + \vG_{\vk,\vq}}(\vq + \vG)^{*}]_{nm} \left(\hat{f}_{m(\vk+\vG_{\vk,\vq})}^{\dagger} \hat{f}_{n(\vk - \vq + \vG_{\vk,\vq})} - \frac{1}{2} \delta_{\vq,0} \delta_{mn} \right) \\
    & = \sum_{\vk \in \mc{K}} \sum_{m,n \in \mc{N}} [\Lambda_{\vk - \vq}(\vq + \vG)^{*}]_{nm} \left(\hat{f}_{m\vk}^{\dagger} \hat{f}_{n(\vk - \vq)} - \frac{1}{2} \delta_{\vq,0} \delta_{mn} \right) \\
    & = \sum_{\vk \in \mc{K}} \sum_{m,n \in \mc{N}} [\Lambda_{\vk}(-\vq - \vG)]_{mn} \left(\hat{f}_{m\vk}^{\dagger} \hat{f}_{n(\vk - \vq)} - \frac{1}{2} \delta_{\vq,0} \delta_{mn} \right) = \widehat{\rho}(- \vq - \vG)
  \end{split}
\end{equation}
where in the second line we have used~\cref{eq:period_f} and~\cref{eq:form-factor-shift}.
\end{proof}
Since $\hat{H}_{FBI}$ takes the form
\begin{equation}
  \hat{H}_{FBI} = \sum_{\vq'} \hat{V}(\vq') \widehat{\rho}(\vq') \widehat{\rho}(-\vq') = \sum_{\vq'} \hat{V}(\vq') \widehat{\rho}^{\dag}(\vq') \widehat{\rho}(\vq')
\end{equation}
and $\hat{V}(\vq')=\hat{V}(-\vq') > 0$, and immediate corollary is that \textit{any} many-body state $\ket{\Psi}$ so that $\hat{H}_{FBI} \ket{\Psi} = 0$ must be a ground state.

\subsection{Ferromagnetic Slater Determinants are Ground States}
\label{sec:mb-ferr-slat-determ}
We can now check that the ferromagnetic Slater determinant states are ground states by verifying they are zero energy eigenstates.

\begin{prop}[Proof of \cref{result:exact-gs}]
  \label{prop:hf-gs}
  Suppose that the single particle Hamiltonian $H$ satisfies the symmetry assumption (\cref{assume:symm}). Then the two ferromagnetic Slater determinant states (\cref{def:ferro-sd}) are exact many-body ground states of the interacting model \cref{eq:h-fbi}.
\end{prop}
\begin{proof}
  We show that for all $\vq' = \vq + \vG$, $\widehat{\rho}(\vq') \ket{\Psi_{\pm}} = 0$ where $\ket{\Psi_{\pm}}$ is a ferromagnetic Slater determinant.
  While we only consider $\ket{\Psi_{+}}$, the calculation for $\ket{\Psi_{-}}$ follows similar steps.
  \begin{equation}
    \begin{split}
      \widehat{\rho}(\vq + \vG) \ket{\Psi_{+}}
      & = \sum_{\vk \in \mc{K}} \left[ \sum_{m,n} [\Lambda_{\vk}(\vq + \vG)]_{mn} \left(\hat{f}_{m\vk}^\dagger \hat{f}_{n(\vk+\vq)} - \frac{1}{2} \delta_{mn} \delta_{\vq,\vzero}  \right) \right] \ket{\Psi_{+}} \\
      & = \delta_{\vq, \vzero} \sum_{\vk \in \mc{K}} \left[ \sum_{m,n} [\Lambda_{\vk}(\vG)]_{mn} \left(\hat{f}_{m\vk}^\dagger \hat{f}_{n\vk} - \frac{1}{2} \delta_{mn}  \right) \right] \ket{\Psi_{+}} \\
      & = \delta_{\vq, \vzero} \sum_{\vk \in \mc{K}} \left[ (\sum_{m > 0} [\Lambda_{\vk}(\vG)]_{mm}) - \frac{1}{2} (\sum_{m} [\Lambda_{\vk}(\vG)]_{mm})  \right] \ket{\Psi_{+}} \\
      & = \frac{1}{2}\delta_{\vq, \vzero} \left[\sum_{\vk \in \mc{K}} (\sum_{m > 0} [\Lambda_{\vk}(\vG)]_{mm} - [\Lambda_{\vk}(\vG)^{*}]_{mm})  \right] \ket{\Psi_{+}} \\
      & = \frac{1}{2}\delta_{\vq, \vzero} \left[\sum_{\vk \in \mc{K}} \Im \tr(A_{\vk}(\vG)) \right] \ket{\Psi_{+}} = 0
    \end{split}
  \end{equation}
  where in the second line, we have use the fact that $\Lambda_{\vk}(\vq + \vG)$ is block diagonal and the ferromagnetic Slater determinant state fully fills the $n > 0$ states.
  The last expression vanishes due to the sum rule in \cref{eq:form-factor-sum-rule}.
\end{proof}

\subsection{Charge Gap at the Thermodynamic Limit}
\label{sec:mb-charge-gap}

As we will see, adding or removing an electron to a ferromagnetic Slater determinant always costs a non-zero amount of energy.
\begin{prop}[Charge Gap of Ferromagnetic Slater Determinants]
Fix $\eta \in \{ \pm 1 \}$ and consider an arbitrary single electron excitation of the ferromagnetic Slater determinant $\ket{\Psi_{\eta}}$.
All such excitations can be written in terms of creation and annihilation operators of the following form:
  \begin{equation}
      \hat{c}^\dagger = \sum_{\vk \in \mc{K}} \sum_{\eta \ell < 0} \hat{f}^\dagger_{\ell\vk} c_{\ell\vk} \qquad       \hat{c} = \sum_{\vk \in \mc{K}} \sum_{\eta \ell > 0} \hat{f}_{\ell\vk} \overline{c_{\ell\vk}} \qquad \sum_{\vk \in \mc{K}} \sum_{\eta \ell > 0} | c_{\ell,\vk} |^2 = 1.
  \end{equation}
  where $\eta \ell < 0$ denotes summation over the set $\{ \ell : \eta \ell < 0 \}$ and similarly for $\eta \ell > 0$.
  
  For such excitations, we have 
  \begin{equation}
  \begin{split}
      \braket{\Psi_{\eta} | \hat{c}^\dagger \hat{H}_{FBI} \hat{c} | \Psi_{\eta}} & = \frac{1}{N_{\vk} | \Omega |} \sum_{\vq'} \hat{V}(\vq') \sum_{\vk \in \mc{K}} \sum_{\eta \ell > 0} \sum_{\eta \ell' > 0} c_{\ell'\vk} [ \Lambda_{\vk}(\vq') \Lambda_{\vk}(\vq')^\dagger]_{\ell',\ell} \overline{c_{\ell\vk}}, \\
    \braket{\Psi_{\eta} | \hat{c} \hat{H}_{FBI} \hat{c}^\dagger | \Psi_{\eta}} & = \frac{1}{N_{\vk} | \Omega |} \sum_{\vq'} \hat{V}(\vq') \sum_{\vk \in \mc{K}} \sum_{\eta \ell < 0} \sum_{\eta \ell' < 0} \overline{c_{\ell'\vk}} [ \Lambda_{\vk}(-\vq')^\dagger \Lambda_{\vk}(-\vq')]_{\ell',\ell} c_{\ell\vk}.
  \end{split}
  \end{equation}
 Hence, since for all $\vk$, $[\Lambda_{\vk}(\vzero)] = \delta_{mn}$, the energy of $\hat{c}^\dagger \ket{\Psi_{\eta}}$ and $\hat{c} \ket{\Psi_{\eta}}$ is positive for all choices of $\hat{c}^\dagger$, $\hat{c}$.
\end{prop}
Due to the constraint $\sum_{\vk} \sum_{\ell} | c_{\ell,\vk} |^2 = 1$,  in the thermodynamic limit ($N_{\vk} \rightarrow \infty$) the energy of the single electron excitations converges to
\begin{equation}
\begin{split}
        \bra{\Psi_{\eta}} & \hat{c}^{\dag} \hat{H}_{FBI} \hat{c}\ket{\Psi_{\eta}} \\ & = \int_{\R^2} \int_{\Omega} \hat{V}(\vq') \sum_{\eta \ell > 0} \sum_{\eta \ell' > 0} c_{\ell'}(\vk) [ \Lambda_{\vk}(\vq') \Lambda_{\vk}(\vq')^\dagger]_{\ell',\ell} \overline{c_{\ell}(\vk)} \ud{\vk} \ud{\vq'}, \\
    \bra{\Psi_{\eta}} & \hat{c} \hat{H}_{FBI} \hat{c}^\dagger\ket{\Psi_{\eta}} \\ & = \int_{\R^2} \int_{\Omega} \hat{V}(\vq') \sum_{\eta \ell < 0} \sum_{\eta \ell' < 0} \overline{c_{\ell'}(\vk)} [ \Lambda_{\vk}(-\vq')^\dagger \Lambda_{\vk}(-\vq')]_{\ell',\ell} c_{\ell}(\vk) \ud{\vk} \ud{\vq'}. 
\end{split}
\end{equation}
which is again strictly positive and so the charge gap does not vanish in the thermodynamic limit.
We now proceed to the proof.

\begin{proof}
From the canonical anticommutation relations, we have that for all  $\vq' = \vq + \vG$
\begin{equation}
  \begin{split}
    [ \hat{f}_{n\vk}^\dagger \hat{f}_{m(\vk + \vq')}, \hat{f}_{\ell,\vk'}^{\dagger} ] & = \delta_{m\ell} \delta_{\vk',(\vk + \vq')} \hat{f}_{n\vk}^{\dagger}, \\
    [ \hat{f}_{n\vk}^\dagger \hat{f}_{m(\vk + \vq')}, \hat{f}_{\ell,\vk'} ] & = -\delta_{n\ell} \delta_{\vk',\vk} \hat{f}_{m(\vk+\vq')}.
  \end{split}
\end{equation}
Therefore,
\begin{equation}
\label{eq:creation-comm}
\begin{split}
    [ \widehat{\rho}(\vq'), \hat{f}_{\ell,\vk'}^\dagger ]
    & = \sum_{\vk} \sum_{mn} [\Lambda_{\vk}(\vq')]_{mn} \delta_{n\ell} \delta_{\vk',(\vk + \vq')} \hat{f}_{n\vk}^\dagger,
      = \sum_{n} [\Lambda_{\vk' - \vq'}(\vq')]_{\ell n} \hat{f}_{n(\vk'-\vq')}^\dagger, \\
    [ \widehat{\rho}(\vq'), \hat{f}_{\ell,\vk'} ]
    & = -\sum_{\vk} \sum_{mn} [\Lambda_{\vk}(\vq')]_{mn} \delta_{n\ell} \delta_{\vk',\vk} \hat{f}_{m(\vk + \vq')}
      = - \sum_{m} [\Lambda_{\vk'}(\vq')]_{m\ell} \hat{f}_{m(\vk'+\vq')}. \\
\end{split}
\end{equation}
Since the ferromagnetic Slater determinant state $\ket{\Psi_{\eta}}$ fully fills either the positive or negative bands, $\hat{f}_{\ell'\vk''} \ket{\Psi_{\eta}} \neq 0$ if and only if $\eta \ell' > 0$.
Therefore, for a fixed $\eta$, we only consider $\ell'$ so that $\eta \ell' > 0$.
For the ferromagnetic Slater determinant state $\ket{\Psi_{\eta}}$ for all $\vq'$ we have
\begin{equation}
    \begin{split}
        \bra{\Psi_{\eta}} & \hat{f}_{\ell\vk'}^\dagger \widehat{\rho}(\vq') \widehat{\rho}(-\vq') \hat{f}_{\ell'\vk''} \ket{\Psi_{\eta}} = - \braket{\Psi_{\eta} | [ \widehat{\rho}(\vq'), \hat{f}_{\ell\vk'}^\dagger] [\widehat{\rho}(-\vq'), \hat{f}_{\ell'\vk''} ] | \Psi_{\eta}} \\[1ex]
        & = \sum_{mn} [\Lambda_{\vk'-\vq'}(\vq')]_{m\ell} [\Lambda_{\vk''}(-\vq')]_{\ell'n} \braket{\Psi_{\eta} | \hat{f}_{m(\vk'-\vq')}^\dagger \hat{f}_{n(\vk''-\vq')} | \Psi_{\eta} } \\
        & = \sum_{\eta m > 0} \sum_{\eta n > 0} [\Lambda_{\vk'-\vq'}(\vq')]_{m\ell} [\Lambda_{\vk'}(-\vq')]_{\ell'n} \delta_{mn} \delta_{\vk'\vk''} \\
        & = [\Lambda_{\vk'}(-\vq') \Lambda_{\vk'}(-\vq')^\dagger]_{\ell',\ell} \delta_{\vk'\vk''} 
    \end{split}
\end{equation}
where in the last line we have used the identity $\Lambda_{\vk - \vq'}(\vq') = \Lambda_{\vk}(-\vq')^\dagger$ \cref{eq:form-factor-dagger}.

Similar calculations show for $\eta \ell' < 0$
\begin{equation}
    \bra{\Psi_{\eta}} \hat{f}_{\ell\vk'} \widehat{\rho}(\vq') \widehat{\rho}(\vq') \hat{f}_{\ell'\vk''}^\dagger \ket{\Psi_{\eta}} = [\Lambda_{\vk'}(\vq')^\dagger \Lambda_{\vk'}(\vq')]_{\ell',\ell} \delta_{\vk'\vk''}. 
\end{equation}
Therefore, by linearity 
\begin{equation}
    \begin{split}
        \braket{\Psi_{\eta} | \hat{c}^\dagger \widehat{\rho}(\vq') \widehat{\rho}(-\vq') \hat{c} | \Psi_{\eta}} & = \sum_{\vk \in \mc{K}} \sum_{\eta \ell > 0} \sum_{\eta \ell' > 0} c_{\ell'\vk} [ \Lambda_{\vk}(\vq') \Lambda_{\vk}(\vq')^\dagger]_{\ell',\ell} \overline{c_{\ell\vk}} \\
        \braket{\Psi_{\eta} | \hat{c} \widehat{\rho}(\vq') \widehat{\rho}(-\vq') \hat{c}^\dagger | \Psi_{\eta}} & = \sum_{\vk \in \mc{K}} \sum_{\eta \ell < 0} \sum_{\eta \ell' < 0} \overline{c_{\ell'\vk}} [ \Lambda_{\vk}(-\vq')^\dagger \Lambda_{\vk}(-\vq')]_{\ell',\ell} c_{\ell\vk}
    \end{split}
\end{equation}
which implies the proposition.
\end{proof}

\section{The Hartree-Fock Ground States of the Flat-Band Interacting Hamiltonian}
\label{sec:hartree-fock-ground}

\subsection{Rigorous Statement of \cref{result:main-informal}}

We can now state our main result rigorously:
\begin{theo}[Main theorem]
  \label{thm:hf-gs-unique}
  Suppose that the single particle Hamiltonian $H$ satisfies~\cref{assume:h,assume:symm}.
  Suppose further that the Monkhorst-Pack grid $\mc{K}$ has been chosen to satisfy~\cref{assume:grid}.
  If there exists a $\vk \in \mc{K}$ so that
  \begin{enumerate}
    \item For some $\vG$, $\Im \tr{(A_{\vk}(\vG))} \neq 0$
    \item For all non-trivial orthogonal projectors $\Pi$ (i.e. $\Pi$ is not zero or identity), there exists a $\vG'$ so that
          \begin{equation}
            \| (I - \Pi) A_{\vk}(\vG') \Pi \| > 0
          \end{equation}
  \end{enumerate}
  then the two ferromagnetic Slater determinants are the unique
  translation-invariant Hartree-Fock ground states of $\hat{H}_{FBI}$ in ~\cref{eq:h-fbi}.
\end{theo}

The conditions in \cref{thm:hf-gs-unique} involves checking all orthogonal
projectors and are difficult to verify computationally.  For the special case
of $2$ and $4$ flat bands, these conditions can be significantly simplified.
\begin{corr}[Two Band Case]
  \label{corr:two-band-gs}
  Suppose that the single particle Hamiltonian $H$ satisfies~\cref{assume:h,assume:symm}, and has two flat bands.
  Suppose further that the Monkhorst-Pack grid $\mc{K}$ has been chosen to satisfy~\cref{assume:grid}.
  If there exists a $\vk \in \mc{K}$ and $\vG$ so that $\Im (A_{\vk}(\vG)) \neq
  0$ then the two ferromagnetic Slater determinants are the unique
  translation-invariant Hartree-Fock ground states of~\cref{eq:h-fbi}.
\end{corr}
\begin{proof}
  This follows immediately from~\cref{thm:hf-gs-unique} since in this case $A_{\vk}(\vG)$ is a scalar ($1 \times 1$ matrix) and there are no nontrivial orthogonal projectors for scalars.
\end{proof}
\begin{corr}[Four Band Case]
  \label{corr:four-band-gs}
  Suppose that the single particle Hamiltonian $H$ satisfies~\cref{assume:h,assume:symm}, and has four flat bands.
  Suppose further that the Monkhorst-Pack grid $\mc{K}$ has been chosen to satisfy~\cref{assume:grid}.
  If there exists a $\vk \in \mc{K}$ so that
  \begin{enumerate}
    \item For some $\vG$, $\Im \tr{(A_{\vk}(\vG))} \neq 0$,
    \item For some $\vG', \vG''$, $[A_{\vk}(\vG'), A_{\vk}(\vG'')] \neq 0,$
  \end{enumerate}
  then the two ferromagnetic Slater determinants are the unique translation-invariant Hartree-Fock ground states of~\cref{eq:h-fbi}.
\end{corr}
\begin{proof}
  The first condition is the same as~\cref{thm:hf-gs-unique} so we only need to show the commutator condition implies the condition on projectors.

  Since for four bands, $A_{\vk}(\vG)$ is a $2 \times 2$ matrix, the only non-trivial projectors are rank one projectors.
  Therefore, we only need to show that for all $\ket{v} \in \CC^{2}$ with $\| v \| = 1$, there exists a $\vG$ so that $\braket{v^{\perp} | A_{\vk}(\vG) | v} \neq 0$ where $\ket{v^{\perp}}$ is a unit vector orthogonal to $\ket{v}$.

  Suppose that $\vG', \vG''$ are so that $[A_{\vk}(\vG'), A_{\vk}(\vG'')] \neq 0$ and pick some $\ket{v} \in \CC^{2}$.
  We have three cases:
  \begin{description}[itemsep=2ex]
    \item[Case 1 $\braket{v^{\perp} | A_{\vk}(\vG') | v} \neq 0$] Take $\vG = \vG'$
    \item[Case 2 $\braket{v^{\perp} | A_{\vk}(\vG') | v} = 0$ but $\braket{v | A_{\vk}(\vG') | v^{\perp}} \neq 0$] Observe that
          \begin{equation}
            \overline{\braket{v | A_{\vk}(\vG') | v^{\perp}}} = \braket{v^{\perp} | A_{\vk}(\vG')^{\dagger} | v} \neq 0
          \end{equation}
          but by~\cref{eq:form-factor-dagger} $A_{\vk}(\vG')^{\dagger} = A_{\vk}(-\vG')$ so we take $\vG = -\vG'.$
    \item[Case 3 $\braket{v^{\perp} | A_{\vk}(\vG') | v} = 0$ and $\braket{v | A_{\vk}(\vG') | v^{\perp}} = 0$] Since $\ket{v}$ and $\ket{v^{\perp}}$ are an orthogonal basis for $\CC^{2}$, the assumptions in this case implies that $\ket{v}$ and $\ket{v^{\perp}}$ are eigenvectors of $A_{\vk}(\vG')$.
          Therefore, we can orthogonally diagonalize $A_{\vk}(\vG')$ as  $A_{\vk}(\vG') = V \Sigma V^{\dagger}$ where $V$ is an orthogonal matrix whose columns are $\ket{v}$, $\ket{v^{\perp}}$ and $\Sigma$ is a diagonal matrix of eigenvalues.

          Since $[A_{\vk}(\vG'), A_{\vk}(\vG'')] \neq 0$, using the eigendecomposition of $A_{\vk}(\vG')$ we have that $V \Sigma V^{\dagger} A_{\vk}(\vG'') - A_{\vk}(\vG'') V \Sigma V^{\dagger} \neq 0$ which implies
          \begin{equation}
            \label{eq:form-factor-diag}
               \Sigma \Big(V^{\dagger} A_{\vk}(\vG'') V \big) - \Big(V^{\dagger} A_{\vk}(\vG'') V \Big) \Sigma \neq 0
          \end{equation}
          If $V^{\dagger} A_{\vk}(\vG'') V$ were a diagonal matrix then \cref{eq:form-factor-diag} would be zero so it must be that either $\braket{v^{\perp} | A_{\vk}(\vG'') | v} \neq 0$ or $\braket{v | A_{\vk}(\vG'') | v^{\perp}} \neq 0$.
          If $\braket{v^{\perp} | A_{\vk}(\vG'') | v} \neq 0$ we can take $\vG = \vG''$.
          If $\braket{v | A_{\vk}(\vG'') | v^{\perp}} \neq 0$, appealing to~\cref{eq:form-factor-dagger}, we can take $\vG = -\vG''$.
  \end{description}
\end{proof}
\section{Proof of {\cref{thm:hf-gs-unique}}}
\label{sec:proof-thm-hf-gs}

As we saw in~\cref{sec:requ-symm-gauge}, with a proper choice of gauge, the sublattice symmetry $\mc{Z}$ implies we can partition the set of flat bands $\mc{N}$ into two sets $n > 0$ and $n < 0$ whose basis functions are supported on the $A$ or $B$ sublattices respectively.
The composite symmetry $\mc{Q}$ further implies these two sets are of equal size and are related by an antiunitary transformation.

To make use of this observation, suppose that we have a 1-RDM $P(\vk)$ for a model with $2 M$ flat bands. If this state is uniformly half-filled, then for each $\vk \in \mc{K}$, $P(\vk)$ can be expressed as a $(2 M) \times (2 M)$ projection matrix with rank $M$.
Hence, we can write $P(\vk) = \Phi(\vk) \Phi(\vk)^{\dagger}$ where $\Phi(\vk)$ is a $(2 M) \times M$ matrix with orthogonal columns.

Since $\Phi(\vk)$ has orthogonal columns, we may apply the cosine-sine (CS) decomposition \cite{VanLoan1985a} to decompose $\Phi(\vk)$ so that it respects the decomposition into $n > 0$ and $n < 0$:
\begin{equation}
  \label{eq:cs-decomp-psi}
  \Phi(\vk)
  =
  \begin{bmatrix}
    U_1(\vk) & \\
             & U_2(\vk)
  \end{bmatrix}
  \begin{bmatrix}
    \tilde{c}(\vk) & -\tilde{s}(\vk) \\
    \tilde{s}(\vk) & \tilde{c}(\vk)
  \end{bmatrix}
  \begin{bmatrix}
    V(\vk)^\dagger \\
    0
  \end{bmatrix}
\end{equation}
where
\begin{equation}
  \begin{split}
    &\tilde{c}(\vk) = \diag(
        \cos{(\theta_1(\vk) / 2)} , \cos{(\theta_2(\vk) / 2)} , \cdots ,\cos{(\theta_{M}(\vk) / 2)}) \\
    &\tilde{s}(\vk) = \diag(
        \sin{(\theta_1(\vk) / 2)}, \sin{(\theta_2(\vk) / 2)} ,
        \cdots ,
        \sin{(\theta_{M}(\vk) / 2)})
  \end{split}
\end{equation}
and $U_1(\vk), U_2(\vk), V(\vk)$ are $M \times M$ unitary matrices.

Using this decomposition for $\Phi(\vk)$, we can express the Hartree-Fock energy (\cref{eq:hartree-energy,eq:fock-energy}) in terms of the quantities $\theta_{i}(\vk), U_{1}(\vk), U_{2}(\vk)$ and the blocks of the form factor $A_{\vk}(\vq + \vG)$ (see \cref{lem:form-factor-properties} for the definition of $A_{\vk}(\vq + \vG)$).

Since the Hartree-Fock energy is written in terms of the matrix $Q(\vk)$, we begin writing $Q(\vk)$ in terms of the CS decomposition \cref{eq:cs-decomp-psi}.
By definition we have
\begin{equation}
  \begin{split}
    P(\vk) & = \Phi(\vk) \Phi(\vk)^\dagger \\
           & =
             \begin{bmatrix}
               U_1(\vk) & \\
                        & U_2(\vk)
             \end{bmatrix}
             \begin{bmatrix}
               \tilde{c}(\vk)^2 & \tilde{c}(\vk)\tilde{s}(\vk) \\
               \tilde{c}(\vk)\tilde{s}(\vk) & \tilde{s}(\vk)^2
             \end{bmatrix}
             \begin{bmatrix}
               U_1(\vk)^\dagger & \\
                                & U_2(\vk)^\dagger
             \end{bmatrix}.
  \end{split}
\end{equation}
Due to our decomposition of $\Phi(\vk)$, $P(\vk)$ is a ferromagnetic Slater determinant if and only if one of the following holds:
\begin{itemize}
  \item For all $i \in \{ 1, \cdots, M\}$ and all $\vk$, $\theta_{i}(\vk) = 0$, or
  \item For all $i \in \{ 1, \cdots, M\}$ and all $\vk$, $\theta_{i}(\vk) = \pi$
\end{itemize}
These two conditions imply that for the ferromagnetic Slater determinant states either $\tilde{c}(\vk)^{2} = I$ or $\tilde{s}(\vk)^{2} = I$.

Since $Q(\vk) = 2 P(\vk) - I$, we have 
\begin{equation}
  Q(\vk)
  =
  \begin{bmatrix}
    U_1(\vk) & \\
             & U_2(\vk)
  \end{bmatrix}
  \begin{bmatrix}
    c(\vk) & s(\vk)  \\
    s(\vk) & -c(\vk)
  \end{bmatrix}
  \begin{bmatrix}
    U_1(\vk)^\dagger & \\
                     & U_2(\vk)^\dagger
  \end{bmatrix}
\end{equation}
where
\begin{equation}
  \begin{split}
    &c(\vk) = \diag(
        \cos{(\theta_1(\vk))}, \cos{(\theta_2(\vk))},  \cdots, \cos{(\theta_{M}(\vk))}) \\
    &s(\vk) = \diag(
        \sin{(\theta_1(\vk))}, \sin{(\theta_2(\vk))}, \cdots ,\sin{(\theta_{M}(\vk))}.
  \end{split}
\end{equation}
Using trigonometric identities, we see that the ferromagnetic Slater states correspond to having $c(\vk) = \pm I$.

To prove uniqueness of the ferromagnetic Slater determinant states, we will show that they are the unique states in $\mc{S}$ which achieve the minimum value for both the Hartree and the Fock energies \emph{simultaneously}.
We will first show that ferromagnetic Slater determinants are minimizers of the Hartree energy in~\cref{sec:minim-hartr-energy}.
Then in~\cref{sec:minim-fock-energy}, we will show that the assumptions of~\cref{thm:hf-gs-unique} imply that the ferromagnetic Slater determinant is the unique minimizer of the Fock energy.

\subsection{Minimizing the Hartree energy}
\label{sec:minim-hartr-energy}
Since $Q(\vk) = Q(\vk)^{\dagger}$,
\begin{equation}
\overline{\sum_{\vk \in \mc{K}} \tr{(\Lambda_{\vk}(\vG) Q(\vk))}} = \sum_{\vk \in \mc{K}} \tr{(\Lambda_{\vk'}(\vG)^{\dagger} Q(\vk'))}.
\end{equation}
If $\vk, \vq \in \mc{K}$, then $\vq - \vk \in \mc{K} + \Gamma^{*}$. Since $\Lambda_{\vk + \vG'}(\vG) = \Lambda_{\vk}(\vG)$ (\cref{eq:form-factor-shift}) we can perform the change of variables $(\vk, \vq) \rightarrow (\vk, \vq - \vk)$ and write the Hartree energy as
\begin{equation}
J(P) = \frac{1}{|\Omega| N_{\vk}} \sum_{\vG} V(\vG) \left| \sum_{\vk \in \mc{K}} \tr{(\Lambda_{\vk}(\vG) Q(\vk))} \right|^{2}.
\end{equation}
Note that since $V(\vG) > 0$, necessarily $J[P] \geq 0$.

Now recall that, due to the symmetry assumptions on the form factor (\cref{lem:form-factor-properties}) we can write
\begin{equation}
  \Lambda_{\vk}(\vq + \vG)
  =
  \begin{bmatrix}
    A_{\vk}(\vq + \vG) & \\
                       & \overline{A_{\vk}(\vq + \vG)}
  \end{bmatrix}.
\end{equation}
Since we will multiply the form factor by $Q(\vk)$, it will be convenient to define
\begin{equation}
  \begin{split}
    B^{(1)}_{\vk}(\vG) & = U_{1}(\vk) A_{\vk}(\vG) U_{1}(\vk)^{\dagger} \\
    B^{(2)}_{\vk}(\vG) & = U_{2}(\vk) \overline{A_{\vk}(\vG)} U_{2}(\vk)^{\dagger}.
  \end{split}
\end{equation}
The trace in the Hartree energy can then be written in terms of $B^{(1)}_{\vk}(\vG)$ and $B^{(2)}_{\vk}(\vG)$ as follows:
\begin{equation}
  \begin{split}
    \tr{(\Lambda_{\vk}(\vG) Q(\vk) )}
    & = \tr{\Big(
      \begin{bmatrix}
        B^{(1)}_{\vk}(\vG) & \\
        & B^{(2)}_{\vk}(\vG)
      \end{bmatrix}
      \begin{bmatrix}
        c(\vk) & s(\vk) \\
        s(\vk) & -c(\vk)
      \end{bmatrix}
      \Big)} \\
    & = \tr{\Big( (B^{(1)}(\vG) - B^{(2)}(\vG) ) c(\vk) \Big)}.
  \end{split}
\end{equation}

For the ferromagnetic Slater determinant state, $c(\vk) = \pm I$ and hence
\begin{equation}
    \tr{(\Lambda_{\vk}(\vG) Q(\vk) )} = \pm \tr{\Big( (B^{(1)}_{\vk}(\vG) - B^{(2)}_{\vk}(\vG) ) \Big)} = \pm 2 i \Im{\tr{( A_{\vk}(\vG) )}}
\end{equation}
\begin{equation}
J(P) = \frac{2}{| \Omega | N_{\vk}}  \sum_{\vG} V(\vG) \left|\sum_{\vk} \Im{\tr{( A_{\vk}(\vG) )}} \right|^{2} = 0
\end{equation}
where the last equality is due to the sum rule~\cref{eq:form-factor-sum-rule}.
Therefore, the ferromagnetic Slater determinants minimize the Hartree energy.

\subsection{Minimizing the Fock energy}
\label{sec:minim-fock-energy}

For these calculations, we will adopt the shorthand $\vk' := \vk + \vq$ and $\vq' = \vq + \vG$ and generalize the definitions of $B^{(1)}_{\vk}(\vG)$ and $B^{(2)}_{\vk}(\vG)$ from the previous section:
\begin{equation}
  \label{eq:rotated-form-factor}
  \begin{split}
    B^{(1)}_{\vk}(\vq + \vG) & = B^{(1)}_{\vk}(\vq') = U_{1}(\vk) A_{\vk}(\vq') U_{1}(\vk')^{\dagger} \\
    B^{(2)}_{\vk}(\vq + \vG) & = B^{(2)}_{\vk}(\vq') = U_{2}(\vk) \overline{A_{\vk}(\vq')} U_{2}(\vk')^{\dagger}.
  \end{split}
\end{equation}
After some lengthy computations (\cref{sec:reform-fock-energy}) it can be shown that
\begin{equation}
  \begin{split}
  \label{eq:fock-energy-proof}
    K[P] = -\frac{1}{8 | \Omega | N_{\vk}} \sum_{\vk} V(\vq') \sum_{ij} \bigg\{ & | [B^{(1)}_{\vk}(\vq') + B^{(2)}_{\vk}(\vq')]_{ij} |^2 \cos{(\theta_i(\vk) - \theta_j(\vk'))} \\
    & + | [B^{(1)}_{\vk}(\vq') - B^{(2)}_{\vk}(\vq')]_{ij} |^2  \cos{(\theta_i(\vk) + \theta_j(\vk'))} \bigg\}.
  \end{split}
\end{equation}
Let's take a closer look at each of terms in~\cref{eq:fock-energy-proof}
\begin{equation}
  \label{eq:fock-energy-term}
  \begin{split}
    -|& [B^{(1)}_{\vk}(\vq') + B^{(2)}_{\vk}(\vq')]_{ij} |^2 \cos{(\theta_i(\vk) - \theta_j(\vk'))} \\
    & - | [B^{(1)}_{\vk}(\vq') - B^{(2)}_{\vk}(\vq')]_{ij} |^2  \cos{(\theta_i(\vk) + \theta_j(\vk'))}
  \end{split}
\end{equation}
From \cref{eq:fock-energy-term}, we can understand the fundamental mechanism which forces the ground state to be a ferromagnetic Slater determinant. Since $\hat{V}(\vq') > 0$, if we assume that
\begin{equation}
  | [B^{(1)}_{\vk}(\vq') + B^{(2)}_{\vk}(\vq')]_{ij} | > 0 \qquad | [B^{(1)}_{\vk}(\vq') - B^{(2)}_{\vk}(\vq')]_{ij} | > 0
\end{equation}
then for a state to minimize the Fock energy it must be that
\begin{equation}
  \begin{split}
    \cos{(\theta_{i}(\vk) - \theta_{j}(\vk') )} = 1 & \Rightarrow \theta_{i}(\vk) = \theta_{j}(\vk') \pmod{2 \pi} \\
    \cos{(\theta_{i}(\vk) + \theta_{j}(\vk') )} = 1 & \Rightarrow \theta_{i}(\vk) = -\theta_{j}(\vk') \pmod{2 \pi}.
  \end{split}
\end{equation}
as if this weren't the case then we could decrease the energy further.

The first constraint forces $\theta_{i}(\vk)$ to be constant (independent of $i$ and $\vk$) and the second constraint forces $\theta_{i}(\vk) \in \{ 0, \pi\}$. These two facts combined show that the ferromagnetic Slater determinants minimize the Fock energy and suggest a strategy for proving these states are the unique Hartree-Fock minimizers.

By definition
\begin{equation}
  [B^{(1)}_{\vk}(\vq') \pm B^{(2)}_{\vk}(\vq')]_{ij} = [U_{1}(\vk) A_{\vk}(\vq') U_{1}(\vk')^{\dagger} \pm  U_{2}(\vk) \overline{A_{\vk}(\vq')} U_{2}(\vk')^{\dagger}]_{ij}
\end{equation}
While generically, it may be true that that the above quantity does not vanish, since $U_{1}(\vk)$ and $U_{2}(\vk)$ are arbitrary unitaries, for any $i,j$ we can always find specific choices of $U_{1}(\vk)$, $U_{2}(\vk)$ so that the above vanishes.
The assumptions of~\cref{thm:hf-gs-unique} guarantee that enough of these terms do not vanish for every choice of $U_{1}(\vk)$, $U_{2}(\vk)$ to force the ferromagnetic Slater determinant to be the unique minimizer of the Fock energy.

To prove the assumptions of~\cref{thm:hf-gs-unique} are sufficient, we proceed in two steps.
First we show that for one special $\vk$-point, $\vk_{*}$, enough of the entries of $B_{\pm,\vk_{*}}(\vG)$ do not vanish to force $\theta_{i}(\vk_{*}) = \theta_{j}(\vk_{*}) = 0$ or $\theta_{i}(\vk_{*}) = \theta_{j}(\vk_{*}) = \pi$ for all $i,j \in \{ 1, \cdots, M\}$.
This implies that the Fock energy maximizing 1-RDM at $\vk_{*}$ is a ferromagnetic Slater determinant.
After showing this, we use the fact that the grid $\mc{K}$ has been chosen sufficiently finely so that \cref{assume:grid} holds.
Once this is the case, for any $\vk \in \mc{K}$, we can find a path connecting $\vk_{*}$ and $\vk$ and we will show that for all momenta along this path, the minimizing 1-RDM must be the same ferromagnetic Slater determinant as at $\vk_{*}$.

\subsubsection{Local Uniqueness of Ground State}
\label{sec:local-uniq-ground}
We focus on the point $\vk_{*}$, take $\vq = \vzero$ and $i = j$.
In this case,~\cref{eq:fock-energy-term} reduces to
\begin{equation}
    \frac{1}{2} | [B^{(1)}_{\vk_{*}}(\vG) + B^{(2)}_{\vk_{*}}(\vG)]_{ii} |^{2} + \frac{1}{2} | [B^{(1)}_{\vk_{*}}(\vG) - B^{(2)}_{\vk_{*}}(\vG)]_{ii} |^2  \cos{(2 \theta_{i}(\vk_{*}) )}.
\end{equation}
Now notice that
\begin{equation}
  \begin{split}
    \Im{\tr{(A_{\vk_{*}}(\vG))}}
    & = \frac{1}{2i} \tr{(A_{\vk_{*}}(\vG) - \overline{A_{\vk_{*}}(\vG)})} \\
    & = \frac{1}{2i} \tr{\Big(U_{1}(\vk_{*}) A_{\vk_{*}}(\vG) U_{1}(\vk_{*})^{\dagger} - U_{2}(\vk_{*}) \overline{A_{\vk_{*}}(\vG)} U_{2}(\vk_{*})^{\dagger}\Big)} \\
    & = \frac{1}{2i} \tr{(B^{(1)}_{\vk_{*}}(\vG) - B^{(2)}_{\vk_{*}}(\vG))}.
  \end{split}
\end{equation}
Since by assumption $\Im{\tr{(A_{\vk_{*}}(\vG))}} \neq 0$, it must be there exists an $m$ so that $[B^{(1)}_{\vk_{*}}(\vG) - B^{(2)}_{\vk_{*}}(\vG)]_{mm} \neq 0$.
Therefore, to be an optimizer $\cos{(2 \theta_{m}(\vk_{*}) )} = 1$ which implies $\theta_{m}(\vk_{*}) \in \{ 0, \pi\}$.
Now we show that $\theta_{j}(\vk_{*}) = \theta_{m}(\vk_{*})$ for all $j$.

For this part of the proof, we fix the unitary $U_{1}(\vk_{*})$ and show that for this choice of $U_{1}(\vk_{*})$ all of the $\theta_{j}(\vk_{*})$ must agree.
For this let $\{ \ket{i} : i \in \{ 1, \cdots, M\} \}$ denote the standard basis for $\CC^{M}$ so that for any matrix $A_{ij} = \braket{i | A |j}$.

We will prove this result by induction. Let $m_{1} = m$ and suppose that we have already shown that the angles $\{ \theta_{m_{i}}(\vk_{*}): i \in \{ 1, \cdots, n\} \}$ are all equal to $\theta_{m_{1}}(\vk_{*}) \in \{ 0, \pi \}$.
For fixed $U_{1}(\vk_{*})$, consider the orthogonal projector
\begin{equation}
  \Pi = U_{1}(\vk_{*})^{\dagger} \left( \sum_{i=1}^{n} \ket{m_{i}}\bra{m_{i}} \right) U_{1}(\vk_{*}).
\end{equation}

For any $\vG$ we have
\begin{equation}
  \begin{split}
    (I &- \Pi) A_{\vk_{*}}(\vG) \Pi \\
    & = \left(I - U_{1}(\vk_{*})^{\dagger} \left( \sum_{i=1}^{n} \ket{m_{i}}\bra{m_{i}} \right) U_{1}(\vk_{*})\right) A_{\vk_{*}}(\vG) \left( U_{1}(\vk_{*})^{\dagger} \left( \sum_{i=1}^{n} \ket{m_{i}}\bra{m_{i}} \right) U_{1}(\vk_{*}) \right) \\[1ex]
    & = U_{1}(\vk_{*})^{\dagger}  \left( I - \sum_{i=1}^{n} \ket{m_{i}}\bra{m_{i}} \right)   U_{1}(\vk_{*}) A_{\vk_{*}}(\vG) U_{1}(\vk_{*})^{\dagger}  \left( \sum_{i=1}^{n} \ket{m_{i}}\bra{m_{i}} \right) U_{1}(\vk_{*}).
  \end{split}
\end{equation}
But since the spectral norm is unitarily invariant, the second assumption of~\cref{thm:hf-gs-unique} implies that there exists a $\vG'$ so that
\begin{equation}
  \label{eq:spec-norm}
  \left\| \left( I - \sum_{i=1}^{n} \ket{m_{i}}\bra{m_{i}} \right)   U_{1}(\vk_{*}) A_{\vk_{*}}(\vG') U_{1}(\vk_{*})^{\dagger}  \left( \sum_{i=1}^{n} \ket{m_{i}}\bra{m_{i}} \right) \right\| > 0.
\end{equation}
Let $v, w$ be the top right/left singular vectors of the above operator.
Since $\sum_{i=1}^{n} \ket{m_{i}}\bra{m_{i}}$ is an orthogonal projection and $\{ \ket{i} : i \in \{1, \cdots, b\}\}$ forms a complete basis we can write
\begin{equation}
    v = \sum_{i=1}^{n} \alpha_{i} \ket{m_{i}}, \qquad
    w = \sum_{m \not\in \{ m_{i} : i \in \{ 1, \cdots, n \} \}} \beta_{m} \ket{m}
\end{equation}
for some constants $\alpha_{i}, \beta_{i} \in \CC$. By definition of $v,w$ we know that
\begin{equation}
  \begin{split}
    0 & < \braket{w | U_{1}(\vk_{*}) A_{\vk_{*}}(\vG') U_{1}(\vk_{*})^{\dagger}| v} \\
    & = \sum_{i=1}^{n} \sum_{m \not\in \{ m_{i} : i \in \{ 1, \cdots, n \} \}} \alpha_{i} \overline{\beta_{m}} \braket{m | U_{1}(\vk_{*}) A_{\vk_{*}}(\vG') U_{1}(\vk_{*})^{\dagger} | m_{i } }.
  \end{split}
\end{equation}
Hence, there must exist an $m' \not\in \{ m_{i} : i \in \{ 1, \cdots, n \} \}$ and an $m_{i}$ so that
\begin{equation}
  \label{eq:non-vanishing}
  [U_{1}(\vk_{*}) A_{\vk_{*}}(\vG') U_{1}(\vk_{*})^{\dagger}]_{m',m_{i}} \neq 0.
\end{equation}
Define $m_{n+1} := m'$, since
\begin{equation}
  \Big(B^{(1)}_{\vk_{*}}(\vG') + B^{(2)}_{\vk_{*}}(\vG')\Big) + \Big(B^{(1)}_{\vk_{*}}(\vG') - B^{(2)}_{\vk_{*}}(\vG')\Big) = 2 U_{1}(\vk_{*}) A_{\vk_{*}}(\vG') U_{1}(\vk_{*})^{\dagger}.
\end{equation}
\Cref{eq:non-vanishing} implies that either $[B^{(1)}_{\vk_{*}}(\vG') + B^{(2)}_{\vk_{*}}(\vG')]_{m_{n+1},m_{i}} \neq 0$ or $[B^{(1)}_{\vk_{*}}(\vG') - B^{(2)}_{\vk_{*}}(\vG')]_{m_{n+1},m_{i}} \neq 0$.
For simplicity of discussion suppose $[B^{(1)}_{\vk_{*}}(\vG') + B^{(2)}_{\vk_{*}}(\vG')]_{m_{n+1},m_{i}} \neq 0$, the other case follows similarly.
Now recall the terms in the Fock energy:
\begin{equation}
  \begin{split}
    | [B^{(1)}_{\vk_{*}}(\vG') & + B^{(2)}_{\vk_{*}}(\vG')]_{m_{n+1},m_{i}} |^2 \cos{(\theta_{m_{n+1}}(\vk_{*}) - \theta_{m_{i}}(\vk_{*}))} \\
    & + | [B^{(1)}_{\vk_{*}}(\vG') - B^{(2)}_{\vk_{*}}(\vG')]_{m_{n+1},m_{i}} |^2  \cos{(\theta_{m_{n+1}}(\vk_{*}) + \theta_{m_{i}}(\vk_{*}))}.
  \end{split}
\end{equation}
Since $\theta_{m_{i}}(\vk_{*}) = \theta_{m_{1}}(\vk_{*})$, $\theta_{m_{1}} \in \{ 0, \pi \}$, and $[B_{-,\vk_{*}}(\vG')]_{m_{n+1},m_{i}} \neq 0$, to be a minimizer it must be that $\theta_{m_{n+1}}(\vk_{*}) = \theta_{m_{1}}(\vk_{*})$.
Hence by induction, $\theta_{i}(\vk_{*}) \in \{ 0, \pi \}$ and $\theta_{i}(\vk_{*}) = \theta_{j}(\vk_{*})$ for all $i,j$.

\subsubsection{Global Uniqueness of Ground State}
\label{sec:glob-uniq-ground}
For simplicity, let us assume that $\theta_{j}(\vk_{*}) = 0$ for all $j$. The case $\theta_{j}(\vk_{*}) = \pi$ follows similarly.
In this case, for any $\vq$ we have that
\begin{equation}
\begin{split}
  \frac{1}{2} & | [B^{(1)}_{\vk_{*}}(\vq + \vG) + B^{(2)}_{\vk_{*}}(\vq + \vG)]_{ij} |^2 \cos{(\theta_{j}(\vk_{*} + \vq))} \\
  & + \frac{1}{2} | [B^{(1)}_{\vk_{*}}(\vq + \vG) - B^{(2)}_{\vk_{*}}(\vq + \vG)]_{ij} |^2  \cos{(\theta_j(\vk_{*} + \vq))}.
\end{split}
\end{equation}

Therefore, to conclude $\theta_j(\vk_{*} + \vq) = 0$, it is enough to show that for each $r \in \{1, \cdots, M\}$
\begin{equation}
  \label{eq:b-pm-maxrow}
  \text{one of} \quad
  \begin{cases}
    \max_{m} |[B^{(1)}_{\vk_{*}}(\vq + \vG) + B^{(2)}_{\vk_{*}}(\vq + \vG)]_{rm}| > 0  &  \\
    \max_{m} |[B^{(1)}_{\vk_{*}}(\vq + \vG) - B^{(2)}_{\vk_{*}}(\vq + \vG)]_{rm}| > 0  &  \\
  \end{cases}
  \quad
  \text{holds}.
\end{equation}
That is, for each row, we can find a non-zero entry in one of $B^{(1)}_{\vk_{*}}(\vq + \vG) \pm B^{(2)}_{\vk_{*}}(\vq + \vG)$

Now observe that
\begin{equation}
  B_{+,\vk_{*}}(\vq + \vG) + B_{-,\vk_{*}}(\vq + \vG) = U_{1}(\vk_{*}) A_{\vk_{*}}(\vq + \vG) U_{1}(\vk_{*} + \vq)^{\dagger}
\end{equation}
and recall that a full rank matrix must have a non-zero entry in every row.
Since unitary transformations cannot change the rank of a matrix, if we can show that $A_{\vk_{*}}(\vq + \vG)$ is full rank for some $\vG$, then~\cref{eq:b-pm-maxrow} must be true and hence the fact that the minimizing 1-RDM at $\vk_{*}$ is an FSD state implies that the minimizing 1-RDM at $\vk_{*} + \vq$ is the same FSD state.

The following key fact, along with~\cref{assume:grid} and the above calculation, allows us to propagate the ferromagnetic Slater determinant from $\vk_{*}$ to the entire Brillouin zone:
\begin{lemm}
  \label{lem:full-rank}
  Let $\vk, \vk' \in \mc{K}$ and let $\Pi(\vk)$ denote the flat-band projection at $\vk$.
  If $\| \Pi(\vk) - \Pi(\vk') \| < 1$ then there exists a $\vG$ so that $A_{\vk}((\vk' - \vk) + \vG)$ is full rank.
\end{lemm}
\begin{proof}
  Proven in \cref{sec:proof-full-rank}.
\end{proof}
Assuming~\cref{lem:full-rank} is true, the argument is as follows.
By~\cref{assume:grid}, for any $\vk \in \Omega^{*}$ we may construct a path $\{ \vk_{i} \}_{i=1}^{L}$ connecting $\vk_{*}$ and $\vk$ so that $\| \Pi(\vk_{i}) - \Pi(\vk_{i+1}) \| < 1$.
By the above calculation, since $\vk_{1} = \vk_{*}$ and the projector condition holds between $\vk_{1}$ and $\vk_{2}$, it must be that the minimizing 1-RDM at $\vk_{2}$ is the same FSD state as at $\vk_{*}$.
Since the minimizing 1-RDM at $\vk_{2}$ is a FSD state, we may repeat the same argument to conclude the 1-RDM at $\vk_{3}$ \textit{also} the same FSD state.
Continuing down the path, we conclude that the minimizing 1-RDM at $\vk_{i}$ must be the same FSD state as at $\vk_{*}$ for all $i \in \{1, \cdots, L\}$.
Since the choice of $\vk$ was arbitrary, we conclude that minimizing 1-RDM must agree with $\vk_{*}$ throughout the whole Brillouin zone completing the proof.

\section{Applications of \cref{thm:hf-gs-unique} to TBG and eTTG}
\label{sec:appl-main-result}
We begin by translating the conditions of~\cref{thm:hf-gs-unique} from momentum space to real space in~\cref{sec:cond-main-real}.
We then prove~\cref{result:unique} by verifying these real space conditions hold for TBG-2 (\cref{sec:application-tbg-2}), TBG-4  (\cref{sec:application-tbg-4}), and eTTG-4 (\cref{sec:application-ettg-4}).

\subsection{Conditions of~\cref{thm:hf-gs-unique} in Real Space}
\label{sec:cond-main-real}
While the proof of~\cref{thm:hf-gs-unique} is stated in momentum space, the conditions also have natural analogs in real space.
We begin by recalling the definition of $A_{\vk}(\vq + \vG)$
\begin{equation}
  [A_{\vk}(\vq + \vG)]_{mn} = \frac{1}{|\Omega|} \sum_{\vG'} \sum_{\sigma,j} \overline{\hat{u}_{m\vk}(\vG'; \sigma, j)} \hat{u}_{n(\vk + \vq)}(\vG + \vG'; \sigma, j).
\end{equation}
We can substitute in the definition of the Fourier transform to conclude
\begin{equation}
  \begin{split}
    [A_{\vk}&(\vq + \vG)]_{mn} \\
    & = \frac{1}{|\Omega|} \int_{\Omega} \int_{\Omega} \sum_{\vG'} \sum_{\sigma,j} e^{i \vG' \cdot \vr} e^{-i(\vG + \vG') \cdot \vr'} \overline{u_{m\vk}(\vr; \sigma, j)} u_{n(\vk + \vq)}(\vr'; \sigma, j) \ud\vr \ud\vr' \\
    & = \frac{1}{|\Omega|} \int_{\Omega} \int_{\Omega} \sum_{\vG'} \sum_{\sigma,j} e^{i \vG' \cdot ( \vr - \vr') } e^{-i \vG \cdot \vr'} \overline{u_{m\vk}(\vr; \sigma, j)} u_{n(\vk + \vq)}(\vr'; \sigma, j) \ud\vr \ud\vr' \\
    & = \int_{\Omega} e^{-i \vG \cdot \vr}  \sum_{\sigma,j} \overline{u_{m\vk}(\vr; \sigma, j)} u_{n(\vk + \vq)}(\vr; \sigma, j) \ud\vr.
  \end{split}
\end{equation}
We now define the pair product $\rho_{\vk,\vk+\vq}(\vr)$
\begin{equation}
    [\rho_{\vk,\vk+\vq}(\vr)]_{mn} := \sum_{\sigma,j}  \overline{u_{m\vk}(\vr; \sigma, j)} u_{n(\vk + \vq)}(\vr; \sigma, j).
\end{equation}
Since $\rho_{\vk,\vk+\vq}$ is periodic with respect to $\Gamma$ (since $u_{n\vk}$ are periodic) we can view the form factor $A_{\vk}(\vq + \vG)$ as the Fourier series coefficients of the pair product.
Since the Fourier transform is isometric up to scaling, we can equivalently state conditions for~\cref{thm:hf-gs-unique} in real space.
\begin{prop}[\Cref{thm:hf-gs-unique} in Real Space]
  \label{prop:hf-gs-unique-real}
  Suppose that the single particle Hamiltonian $H$ satisfies~\cref{assume:h,assume:symm}.
  Suppose further that the Monkhorst-Pack grid $\mc{K}$ has been chosen to satisfy~\cref{assume:grid}.
  If there exists a $\vk \in \mc{K}$ so that
  \begin{enumerate}
    \item For some $\vr$, $\tr{(\rho_{\vk,\vk}(\vr))} \neq \overline{\tr{(\rho_{\vk,\vk}(-\vr))}}$
    \item For all non-trivial projections $\Pi$, there exists a $\vr'$ so that
          \begin{equation}
           \| (I - \Pi) \rho_{\vk,\vk}(\vr') \Pi \| > 0
          \end{equation}
  \end{enumerate}
  then the two ferromagnetic Slater determinants are the unique Hartree-Fock ground states of~\cref{eq:h-fbi}.
\end{prop}
We also have the following simple corollaries:
\begin{corr}[Two Bands Case in Real Space]
  \label{corr:two-band-gs-real}
  Suppose that the single particle Hamiltonian $H$ satisfies~\cref{assume:h,assume:symm} and has two flat bands.
  Suppose further that the Monkhorst-Pack grid $\mc{K}$ has been chosen to satisfy~\cref{assume:grid}.
  If there exists a $\vk \in \mc{K}$ and $\vr \in \Omega$ so that $\rho_{\vk,\vk}(\vr) \neq \overline{\rho_{\vk,\vk}(-\vr)}$ then the two ferromagnetic Slater determinants are the unique Hartree-Fock ground states of~\cref{eq:h-fbi}.
\end{corr}
\begin{corr}[Four Band Case in Real Space]
  \label{corr:four-band-gs-real}
  Suppose that the single particle Hamiltonian $H$ satisfies~\cref{assume:h,assume:symm} and has four flat bands.
  Suppose further that the Monkhorst-Pack grid $\mc{K}$ has been chosen to satisfy~\cref{assume:grid}.
  If there exists a $\vk \in \mc{K}$ so that
  \begin{enumerate}
    \item For some $\vr$, $\tr{(\rho_{\vk,\vk}(\vr))} \neq \overline{\tr{(\rho_{\vk,\vk}(-\vr))}}$,
    \item For some $\vr', \vr''$
          \begin{equation}
            \det{
              \begin{bmatrix}
                \| u_{1\vk}(\vr)\|^2 - \| u_{2\vk}(\vr)\|^2  & \braket{u_{1\vk}(\vr), u_{2\vk}(\vr)} \\
                \| u_{1\vk}(\vr')\|^2 - \| u_{2\vk}(\vr')\|^2 & \braket{u_{1\vk}(\vr'), u_{2\vk}(\vr')}
              \end{bmatrix}}
          \end{equation}
  \end{enumerate}
  then the two ferromagnetic Slater determinants are the unique Hartree-Fock ground states of~\cref{eq:h-fbi}.
\end{corr}
For the proofs of \cref{prop:hf-gs-unique-real,corr:two-band-gs-real,corr:four-band-gs-real} we refer the reader to~\cref{sec:real-space-proof}.

\subsection{Application of Main Theorem to TBG-2}
\label{sec:application-tbg-2}
For our next two propositions, we will use the Jacobi $\theta$ function
\begin{equation}
\label{eq:theta}
\begin{gathered}
 \theta_{1} ( \zeta | \omega )
 := - \sum_{ n \in \mathbb Z } \exp ( \pi i (n+\tfrac12) ^2 \omega+ 2 \pi i ( n + \tfrac12 ) (\zeta + \tfrac
12 )  ) ,
\end{gathered}
\end{equation}
which satisfies
\[ \theta_{1} ( \zeta + m | \omega  ) = (-1)^m \theta_{1}  ( \zeta| \omega ) , \ \ \theta_{1} ( \zeta + n \omega| \omega) = (-1)^n e^{ - \pi i n^2 \omega -
2 \pi i \zeta  n } \theta_{1}  ( \zeta |\omega ) .
\]
Furthermore, $ \theta_1(\bullet\vert \omega) $ has simple zeros at $ \Lambda $ (and no other zeros) -- see \cite{MumfordMusili2007}. In the sequel we shall just write $\theta(\zeta):=\theta_1(\zeta\vert \omega)$ to simplify the notation.

We also use the $\wp$ function $\wp(z):=\wp ( z ; \tfrac 43 \pi i \omega,
\tfrac 43 \pi i \omega^2 )$. Here $ \wp ( z; \omega_1 , \omega_2 ) $ is the Weierstrass
$\wp$-function -- see \cite[\S I.6]{MumfordMusili2007}. It is periodic with respect to $ \ZZ \omega_1 + \ZZ \omega_2 $ and its derivative has a pole of order
$ 3 $ at $ z = 0 $. It has the property that
\begin{equation}
\label{eq:property}
\wp(\omega z) = \omega \wp(z).
\end{equation}

We now turn to verifying the conditions of having a unique many-body ground state.
\begin{prop}[Simple magic angle]
\label{prop:simple}
Let $\alpha \in \CC$ be a simple magic angle of TBG-2, then for all $\vk \notin \Gamma^*$ there is $\vG \in \Gamma^*$ such that
\[ \Im \tr(A_{\vk}(\vG))\neq 0\]
while for $\vk \in \Gamma^*$ one has
\begin{equation}
\label{eq:not0}
\Im \tr(A_{\vk}(\vG))=0 \text{ for all }\vG \in \Gamma^*.
\end{equation}
\end{prop}
\begin{proof}
From \cref{corr:two-band-gs-real} we deduce that the existence of some $\vG \in \Gamma^*$ such that
\[ \Im \tr(A_{\vk}(\vG))\neq 0\]
is equivalent to the existence of some $\vr$ such that $\Vert \vu_{\vk}(\vr)\Vert \neq \Vert \vu_{\vk}(-\vr)\Vert$, where we use here the Euclidean norm of the $2$-vector $\vu_{\vk}$ that is given by
\[ \vu_{\vk} \in \ker_{L^2_2}(D(\alpha)+(k_1+ik_2) \operatorname{id}_{\mathbb C^2}) \setminus \{0\}.\]
At simple magic angles there is a unique solution $\vu_{\mathbf 0}$ which obey reflection symmetry, see \cite[Theo.1]{BeckerHumbertZworski2023} (also Figure \ref{fig:seeking_alpha} for an illustration).
\begin{figure}
\includegraphics[width=7.5cm]{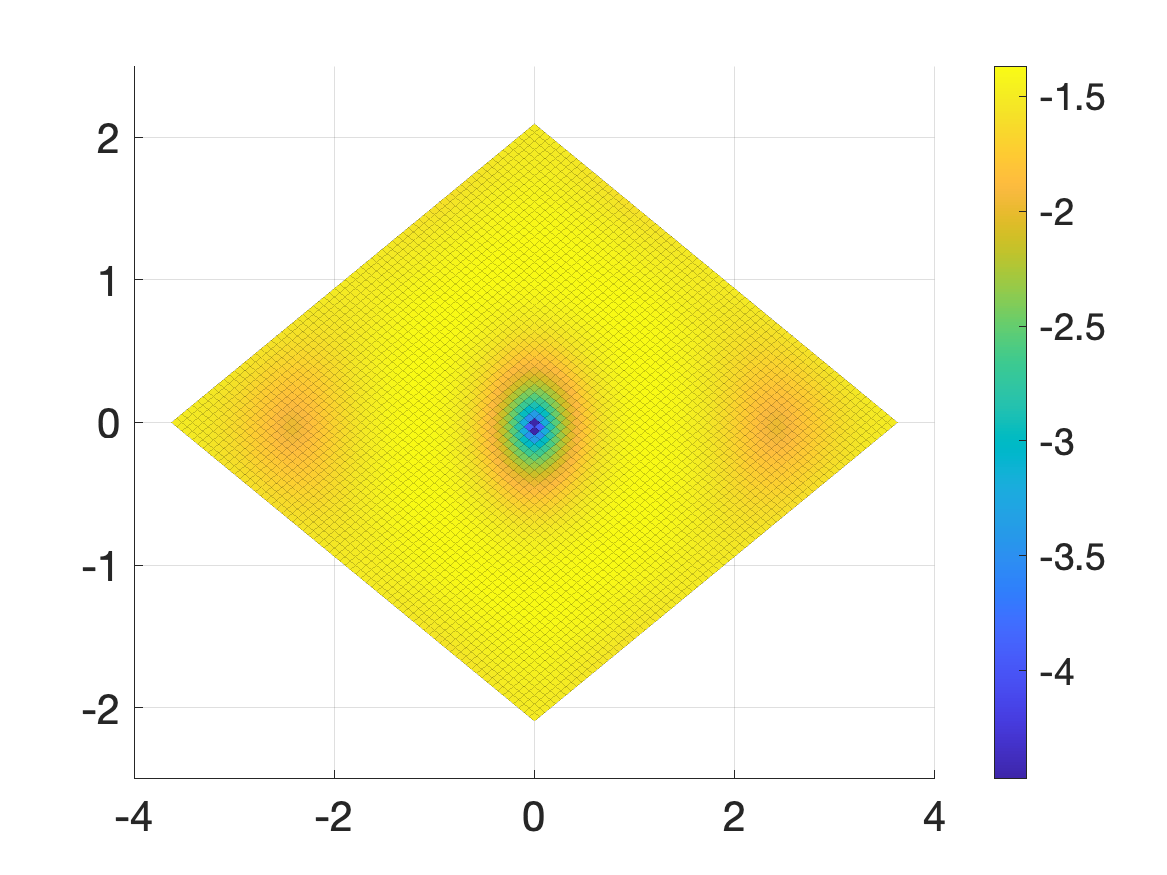}
\includegraphics[width=7.5cm]{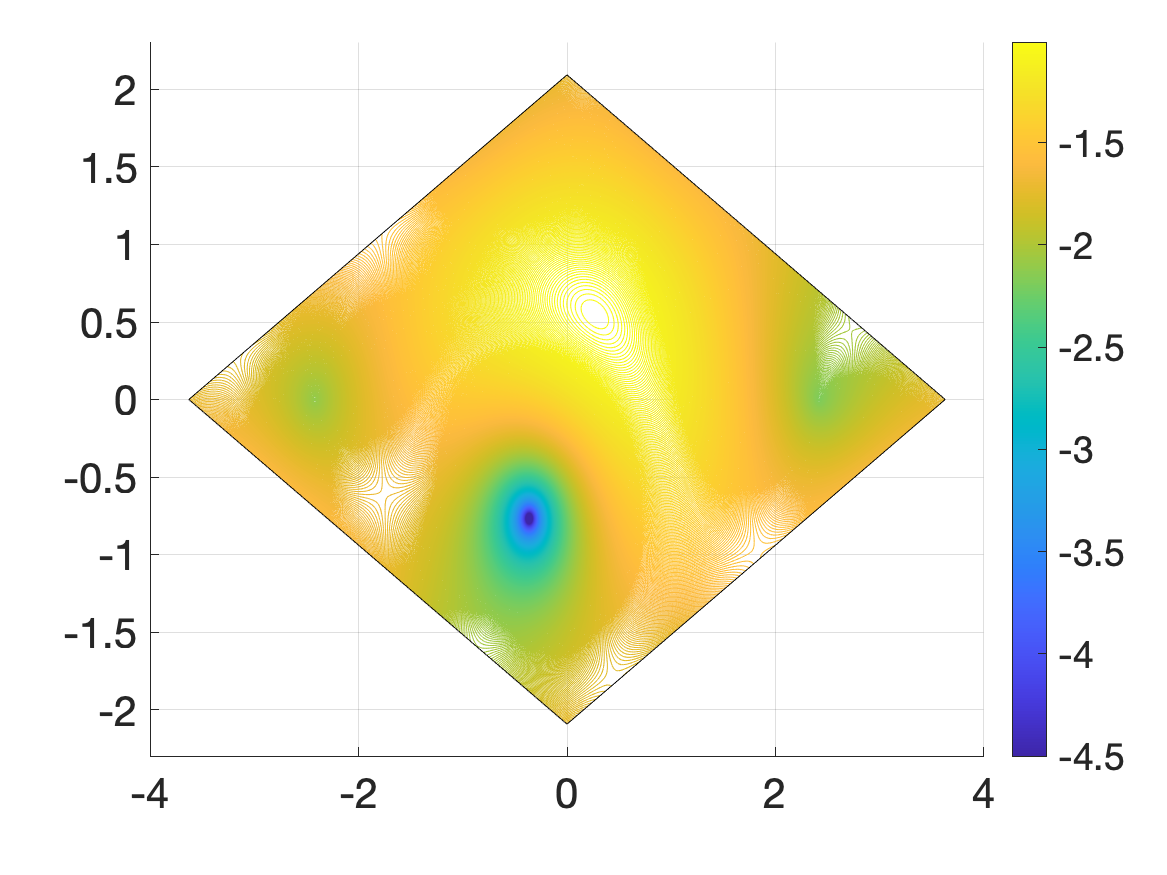}
\caption{\label{fig:seeking_alpha}Log of modulus of flat-band wavefunction with $U_0$ as in \eqref{eq:ref_pots} and $\vk=0$, $\alpha \approx 0.58656$ exhibiting reflection symmetry with zero in the center (left) and $\vk =(2,0.7)$ with zero away from center (right) showing $\Vert u_{\vk}(\vr)\Vert \neq \Vert u_{\vk}(-\vr)\Vert.$}
\end{figure}
Indeed, using symmetry \eqref{eq:Lsymm}, we find
\[ \Vert \vu_{\mathbf 0}(\vr)\Vert =\Vert \mathcal L \vu_{\mathbf 0}(\vr) \Vert = \Vert \vu_{\mathbf 0}(-\vr)\Vert \]
showing \eqref{eq:not0}.
By \cite[Theo. $3$]{BeckerHumbertZworski2022a} the function $\mathbf u_{\mathbf 0}$ has a simple zero at $\mathbf r=\mathbf 0$ and no other zeros in its fundamental domain.

From \cite[Lemma $3.1$]{BeckerHumbertZworski2022a},
\begin{equation}
\label{eq:identity}
 \vu_{\vk}(\vr) = F_{\vk}(\vr) \vu_{\mathbf 0}(\vr) \in \ker(D(\alpha)+(k_1+ik_2)\operatorname{id}_{\CC^2}),
 \end{equation}
where
\[ F_{\vk}(\vr):=e^{\frac{(k_1+ik_2)}{2}(-i(1+\omega)x+(\omega-1)y)} \frac{ \theta ( \frac{3(x_1+ix_2)}{4\pi i \omega} + \frac{k_1+ik_2}{\sqrt{3}\omega} ) }{
\theta\Big( \frac{3(x_1+ix_2)}{4\pi i \omega}\Big)}.\]
It follows that $\vu_{\vk}(\vr)$ has a unique zero at $\vr = \frac{4\pi}{3\sqrt{3}}(k_2,-k_1)^\top$ per unit cell (note that both $\vu_{\mathbf 0}$ and $\theta$ have a simple zero at $\vr=\mathbf 0$). Thus, $\Vert \vu_{\vk}\Vert$ cannot be an even function for $\vk \notin \Gamma^*.$
\end{proof}
Therefore, by \cref{corr:two-band-gs-real}, we have the following result
\begin{theo}
\label{thm:application-tbg-2}
   The ferromagnetic Slater determinants states are the unique ground states of the corresponding flat-band interacting model of TBG-2.
\end{theo}

\subsection{Application of Main Theorem to TBG-4}
\label{sec:application-tbg-4}
A similar result also holds for two-fold degenerate magic angles in chiral limit TBG.

\begin{prop}[Two-fold degenerate magic angle]
\label{prop:two}
Let $\alpha \in \CC$ be a two-fold degenerate magic angle of TBG-4, then for $\vk  = \pm\vq_1$, are invariant points under the rotation and translation symmetry, there are $\vG \in \Gamma^*$ such that
\begin{equation}
\label{eq:not02}
\Im \tr(A_{\vk}(\vG))\neq 0,
\end{equation}
while for $\vk \in \Gamma^*$ one has
\begin{equation}
\label{eq:not03}
\Im \tr(A_{\vk}(\vG))=0 \text{ for all }\vG \in \Gamma^*.
\end{equation}
\end{prop}
\begin{proof}
For a two-fold degenerate magic angle there is by \cite[Theo.$1$]{BeckerHumbertZworski2023}, a unique element $\vw_{\mathbf 0} \in \ker_{L^2_{2,1}}(D(\alpha))$ and $\vv_{\mathbf 0}(\vr ) = \wp(x_1+ix_2)\vw_{\mathbf 0}(\vr) \in \ker_{L^2_{2,0}}(D(\alpha)).$

Since they belong to $L^2$-orthogonal subspaces, due to different rotational symmetry using \eqref{eq:property}, we have $\tr(A_{\mathbf 0}(\vr)) = \Vert \vv_{\mathbf 0}(\vr ) \Vert^2 + \Vert \vw_{\mathbf 0}(\vr ) \Vert^2.$
Thus, by the symmetry \eqref{eq:Lsymm}, we find
\begin{equation}
\begin{split}
\label{eq:trace}
\tr(A_{\mathbf 0}(\vr))&= \Vert \vw_{\mathbf 0}(\vr) \Vert^2 + \Vert \vv_{\mathbf 0}(\vr) \Vert^2 \\
&= \Vert \mathcal L \vw_{\mathbf 0}(\vr) \Vert^2 + \Vert \mathcal L \vv_{\mathbf 0}(\vr) \Vert^2\\
&= \Vert \vw_{\mathbf 0}(-\vr) \Vert^2 + \Vert \vv_{\mathbf 0}(-\vr) \Vert^2 = \tr(A_{\mathbf 0}(-\vr)).
\end{split}
\end{equation}
We recall from \cref{corr:four-band-gs-real} the equivalence
\[ \Im \tr(A_{\vk}(\vG)) \neq 0 \text{ for some }\vG\text{ if and only if }\tr(\rho_{\vk,\vk}(\vr)) \neq \tr(\rho_{\vk,\vk}(-\vr))\text{ for some }\vr\]
showing  \eqref{eq:not03}, as all $\vk \in \Gamma^*$ are equivalent to $\vk=\bm 0.$

We recall from \cite[Theo. $3$]{BeckerHumbertZworski2023} and \cite[Lemm. $7.1$]{BeckerHumbertZworski2023}  that $\vw_{\mathbf 0}$ has a zero of order $2$ at $\mathbf r=\mathbf 0$ and $\vv_{\mathbf 0}$ has simple zeros at $\pm \vr_S$ with $\vr_S= (\frac{4\pi}{3\sqrt{3}},0)^{\top},$ see Figure \ref{fig:seeking_alpha2}.
\begin{figure}
 \includegraphics[width=7.5cm]{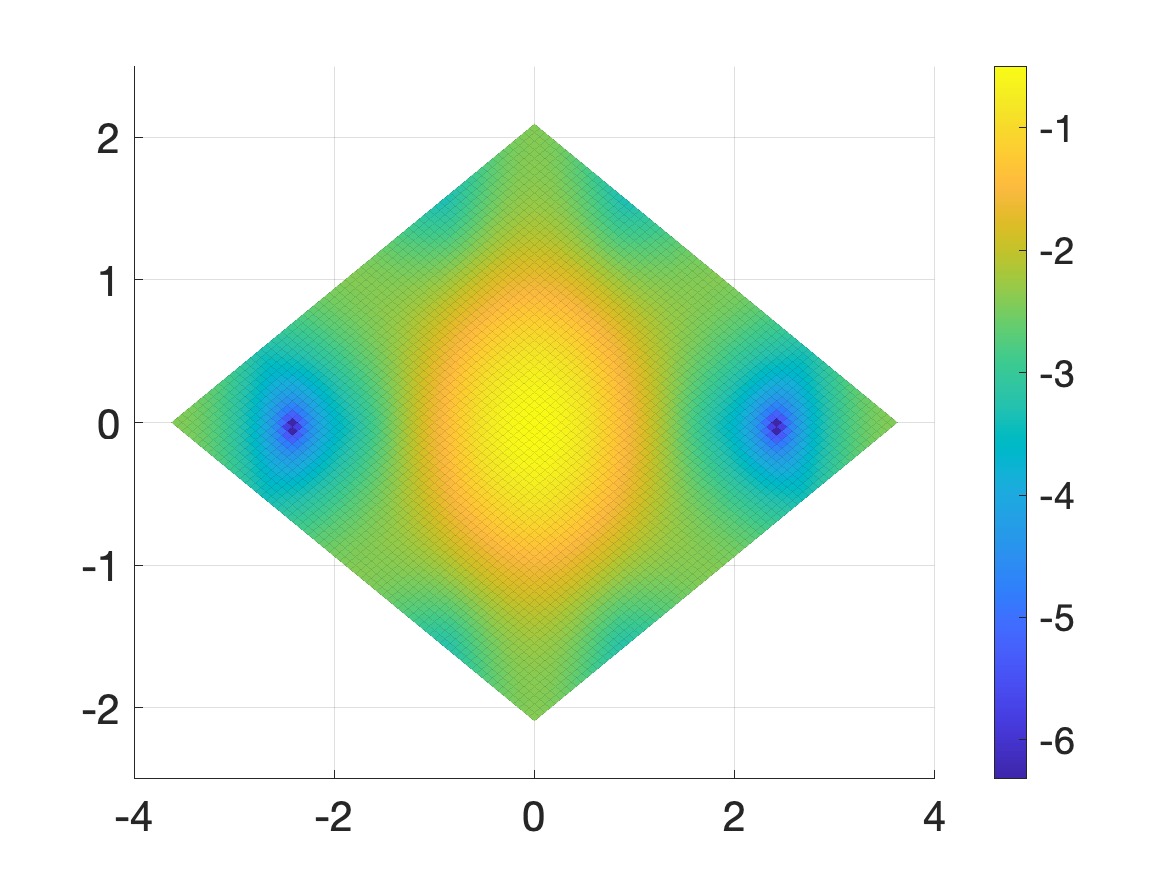} \includegraphics[width=7.5cm]{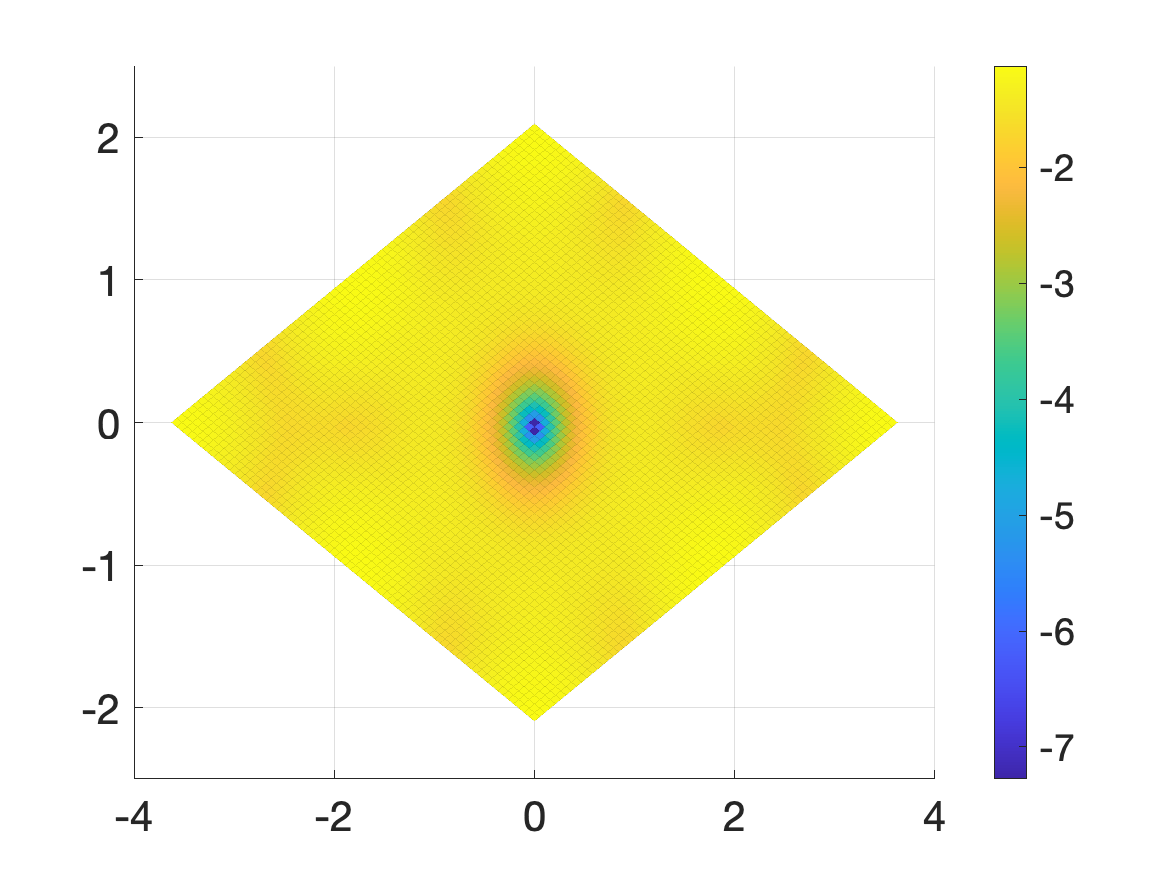}
\caption{\label{fig:seeking_alpha2}Log of modulus of flat-band wavefunctions two flat bands for $U_{7/8}$ as in \eqref{eq:ref_pots} with $\vk=0$, $\alpha \approx 0.853799$ with $\vv_{\mathbf 0}$ (left) (with simple zeros at $\pm \vr_{S}$ and $\vw_{\mathbf 0}$(right) with double zero at $\vr=0$. }
\end{figure}
We then define functions $ \vv_{\vk}, \vw_{\vk} \in \ker(D(\alpha)+(k_1+ik_2)\operatorname{id}_{\CC^2})$ by
\[ \vv_{\vk}(\vr) = \alpha_{\vk} F_{\vk}(\vr - \vr_S) \vv_{\mathbf 0}(\vr)  \text{ and }\vw_{\vk}(\vr) = \beta_{\vk} F_{\vk}(\vr)\vw_{\mathbf 0}(\vr)\]
with normalizing constants $\alpha_{\vk},\beta_{\vk}>0$ that we allow to change in this proof to simplify the notation and
\begin{equation}
\label{eq:Fk}
 F_{\vk}(\vr):=e^{\frac{(k_1+ik_2)}{2}(-i(1+\omega)x+(\omega-1)y)} \frac{ \theta ( \frac{3(x_1+ix_2)}{4\pi i \omega} + \frac{k_1+ik_2}{\sqrt{3}\omega} ) }{
\theta\Big( \frac{3(x_1+ix_2)}{4\pi i \omega}\Big)}.
\end{equation}
From \eqref{eq:identity}, we see that $\vv_{\vk_*},\vw_{\vk_*} \in \ker (D(\alpha)+i \operatorname{id}_{\CC^2}).$
We observe directly from the function $F_{\mathbf k}$ that $\vw_{\vk_*}$ vanishes to first order at $\mathbf r =0$ and $\mathbf r=\vr_S$, but not at any other point.
The function $\vv_{\vk_*}$ vanishes to second order at $\mathbf r=-\vr_S$, i.e. $\vv_{\vk_*}(-\vr_S)=0$ and $\nabla \vv_{\vk_*}(-\vr_S)=0$.

We conclude that they are $L^2$-orthogonal using \eqref{eq:property}, since we can write for $\gamma \neq 0$
\[\begin{split}\vw_{\vk_*}(\vr') ^=\gamma \wp(x_1'+\vert r_S\vert +ix_2') \vv_{\vk_*}(\vr'),\text{ where }\\
\gamma^:=\frac{\vw_{\vk_*}(-\vr_S)}{\Vert  \wp(\bullet+\vert \vr_S\vert) \vv_{\vk_*}(\bullet) \Vert \lim_{x \to -\vr_S}\wp(x+\vert \vr_S\vert ) \vv_{\vk_*}(x)}
\end{split}\]
and thus $\vw_{x}$ and $\vv_{\vk_*}$ have different rotational symmetries.
We thus have
\[\begin{split}
\tr(\rho_{\vk_*,\vk_*}(\vr))
&= \Vert \vw_{\vk_*}(\vr)\Vert^2 +\Vert  \vv_{\vk_*}(\vr)\Vert^2 \\
&= (\alpha_{\vk_*} \vert F_{\vk_*}(\vr- \vr_S) \wp(x_1+ix_2)\vert^2+\beta_{\vk_*} \vert F_{\vk_*}(\vr) \vert^2)\Vert  \vw_{\mathbf 0}(\vr)\Vert^2.
\end{split} \]
In addition, we record that
\[F_{\vk_*}(\vr-\vr_S) = e^{-\frac{2\pi}{3\sqrt{3}}(1+\omega)} e^{\frac{i}{2}(-i(1+\omega)x_1+(\omega-1)x_2)} \frac{ \theta ( \frac{3(x_1+ix_2)}{4\pi i \omega}+\frac{2i}{\sqrt{3}\omega} ) }{
\theta\Big( \frac{3(x_1+ix_2)}{4\pi i \omega}+\frac{i}{\sqrt{3}\omega}\Big)}\]
and
\[F_{\vk_*}(\vr) = e^{\frac{i}{2}(-i(1+\omega)x_1+(\omega-1)x_2)} \frac{ \theta ( \frac{3(x_1+ix_2)}{4\pi i \omega}+\frac{i}{\sqrt{3}\omega} ) }{
\theta\Big( \frac{3(x_1+ix_2)}{4\pi i \omega}\Big)}.\]
In particular, we have for the modulus of the prefactor of \eqref{eq:Fk} as above \[\vert e^{\frac{i}{2}(-i(1+\omega)x_1+(\omega-1)x_2)} \vert^2 = e^{\frac{x_1-\sqrt{3}x_2}{2}}.\]

Combining the above properties, we find
\[\begin{split}
\tr(\rho_{\vk_*,\vk_*}(\vr))
&=e^{\frac{x_1-\sqrt{3}x_2}{2}} \Bigg(\alpha_{\vk_*}  \frac{ \vert \wp(x_1+ix_2) \theta ( \frac{3(x_1+ix_2)}{4\pi i \omega}+\frac{2i}{\sqrt{3}\omega} )\vert^2 }{
\vert \theta\Big( \frac{3(x_1+ix_2)}{4\pi i \omega}+\frac{i}{\sqrt{3}\omega}\Big)\vert^2} + \beta_{\vk_*} \frac{ \vert \theta ( \frac{3(x_1+ix_2)}{4\pi i \omega}+\frac{i}{\sqrt{3}\omega} )\vert^2 }{
\vert \theta\Big( \frac{3(x_1+ix_2)}{4\pi i \omega}\Big)\vert^2} \Bigg)\Vert  \vw_{\mathbf 0}(\vr)\Vert^2.
\end{split} \]

We also notice that $2i/(\sqrt{3}\omega) = -i/(\sqrt{3}\omega) +(1-\omega)$ which allows us together with the identity
\[\wp(x_1+ix_2) = - e^{-\frac{2\pi}{\sqrt{3}\omega}}\Bigg( \frac{\theta'(0)}{\theta(\frac{i}{\sqrt{3}\omega})} \Bigg)^2 \frac{ \theta ( \tfrac{3(x_1+ix_2)}{4\pi i \omega}-\tfrac{i}{\sqrt{3}\omega} ) \theta ( \tfrac{3(x_1+ix_2)}{4\pi i \omega}+\tfrac{i}{\sqrt{3}\omega} )}{ \theta ( \tfrac{3(x_1+ix_2)}{4\pi i \omega} )^2} \]
to reduce the above expression to
\[\begin{split}
\tr(\rho_{\vk_*,\vk_*}(\vr))
&=e^{-(x_1-\sqrt{3}x_2)}\Bigg(\alpha_{\vk_*} \frac{\vert  \theta ( \tfrac{3(x_1+ix_2)}{4\pi i \omega}-\tfrac{i}{\sqrt{3}\omega} )\vert^4}{\vert \theta\Big( \frac{3(x_1+ix_2)}{4\pi i \omega}\Big)\vert^2}+ \beta_{\vk_*} \vert \theta ( \tfrac{3(x_1+ix_2)}{4\pi i \omega}+\tfrac{i}{\sqrt{3}\omega} )\vert^2 \Bigg)\frac{\Vert  \vw_{\mathbf 0}(\vr)\Vert^2}{\vert \theta\Big( \tfrac{3(x_1+ix_2)}{4\pi i \omega}\Big)\vert^2}.
\end{split} \]
Recalling that $\vw_{0}$ and $\vert \theta \vert$ are even, we have to analyze whether for $\gamma:=\beta_{\vk_*}/\alpha_{\vk_*}>0$
\[\Phi(\vr) = e^{-(x_1-\sqrt{3}x_2)}\Bigg( \vert  \theta ( \tfrac{3(x_1+ix_2)}{4\pi i \omega}-\tfrac{i}{\sqrt{3}\omega} )\vert^4+ \gamma \vert \theta ( \tfrac{3(x_1+ix_2)}{4\pi i \omega}+\tfrac{i}{\sqrt{3}\omega} )  \theta\Big( \tfrac{3(x_1+ix_2)}{4\pi i \omega}\Big)\vert^2  \Bigg) \]
is even.
Evaluating $\Phi$ on the lattice $\Gamma$ we have by \cite[Lemma $4.1$]{MumfordMusili2007}
\[ \vr = (x_1,x_2)^{\top} \in \Gamma \Leftrightarrow \theta\Big(\tfrac{3(x_1+ix_2)}{4\pi i \omega}\Big)=0,\]
we find assuming $\Phi$ to be even that
\[ e^{-(x_1-\sqrt{3}x_2)}  \vert  \theta ( \tfrac{3(x_1+ix_2)}{4\pi i \omega}-\tfrac{i}{\sqrt{3}\omega} )\vert^4 = \Phi(\vr) =\Phi(-\vr) = e^{(x_1-\sqrt{3}x_2)} \vert  \theta ( \tfrac{3(x_1+ix_2)}{4\pi i \omega}+\tfrac{i}{\sqrt{3}\omega} )\vert^4. \]

Choosing, explicitly lattice points $x_1 = 2\pi /\sqrt{3}$ and $x_2 = 2\pi/3$, which is just $-\mathbf v_1$, see \eqref{eq:lattice_vector}, we find
\[ \vert  \theta ( -1-\tfrac{i}{\sqrt{3}\omega} )\vert = \vert  \theta ( -1+\tfrac{i}{\sqrt{3}\omega} )\vert \]
which is false.

We have thus shown that
\[\tr(\rho_{\vk_*,\vk_*}(\vr)) \neq \tr(\rho_{\vk_*,\vk_*}(-\vr)). \]
Either repeating the previous construction for $-\vk_*$ or observing that using \eqref{eq:Lsymm}
\[ \tr(\rho_{-\vk_*,-\vk_*}(\vr)) = \Vert  \mathcal L \vw_{\vk_*}(r)\Vert^2 +\Vert \mathcal L \vv_{\vk_*}(\vr)\Vert^2 = \Vert \vw_{\vk_*}(r)\Vert^2 +\Vert \vv_{\vk_*}(\vr)\Vert^2 = \tr(\rho_{\vk_*,\vk_*}(\vr)), \]
we conclude that \eqref{eq:not02} holds.
\end{proof}

Using the notation of the previous proof, we also have
\begin{prop}
\label{prop:three}
Let $\alpha$ be a two-fold degenerate magic angle of TBG, $\vk_*:=(0,1)^{\top}$, let $\vr \in \{\mathbf 0,\pm \vr_S\}$, and $\vr' \notin \{\mathbf 0,\pm \vr_S\} +\Gamma$. Then the expression
\begin{equation}
\label{eq:D}
\begin{split}
D(\vr, \vr') &:=
\det{
\begin{bmatrix}
\| \vw_{\vk}(\vr)\|^2 - \| \vv_{\vk}(\vr)\|^2  & \braket{\vw_{\vk}(\vr), \vv_{\vk}(\vr)} \\
\| \vw_{\vk}(\vr')\|^2 - \| \vv_{\vk}(\vr')\|^2 & \braket{\vw_{\vk}(\vr'), \vv_{\vk}(\vr')}
\end{bmatrix}} \\
&= (\| \vw_{\vk}(\vr)\|^2 - \| \vv_{\vk}(\vr)\|^2)  \braket{\vw_{\vk}(\vr'), \vv_{\vk}(\vr')} \end{split}
\end{equation}
is non-zero.
\end{prop}

\begin{proof}
From the proof of Proposition \ref{prop:two}, we know that $\vv_{\vk}$ vanishes to second order precisely at $-\vr_S$ and $\vw_{\vk}$ vanishes precisely to first order at $0$ and $\vr_S.$ This already implies the second line in \eqref{eq:D} and that $ (\| \vw_{\vk}(\vr)\|^2 - \| \vv_{\vk}(\vr)\|^2) $ is non-zero. Thus, it remains to show that $\braket{\vw_{\vk}(\vr'), \vv_{\vk}(\vr')} $ is non-zero.
We can write up to a negligible phase $\vw_{\vk}(\vr') = \wp(x_1'+\vert r_S\vert +ix_2') \vv_{\vk}(\vr')$ and thus
\[\braket{\vw_{\vk}(\vr'), \vv_{\vk}(\vr')}  =  \wp(x_1' +ix_2'+\vert r_S\vert) \Vert \vv_{\vk}(\vr')\Vert^2\]
and since $\wp(x_1' +ix_2'+\vert r_S\vert)$ has the same zeros as $\vw_{\vk}$ and $\vv_{\vk}(\vr') \neq 0$, the claim follows.
\end{proof}
Therefore, combining \cref{prop:two,prop:three} with \cref{corr:four-band-gs-real}, we have the following result
\begin{theo}
\label{thm:application-tbg-4}
The ferromagnetic Slater determinants states are the unique ground states of the corresponding flat-band interacting model for TBG-4.
\end{theo}

\subsection{Application of Main Theorem to eTTG-4}
\label{sec:application-ettg-4}

Recall the definition of the $\times$ operation in \cref{eqn:ttgmapping}. If we denote by $\vu_{\mathbf 0}\in L^2_{2,2}$ the flat band eigenfunction associated in the nullspace of $D_{\text{TBG}}(\alpha_0)$, then
\[ \vw_{\mathbf 0} := \vu_{\mathbf 0} \times \vu_{\mathbf 0} \in \ker_{L^2_{1,1}}(D_{\text{eTTG}}(\sqrt 2\alpha_0)).\]
We then have that
\[ \Vert \vw_{\mathbf 0}(\vr)\Vert = \Vert \vw_{\mathbf 0}(-\vr)\Vert.\]
In particular, since $\vu_{\mathbf 0}$ has a simple zero at $\vr=0$, $\vw_{\mathbf 0}$ has a double zero at $\vr=0$.

The two flat bands of eTTG-4 are given by the almost everywhere linearly independent expressions
\[ w_{\vk}(\vr)=F_{\vk}(\vr)w_0(\vr), \quad w_{\vk}(\vr)=F_{\vk/2}(\vr)F_{\vk/2}(\vr)w_0(\vr).\]

Since $\vw_{\mathbf 0}$ has therefore the same properties as $\vw_{\mathbf 0}$ in the proof of Prop. \ref{prop:two}.
One can then replicate the proofs of Prop. \ref{prop:two} and Prop. \ref{prop:three} to find that

\begin{prop}
\label{prop:etwo}
Let $\alpha \in \CC$ be a two-fold degenerate magic angle of eTTG-4, with $\vk_*$ as in Prop. \ref{prop:three}, there are $\vG \in \Gamma^*$ such that
\begin{equation}
\label{eq:not02a}
\Im \tr(A_{\vk_*}(\vG))\neq 0
\end{equation}
while for $\vk \in \Gamma^*$ one has
\begin{equation}
\label{eq:not03a}
\Im \tr(A_{\vk}(\vG))=0 \text{ for all }\vG \in \Gamma^*.
\end{equation}
\end{prop}

\begin{prop}
\label{prop:ethreea}
Let $\alpha$ be a two-fold degenerate magic angle of eTTG-4 and $\vk = \pm e_2$. Moreover, let $\vr \in \{\mathbf 0,\pm \vr_S\}$ and $\vr' \notin \{\mathbf 0,\pm \vr_S\} +\Gamma$, then the expression
\begin{equation}
\label{eq:eDa}
\begin{split}
D(\vr, \vr') &:=
\det{
\begin{bmatrix}
\| \vw_{\vk}(\vr)\|^2 - \| \vv_{\vk}(\vr)\|^2  & \braket{\vw_{\vk}(\vr), \vv_{\vk}(\vr)} \\
\| \vw_{\vk}(\vr')\|^2 - \| \vv_{\vk}(\vr')\|^2 & \braket{\vw_{\vk}(\vr'), \vv_{\vk}(\vr')}
\end{bmatrix}} \\
&= (\| \vw_{\vk}(\vr)\|^2 - \| \vv_{\vk}(\vr)\|^2)  \braket{\vw_{\vk}(\vr'), \vv_{\vk}(\vr')} \end{split}
\end{equation}
is non-zero.
\end{prop}
Therefore, combining \cref{prop:etwo,prop:ethreea} with \cref{corr:four-band-gs-real}, we have the following result
\begin{theo}
\label{thm:application-ettg-4}
The ferromagnetic Slater determinants states are the unique ground states of the corresponding flat-band interacting model for eTTG-4.
\end{theo}
\bibliographystyle{plain}
\bibliography{bibliography}

\appendix

\section{Reformulation of Fock Energy}
\label{sec:reform-fock-energy}
Using the notations introduced in \cref{eq:rotated-form-factor} and using the cyclic property of trace, we have that
\begin{equation}
  \begin{split}
    &\tr( \Lambda_{\vk}(\vq') Q(\vk') \Lambda_{\vk}(\vq')^\dagger Q(\vk)) \\
    & = \tr\Big(
      \begin{bmatrix}
        B^{(1)}_{\vk}(\vq') & \\
                                 &  B^{(2)}_{\vk}(\vq')
      \end{bmatrix}
      \begin{bmatrix}
        c(\vk') & s(\vk') \\
        s(\vk') & -c(\vk')
      \end{bmatrix} \\
    & \hspace{5em}
      \begin{bmatrix}
        B^{(1)}_{\vk}(\vq')^{\dagger} & \\
                                           &  B^{(2)}_{\vk}(\vq')^{\dagger}
      \end{bmatrix}
      \begin{bmatrix}
        c(\vk) & s(\vk) \\
        s(\vk) & -c(\vk)
      \end{bmatrix}
      \Big).
  \end{split}
\end{equation}
We will now split the sine and cosine matrices into their diagonal and off-diagonal parts
\begin{equation}
  \begin{bmatrix}
    c(\vk) & s(\vk) \\
    s(\vk) & -c(\vk)
  \end{bmatrix}
  =
  \begin{bmatrix}
    c(\vk) & \\
           & -c(\vk)
  \end{bmatrix}
  +
  \begin{bmatrix}
    & s(\vk) \\
    s(\vk) &
  \end{bmatrix}
\end{equation}
and similarly for the matrix involving $c(\vk')$ and $s(\vk')$.

Since the form factor is block diagonal, one can easily verify that after splitting sine and cosine matrix as above, the only non-vanishing contributions to the trace are
\begin{equation}
  \begin{split}
    & \tr\Big(
      \begin{bmatrix}
        B^{(1)}_{\vk}(\vq') & \\
                                 &  B^{(2)}_{\vk}(\vq')
      \end{bmatrix}
      \begin{bmatrix}
        c(\vk') & \\
                & -c(\vk')
      \end{bmatrix}
      \begin{bmatrix}
        B^{(1)}_{\vk}(\vq')^{\dagger} & \\
                                           &  B^{(2)}_{\vk}(\vq')^{\dagger}
      \end{bmatrix}
      \begin{bmatrix}
        c(\vk) & \\
               & -c(\vk)
      \end{bmatrix}
      \Big) \\[2ex]
    & \tr\Big(
      \begin{bmatrix}
        B^{(1)}_{\vk}(\vq') & \\
                                 &  B^{(2)}_{\vk}(\vq')
      \end{bmatrix}
      \begin{bmatrix}
        & s(\vk') \\
        s(\vk') &
      \end{bmatrix}
      \begin{bmatrix}
        B^{(1)}_{\vk}(\vq')^{\dagger} & \\
                                           &  B^{(2)}_{\vk}(\vq')^{\dagger}
      \end{bmatrix}
      \begin{bmatrix}
        & s(\vk) \\
        s(\vk) &
      \end{bmatrix}
      \Big).
  \end{split}
\end{equation}
Therefore,
\begin{equation}
  \begin{split}
    &\tr( \Lambda_{\vk}(\vq') Q(\vk') \Lambda_{\vk}(\vq')^\dagger Q(\vk)) \\
    & = \tr\Big(\Big[ B^{(1)}_{\vk}(\vq') c(\vk') B^{(1)}_{\vk}(\vq')^{\dagger}  c(\vk) + B^{(2)}_{\vk}(\vq') c(\vk') B^{(2)}_{\vk}(\vq')^{\dagger}  c(\vk) \Big] \\
    & \hspace{3em} + \Big[ B^{(1)}_{\vk}(\vq') s(\vk') B^{(2)}_{\vk}(\vq')^{\dagger} s(\vk) +  B^{(2)}_{\vk}(\vq') s(\vk') B^{(1)}_{\vk}(\vq')^{\dagger} s(\vk) \Big] \Big).
  \end{split}
\end{equation}
The last two terms can be slightly simplified by observing that
\begin{equation}
  \Big(B^{(1)}_{\vk}(\vq') s(\vk') B^{(2)}_{\vk}(\vq')^{\dagger} s(\vk)\Big)^{\dagger}  = s(\vk) B^{(2)}_{\vk}(\vq') s(\vk') B^{(1)}_{\vk}(\vq')^{\dagger}.
\end{equation}
Therefore, using the cyclic property of trace
\begin{equation}
  \overline{\tr{\Big( B^{(1)}_{\vk}(\vq') s(\vk') B^{(2)}_{\vk}(\vq')^{\dagger} s(\vk) \Big)}} = \tr{\Big(B^{(2)}_{\vk}(\vq') s(\vk') B^{(1)}_{\vk}(\vq')^{\dagger} s(\vk) \Big)}.
\end{equation}
From these calculations Fock energy can be written as
\begin{equation}
  \begin{split}
    K[P] = -\frac{1}{4 | \Omega | N_{\vk}} \sum_{\vk,\vq} \sum_{\vG} V(\vq') \Big[
    & \tr\Big(B^{(1)}_{\vk}(\vq') c(\vk') B^{(1)}_{\vk}(\vq')^\dagger c(\vk) \Big) \\[-2ex]
    & + \tr\Big(B^{(2)}_{\vk}(\vq') c(\vk') B^{(2)}_{\vk}(\vq')^\dagger c(\vk) \Big) \\
    & + 2 \Re \Big(\tr\Big(B^{(1)}_{\vk}(\vq') s(\vk') B^{(2)}_{\vk}(\vq')^\dagger s(\vk) \Big) \Big)\Big].
  \end{split}
\end{equation}
We fix $\vk, \vq, \vG$ and define:
\begin{equation}
  \begin{split}
    & c_i := [c(\vk)]_{ii} \qquad\qquad\qquad c_i' := [c(\vk')]_{ii} \\
    & s_{i} := [s(\vk)]_{ii} \qquad\qquad\qquad s_i' := [s(\vk')]_{ii} \\
    & b^{(1)}_{ij} := [B^{(1)}_{\vk}(\vq')]_{ij} \\
    & b^{(2)}_{ij} := [B^{(2)}_{\vk}(\vq')]_{ij}.
  \end{split}
\end{equation}
With these definition, the Fock energy becomes
\begin{equation}
  \begin{split}
    \sum_{ij} & b^{(1)}_{ij} c_{j}' \overline{b^{(1)}_{ij}} c_{i} + \sum_{ij} b^{(2)}_{ij} c_{j}' \overline{b^{(2)}_{ij}} c_{i} + 2 \Re\Big[\sum_{ij} b^{(1)}_{ij} s_{j}' \overline{b^{(2)}_{ij}} s_{i} \Big] \\
    & =
    \sum_{ij}  (|b^{(1)}_{ij}|^{2} + |b^{(2)}_{ij}|^{2}) c_{j}' c_{i} + 2 \Re[ b^{(1)}_{ij} \overline{b^{(2)}_{ij}}] s_{j}' s_{i}.
  \end{split}
\end{equation}
Recalling the trigonometric product-to-sum rules we have
\begin{equation}
  \begin{split}
    c_i c_j' = \cos{(\theta_i(\vk))} \cos{(\theta_j(\vk'))} = \frac{1}{2} \bigg( \cos{(\theta_i(\vk) - \theta_j(\vk'))} + \cos{(\theta_i(\vk) + \theta_j(\vk'))} \bigg) \\[1ex]
    s_i s_j' = \sin{(\theta_i(\vk))} \sin{(\theta_j(\vk'))} = \frac{1}{2} \bigg( \cos{(\theta_i(\vk) - \theta_j(\vk'))} - \cos{(\theta_i(\vk) + \theta_j(\vk'))} \bigg).
  \end{split}
\end{equation}
Therefore, for fixed $i,j$ we have
\begin{equation}
  \begin{split}
    (|b^{(1)}_{ij}|^{2} + |b^{(2)}_{ij}|^{2}) c_{i} c_{j}' + 2 \Re[& b^{(1)}_{ij} \overline{b^{(2)}_{ij}}] s_{i} s_{j}' \\
    = \frac{1}{2} (|b^{(1)}_{ij}|^{2} + |b^{(2)}_{ij}|^{2}) & \bigg( \cos{(\theta_i(\vk) - \theta_j(\vk'))} + \cos{(\theta_i(\vk) + \theta_j(\vk'))} \bigg) \\
    +  \Re[ b^{(1)}_{ij} \overline{b^{(2)}_{ij}}]& \bigg( \cos{(\theta_i(\vk) - \theta_j(\vk'))} - \cos{(\theta_i(\vk) - \theta_j(\vk'))} \bigg).
  \end{split}
\end{equation}
But since
\begin{equation}
  |b^{(1)}_{ij}|^{2} + |b^{(2)}_{ij}|^{2} \pm 2 \Re[ b^{(1)}_{ij} \overline{b^{(2)}_{ij}}] = | b^{(1)}_{ij} \pm b^{(2)}_{ij} |^2
\end{equation}
this simplifies to
\begin{equation}
  \frac{1}{2} | b^{(1)}_{ij} + b^{(2)}_{ij} |^2 \cos{(\theta_i(\vk) - \theta_j(\vk'))} + \frac{1}{2} | b^{(1)}_{ij} - b^{(2)}_{ij} |^2  \cos{(\theta_i(\vk) + \theta_j(\vk'))}.
\end{equation}
Combining these calculations together we finally find that
\begin{equation}
  \begin{split}
    K[P] = \frac{1}{8 | \Omega | N_{\vk}} \sum_{\vk,\vq} \sum_{\vG} V(\vq') \sum_{ij} \bigg\{ & | b^{(1)}_{ij} + b^{(2)}_{ij} |^2 \cos{(\theta_i(\vk) - \theta_j(\vk'))} \\
    & + | b^{(1)}_{ij} - b^{(2)}_{ij} |^2  \cos{(\theta_i(\vk) + \theta_j(\vk'))} \bigg\}.
  \end{split}
\end{equation}

\section{Proof of~\cref{lem:full-rank}}
\label{sec:proof-full-rank}
We set $\vq = \vk' - \vk$ and prove that $\Lambda_{\vk}(\vq)$ is full rank which is equivalent to proving $A_{\vk}(\vq)$ is full rank.
We start by recalling the definition of the form factor for $\vG = \vzero$
\begin{equation}
  [\Lambda_{\vk}(\vq)]_{mn} = \frac{1}{|\Omega|} \sum_{\vG'} \sum_{\sigma,j} \overline{\hat{u}_{m\vk}(\vG'; \sigma, j)} \hat{u}_{n(\vk+\vq)}(\vG'; \sigma, j)
\end{equation}
where $m,n \in \mc{N}$.
Notice that this is just the inner product between the functions $(\vG, \sigma, j) \mapsto u_{m\vk}(\vG, \sigma, j)$ and $(\vG, \sigma, j) \mapsto u_{n(\vk+\vq)}(\vG, \sigma, j)$ on the space $L^{2}(\Gamma^{*} \times \CC^{2} \times \CC^{N})$.
Therefore, if we pick an orthogonal basis $\{ \ket{n} : n \in \{ 1, \cdots, 2 M \}\}$ for $\CC^{\# | \mc{N} |}$ we can define the operator $\Phi(\vk)$ using bra-ket notation:
\begin{equation}
  \Phi(\vk) := \frac{1}{|\Omega|^{1/2}} \sum_{n \in \mc{N}} \ket{\hat{u}_{n\vk}} \bra{n}
\end{equation}
and the form factor at $\vG = 0$ becomes
\begin{equation}
  \label{eq:form-factor-full-rank-1}
  \Lambda_{\vk}(\vq) = \Phi(\vk)^{\dagger} \Phi(\vk + \vq).
\end{equation}
Note that $\Phi(\vk)$ is a partial isometry since by the orthogonality of $\hat{u}_{n\vk}$ we know that $\Phi(\vk)^{\dagger} \Phi(\vk) = \sum_{n} \ket{n} \bra{n} = I$.

Our goal is to show that there exists a $\vq$ so that $\Lambda_{\vk}(\vq)$ is full rank or equivalently for all $v \in \CC^{2M}$ such that $\| v \| = 1$, we must show that $\| \Phi(\vk + \vq)^{\dagger} \Phi(\vk) v\| > 0$.

Now observe that
\begin{equation}
  \begin{split}
    &\braket{v, \Phi(\vk)^{\dagger} \Phi(\vk + \vq) \Phi(\vk + \vq)^{\dagger} \Phi(\vk) v} \\
    & = \braket{v, \Phi(\vk)^{\dagger} \Phi(\vk + \vq) \Phi(\vk + \vq)^{\dagger} \Phi(\vk) v} - \braket{v,  v} + \braket{v, v} \\
    & = \braket{v, \Phi(\vk)^{\dagger} \Phi(\vk + \vq) \Phi(\vk + \vq)^{\dagger} \Phi(\vk) v} - \braket{v, \Phi(\vk)^{\dagger} \Phi(\vk) \Phi(\vk)^{\dagger} \Phi(\vk) v} + 1 \\
    & = \braket{v, \Phi(\vk)^{\dagger} \Big(\Phi(\vk + \vq) \Phi(\vk + \vq)^{\dagger} - \Phi(\vk) \Phi(\vk)^{\dagger}\Big) \Phi(\vk) v} + 1 \\
  \end{split}
\end{equation}
where in the second line we have used that $\Phi(\vk)$ is a partial isometry.
Since
\begin{equation}
\begin{split}
  \sup_{\| v \| = 1} & |\braket{v, \Phi(\vk)^{\dagger} \Big(\Phi(\vk + \vq) \Phi(\vk + \vq)^{\dagger} - \Phi(\vk) \Phi(\vk)^{\dagger}\Big) \Phi(\vk) v} | \\
  & \leq \| \Phi(\vk + \vq) \Phi(\vk + \vq)^{\dagger} - \Phi(\vk) \Phi(\vk)^{\dagger}\|
\end{split}
\end{equation}
it suffices to show that the norm $\| \Phi(\vk + \vq) \Phi(\vk + \vq)^{\dagger} - \Phi(\vk) \Phi(\vk)^{\dagger} \|$ is bounded away from 1.
But observe that $\Phi(\vk) \Phi(\vk)^{\dagger}$ and $\Phi(\vk + \vq) \Phi(\vk + \vq)^{\dagger}$ are spectral projector onto the flat bands at momentum $\vk$ and $\vk + \vq$.
Therefore by assumption $\| \Phi(\vk + \vq) \Phi(\vk + \vq)^{\dagger} - \Phi(\vk) \Phi(\vk)^{\dagger} \| < 1$.

\section{Real Space Conditions for~\cref{thm:hf-gs-unique}}
\label{sec:real-space-proof}

\begin{lemm}
For all orthogonal projectors $\Pi$, the following statements are equivalent:
\begin{enumerate}
\item There exists $\vG'$, so that $\| (I - \Pi) A_{\vk}(\vG') \Pi \| > 0$.
\item There exists $\vr$, so that $\| (I - \Pi) \rho_{\vk,\vk}(\vr) \Pi \| > 0$.
\end{enumerate}
\end{lemm}
This lemma connects the second condition of~\cref{thm:hf-gs-unique} to a property of the real space functions $u_{n\vk}(\vr)$ which is easier to check in practice.
For the special case that the system has four bands, similar to~\cref{corr:four-band-gs}, can derive the following simpler condition:
\begin{lemm}
    Suppose that the system has four flat bands then the following are equivalent for all $\vk$
    \begin{enumerate}
        \item For all non-trivial orthogonal projectors, $\Pi$ exists an $\vr$ so that
        \[
        \| (I - \Pi) \rho_{\vk,\vk}(\vr) \Pi \| > 0
        \]
      \item There exist $\vr, \vr'$ so that $[ \rho_{\vk,\vk}(\vr), \rho_{\vk,\vk}(\vr') ] \neq 0$.
    \end{enumerate}
    Furthermore, to prove $[ \rho_{\vk,\vk}(\vr), \rho_{\vk,\vk}(\vr') ] \neq 0$ it suffices to show there exist $\vr, \vr'$ so that
    \begin{equation}
      \det{
        \begin{bmatrix}
          \| u_{1\vk}(\vr)\|^2 - \| u_{2\vk}(\vr)\|^2  & \braket{u_{1\vk}(\vr), u_{2\vk}(\vr)} \\
          \| u_{1\vk}(\vr')\|^2 - \| u_{2\vk}(\vr')\|^2 & \braket{u_{1\vk}(\vr'), u_{2\vk}(\vr')}
        \end{bmatrix}} \neq 0
    \end{equation}
  \end{lemm}
\begin{proof}
One may easily verify that $\rho_{\vk,\vk}(\vr)$ is a $2 \times 2$ Hermitian matrix for all $\vr$.
As a consequence, since $\Pi$ is a non-trivial projection, $\| (I - \Pi) \rho_{\vk,\vk}(\vr) \Pi \| = 0$ if and only if $\Pi = \ket{v}\bra{v}$ and $(I - \Pi) = \ket{v^\perp}\bra{v^\perp}$ where $\ket{v}$ and $\ket{v^\perp}$ are a complete basis of eigenvectors for $\rho_{\vk,\vk}(\vr)$.
Therefore, $\| (I - \Pi) \rho_{\vk,\vk}(\vr) \Pi \| = 0$ for all $\vr$ if and only if the set $\{ \rho_{\vk,\vk}(\vr) : \vr \in \Omega \}$ are mutually commuting.
Hence, $\| (I - \Pi) \rho_{\vk,\vk}(\vr) \Pi \| > 0$ for all non-trivial $\Pi$ if and only if there exist $\rho_{\vk,\vk}(\vr)$ and $\rho_{\vk,\vk}(\vr')$ which do not commute.

For the second part of the lemma, we begin by calculating the off-diagonal entries of $\rho_{\vk,\vk}(\vr) \rho_{\vk,\vk}(\vr')$:
\begin{equation}
\begin{split}
    & \rho_{\vk,\vk}(\vr) \rho_{\vk,\vk}(\vr') \\
    & =
    \begin{bmatrix}
        \| u_{1\vk}(\vr) \|^2 & \braket{u_{1\vk}(\vr), u_{2\vk}(\vr)} \\
        \braket{u_{1\vk}(\vr), u_{2\vk}(\vr)}^* & \| u_{2\vk}(\vr) \|^2
    \end{bmatrix}
       \begin{bmatrix}
        \| u_{1\vk}(\vr') \|^2 & \braket{u_{1\vk}(\vr'), u_{2\vk}(\vr')} \\
        \braket{u_{1\vk}(\vr'), u_{2\vk}(\vr')}^* & \| u_{2\vk}(\vr') \|^2
    \end{bmatrix} \\
    & =
    \begin{bmatrix}
        * & d(\vr,\vr')   \\
        \overline{d(\vr',\vr)} & *
    \end{bmatrix}
\end{split}
\end{equation}
where
\begin{equation}
   d(\vr,\vr') = \| u_{1\vk}(\vr) \|^2 \braket{u_{1\vk}(\vr'), u_{2\vk}(\vr')} + \| u_{2\vk}(\vr') \|^2 \braket{u_{1\vk}(\vr), u_{2\vk}(\vr)}
\end{equation}
Therefore, using that $(\rho_{\vk,\vk}(\vr) \rho_{\vk,\vk}(\vr'))^\dagger = \rho_{\vk,\vk}(\vr') \rho_{\vk,\vk}(\vr)$ we have
\begin{equation}
[ \rho_{\vk,\vk}(\vr), \rho_{\vk,\vk}(\vr') ] =
\begin{bmatrix}
      *  & d(\vr,\vr') - d(\vr',\vr) \\
    \overline{d(\vr',\vr)} - \overline{d(\vr,\vr')} &  *
\end{bmatrix}.
\end{equation}
Now we calculate
\begin{equation}
\begin{split}
d(\vr,\vr') - d(\vr',\vr)
 = & \Big( \| u_{1\vk}(\vr) \|^2 - \| u_{2\vk}(\vr) \|^2 \Big) \braket{u_{1\vk}(\vr'), u_{2\vk}(\vr')} \\
& - \Big( \| u_{1\vk}(\vr') \|^2 - \| u_{2\vk}(\vr') \|^2 \Big) \braket{u_{1\vk}(\vr), u_{2\vk}(\vr)} \\
= & \det{
    \begin{bmatrix}
    \| u_{1\vk}(\vr)\|^2 - \| u_{2\vk}(\vr)\|^2  & \braket{u_{1\vk}(\vr), u_{2\vk}(\vr)} \\
    \| u_{1\vk}(\vr')\|^2 - \| u_{2\vk}(\vr')\|^2 & \braket{u_{1\vk}(\vr'), u_{2\vk}(\vr')}
    \end{bmatrix}}
\end{split}
\end{equation}
which completes the proof.
\end{proof}

As for the first condition of~\cref{thm:hf-gs-unique}, we observe that
\begin{equation}
  \tr{(A_{\vk}(\vG))} = \int_{\Omega} e^{-i \vG \cdot \vr} \sum_{m} \| u_{m\vk}(\vr) \|^{2} \ud\vr
\end{equation}
where the norm is understood to sum over $\sigma, j$.
And we recall the following simple result:
\begin{lemm}
\label{lemm:symmetry}
    Let $\Gamma$ be a d-dimensional lattice in $\mathbb R^d$ and $f:\RR^d/\Gamma \to \RR$ a continuous function, then $f$ is even, i.e. $\overline{f(-\mathbf r)}=f(\mathbf r)$ if and only if
    \[ \Im\hat{f}(\vG) =0 \text{ for all }\mathbf G \in \Gamma^*,\]
    where $\hat{f}$ is the Fourier transform of $f.$
\end{lemm}
\begin{proof}
    Indeed, assuming $\overline{f(-\mathbf r)}=f(\mathbf r)$, we have using \eqref{eq:Fourier}
    \[\begin{split}
    \overline{\hat{f}(\vG)} &= \overline{\int_{\RR^d/\Gamma} f(\vr) e^{-i\vG \cdot \vr} \ d\vr }=  \int_{\RR^d/\Gamma} \overline{f(\vr)} e^{i\vG \cdot \vr} \ d\vr \\
    &=\int_{\RR^d/\Gamma} \overline{f(-\vr)} e^{-i\vG \cdot \vr} \ d\vr = \int_{\RR^d/\Gamma} f(\vr) e^{-i\vG \cdot \vr} \ d\vr = \hat{f}(\vG).
    \end{split} \]
    The converse implication is easily observed from the Fourier series.
\end{proof}
By~\cref{lemm:symmetry}, we conclude that $\Im\tr{(A_{\vk}(\vG))} \neq 0$ if and only if the function
\begin{equation}
    \vr \mapsto \sum_{m} \| u_{m\vk}(\vr) \|^{2}
\end{equation}
is not an even function.

\end{document}